\documentclass[11pt]{article}
\usepackage{graphicx}
\usepackage{float}
\usepackage{subfigure}
\usepackage{caption}
\usepackage{multirow}
\usepackage{makecell}
\usepackage{appendix}
\usepackage[figuresright]{rotating}
\usepackage{booktabs}
\usepackage[linesnumbered,ruled]{algorithm2e}
\usepackage[margin=1in]{geometry}
\usepackage{hyperref}
\usepackage{amsfonts}
\usepackage{mathrsfs}
\usepackage{comment}
\usepackage{amsmath}
\usepackage{amssymb}
\usepackage{amsthm}
\usepackage{amscd}
\usepackage{graphicx}
\usepackage{indentfirst}
\usepackage[all]{xy}
\usepackage{titlesec}
\usepackage{enumerate}
\usepackage{bm}
\usepackage{enumitem}
\usepackage{color}
\usepackage{dsfont}
\usepackage{arydshln}
\usepackage{booktabs}
\newtheorem{theorem}{Theorem}[section]
\newtheorem{lemma}[theorem]{Lemma}
\newtheorem{proposition}[theorem]{Proposition}
\newtheorem{corollary}[theorem]{Corollary}

\newtheorem{definition}[theorem]{Definition}

\newcommand{\x}{\vec}

\begin{document}
\title{Deep holes of a class of twisted Reed-Solomon codes\footnote{The research was supported by the National Natural Science Foundation of China under the Grants 12222113 and 12441105.}}
\author{Haojie Gu\footnote{Haojie Gu is with the School of Mathematical Sciences, Capital Normal University, Beijing 100048, China. Email: 2200502051@cnu.edu.cn.},
    \and Nan Wang\footnote{Nan Wang is with the School of Mathematical Sciences, Capital Normal University, Beijing 100048, China. Email: },
	\and Jun Zhang\footnote{Jun Zhang is with the School of Mathematical Sciences, Capital Normal University, Beijing 100048, China. Email: junz@cnu.edu.cn.}
}

\date{}
\maketitle

\begin{abstract}
	
  The deep hole problem is a fundamental problem in coding theory, and it has many important applications in code constructions and cryptography. The deep hole problem of Reed-Solomon codes has gained a lot of attention. As a generalization of Reed-Solomon codes, we investigate the problem of deep holes of a class of twisted Reed-Solomon codes in this paper. 
  Firstly, we provide the necessary and sufficient conditions for $\boldsymbol{a}=(a_{0},a_{1},\cdots,a_{n-k-1})\in\mathbb{F}_{q}^{n-k}$ to be the syndrome of some deep hole of $TRS_{k}(\mathcal{A},l,\eta)$. Next, we consider the problem of determining all deep holes of the twisted Reed-Solomon codes $TRS_{k}(\mathbb{F}_{q}^{*},k-1,\eta)$. Specifically, we prove that there are no other deep holes of $TRS_{k}(\mathbb{F}_{q}^{*},k-1,\eta)$ for $\frac{3q+2\sqrt{q}-8}{4}\leq k\leq q-5$ when q is even, and  $\frac{3q+3\sqrt{q}-5}{4}\leq k\leq q-5$ when q is odd. We also completely determine their deep holes for $q-4\leq k\leq q-2$ when $q$ is even.

	\begin{flushleft}
	\end{flushleft}
\end{abstract}

\section{Introduction}

Let $\mathbb{F}_{q}$ be a finite field with size $q$ and characteristic $p$. Let $\mathbb{F}_{q}^n$ be the $n$-dimensional vector space over the finite field $\mathbb{F}_{q}$.
  For any vector $ \boldsymbol{x}=(x_1,x_2,\cdots,x_n)\in \mathbb{F}_{q}^n$, the \emph{Hamming weight} $wt( \boldsymbol{x})$ of $ \boldsymbol{x}$ is defined to be the number of non-zero coordinates, i.e.,
$$wt( \boldsymbol{x})=|\left\{i\,|\,1\leqslant i\leqslant n,\,x_i\neq 0\right\}|.$$

An $[n,k,d]$-linear code $\mathcal{C}\subseteq \mathbb{F}_{q}^n$ is a $k$-dimensional linear subspace of $\mathbb{F}_{q}^n$ with minimum distance $d=d(\mathcal{C})$ defined as $$d(\mathcal{C})=\min\left\{wt(\boldsymbol{c}):\boldsymbol{c}\in\mathcal{C}\backslash\{0\}\right\}.$$ 
For any vector $\boldsymbol{u}\in\mathbb{F}_{q}^n$, the error distance from $\boldsymbol{u}$ to $\mathcal{C}$ is defined as:
$$d(\boldsymbol{u},\mathcal{C})=\min\{d(\boldsymbol{u},\boldsymbol{v})\,|\,\boldsymbol{v}\in C\},$$
where $d(\boldsymbol{u},\boldsymbol{v})=|\{i\,|\,u_{i}\neq v_{i},\,1\le i\le n\}|$
is the Hamming distance between vectors $\boldsymbol{u}$ and $\boldsymbol{v}$.
The error distance plays a crucial role in the decoding of the code. The maximum error distance
$$\rho(C)=\max\{d(\boldsymbol{u},\, \mathcal{C})\,|\,\boldsymbol{u}\in \mathbb{F}_{q}^n\}$$
is called the \emph{covering radius} of $\mathcal{C}$. Vectors that achieve this maximum error distance are referred to as deep holes of the code. The computation of the covering radius is a fundamental problem in coding theory. However, McLoughlin~\cite{mcloughlin1984complexity} has proven that the computational difficulty of determining the covering radius of random linear codes strictly exceeds NP-completeness.

In recent years, the problem of determining the deep holes of Reed-Solomon codes has attracted lots of attention in the literature \cite{cafure2011singularities,cheng2007deciding,keti2016deep,kaipa2017deep,li2008subset,li2010new,li2008error,wu2012deep,zhang2012new,zhang2023deep,zhang2019deep,zhu2012computing,zhuang2015determining}.

\begin{definition}
	Let $\mathcal{A}=\left\{\alpha_{1},\cdots,\alpha_{n}\right\}\subseteq \mathbb{F}_{q}$ be the evaluation set, then the Reed-Solomon code $RS_{k}(\mathcal{A})$ of length $n$ and dimension $k$ is defined as \begin{equation*}
	RS_{k}(\mathcal{A})=\left\{(f(\alpha_{1}),\cdots,f(\alpha_{n})):f(x)\in \mathbb{F}_{q}[x],\deg(f)\leq k-1\right\}.
	\end{equation*}
\end{definition}

It can be demonstrated that the covering radius of $RS_k(\mathcal{A})$ is equal to $n-k$. It was shown in~\cite{guruswami2005maximum} that the problem of determining whether a vector is a deep hole of a given Reed-Solomon code is NP-hard. Furthermore, it has been verified that vectors whose generating polynomials have degree $k$ are indeed deep holes of $RS_{k}(\mathcal{A})$~\cite{cheng2007deciding}. There may be some other deep holes of $RS_k(\mathcal{A})$ for certain subset $\mathcal{A}$ of $\mathbb{F}_q$. 

 For the standard Reed-Solomon code $RS_k(\mathbb{F}_{q},\mathbb{F}_{q}^{*})$, based on numerical
 computations, Cheng and Murray~\cite{cheng2007deciding} conjectured that vectors defined by polynomials of degree $k$ are the only deep holes possible. As a theoretical evidence,
 they proved that their conjecture is true for words $\boldsymbol{u}_{f}$ defined by polynomial
 $f$ if $d=\deg(\boldsymbol{u}_{f})-k$ is small and $q$ is sufficiently large compared to $d+k$. Li and Wan~\cite{li2008subset}(resp. Zhang et al.~\cite{zhang2012new}) proved that vectors with generating polynomials of degree $k+1$ (resp. $k+2$) are not deep holes of $RS_k(\mathbb{F}_q)$. For
 those words defined by polynomials in $\mathbb{F}_{q}[x]$ of low degrees, Li and Wan~\cite{li2010new} applied the method of Cheng and Wan~\cite{cheng2007deciding}
  to study the error distance $d(\boldsymbol{u},\mathcal{C})$ for the standard Reed-Solomon code. Liao~\cite{liao2011reed} extended the results in~\cite{li2010new} to those
 words defined by polynomials in $\mathbb{F}_{q}[x]$ of high degrees. By means of
 a deeper study of the geometry of hypersurfaces, Cafure and et al.~\cite{cafure2011singularities} made
 some improvement of the results in \cite{li2010new}. When $2\leq k\leq p-2$ or $2\leq q-p\leq k\leq q-3$, Zhuang et al. \cite{zhuang2015determining} proved that the conjecture of Cheng and Murray is ture. In particular, the conjecture of Cheng and Murray  holds for prime fields. Applying Seroussi and Roth’s results on the extension of RS codes \cite{seroussi1986mds}, Kaipa \cite{kaipa2017deep} proved that the conjecture of Cheng and Murray  holds for $k\geq \left\lfloor\frac{q-1}{2}\right\rfloor$.

Since Beelen et al. first introduced twisted Reed-Solomon (TRS) codes in \cite{beelen2017twisted,beelen2022twisted}, many coding scholars have studied TRS codes with good properties, including TRS MDS codes, TRS self-dual codes, and TRS LCD codes \cite{gu2023twisted,huang2021mds,liu2021construction,sui2022mds,zhang2022class}. It is also difficult to determine the covering radius and deep holes of twisted RS codes. Fang et al.~\cite{fang2024deep} consider the problem of determining all deep holes of the full-length twisted Reed-Solomon codes $TRS_k(\mathbb{F}_{q},\theta)$. Specifically, they prove that there are no other deep holes of $TRS_k(\mathbb{F}_q, \theta)$ for $\frac{3q-8}{4}\leq k\leq q-4$ when $q$ is even, and $\frac{3q+3\sqrt{q}-7}{4}\leq k\leq q-4$ when $q$ is odd. They also completely determine the deep holes for $q-3\leq k\leq q-1$. In this paper, we consider a class of TRS codes which is more general than the class studied by Fang et al.~\cite{fang2024deep}.

The rest of this paper is organized as follows. In Section II, we present some results on twisted Reed-Solomon codes and character sums. In Section III, we determine the covering radius and some deep holes of twisted RS codes $TRS_{k}(\mathcal{A},l,\eta)$ for a general evaluation set $\mathcal{A}\subseteq \mathbb{F}_{q}$. In Section IV, we present the results on the completeness of deep holes of $TRS_k(\mathbb{F}_{q}^{*}, k-1,\eta)$.

\section{Preliminary}

\subsection{Twisted Reed-Solomon codes $TRS_{k}(\mathcal{A},l,\eta)$}

In this paper, we will study the covering radius problem and deep hole problem of the following class of TRS codes.
\begin{definition}
	Let integers $l, k, n$ be such that $0\leq l\leq k-1\leq n-2$. For any $\eta\in \mathbb{F}_{q}^{*}$ denote by
	\begin{equation}
	\begin{aligned}
	 S_{k, l, \eta}=\left\{\sum_{i \in \{0,1,...,k-1\} \backslash \{l\}}f_i x^i+f_l\left(x^l+\eta x^k\right) \mid f_0, \cdots, f_{k-1} \in \mathbb{F}_{q}\right\}. \notag
	\end{aligned}
	\end{equation}
   For any $\mathcal{A}=\{{\alpha_1,\cdots,\alpha_n}\}\subset \mathbb{F}_{q}$ the linear code
	$$TRS_{k}(\mathcal{A},l,\eta)=\left\{(f(\alpha_1),\cdots,f(\alpha_n)):f\in S_{k, l, \eta}\right\}$$
	is called a twisted Reed-Solomon (TRS) code.
\end{definition}

Obviously, the TRS code $TRS_{k}(\mathcal{A},l,\eta)$ has
a generator matrix 

\begin{equation}
G=\small{\begin{pmatrix}
	1&1&\cdots&1\\
	\alpha_{1}&\alpha_{2}&\cdots&\alpha_{n}\\
	\vdots&\vdots&\vdots&\vdots\\
	\alpha_{1}^{l-1}&\alpha_{2}^{l-1}&\cdots&\alpha_{n}^{l-1}\\
	\alpha_{1}^{l+1}&\alpha_{2}^{l+1}&\cdots&\alpha_{n}^{l+1}\\
	\vdots&\vdots&\vdots&\vdots\\
	\alpha_{1}^{k-1}&\alpha_{2}^{k-1}&\cdots&\alpha_{n}^{k-1}\\
\alpha_{1}^{l}+\eta\alpha_{1}^{k}&\alpha_{2}^{l}+\eta\alpha_{2}^{k}&\cdots&\alpha_{n}^{l}+\eta\alpha_{n}^{k}\\
	\end{pmatrix}}.
\end{equation}

The finite geometry method of syndromes is an important method in determining deep holes of Reed-Solomon codes and related codes.
In order to computing a parity check matrix of the TRS code $TRS_{k}(\mathcal{A},l,\eta)$, we recall a useful result from~\cite{sui2022mds}.

\begin{lemma}{\cite{sui2022mds}}\label{lemcont:4.1}
	Let $\alpha_{1},\cdots,\alpha_{n}$ be distinct elements of $\mathbb{F}_{q}$ and $\prod\limits_{i=1}^{n}(x-\alpha_{i})=\sum\limits_{j=0}^{n}\sigma_{j}x^{n-j}$. Let $\Lambda_{0}=1$ and $\boldsymbol{y}=(\Lambda_{0},\Lambda_{1},\cdots,\Lambda_{n})$ be the unique solution of the following system of equations:
	\begin{equation*}
	\begin{pmatrix}
	\sigma_{0}&0&0&\cdots&0\\
	\sigma_{1}&\sigma_{0}&0&\cdots&0\\
	\sigma_{2}&\sigma_{1}&\sigma_{0}&\cdots&0\\
	\vdots&\vdots&\vdots&\ddots&\vdots\\
	\sigma_{n}&\sigma_{n-1}&\sigma_{n-2}&\cdots&\sigma_{0}
	\end{pmatrix}
	\begin{pmatrix}
	\Lambda_{0}\\\Lambda_{1}\\ \vdots\\ \Lambda_{n}
	\end{pmatrix}=
	\begin{pmatrix}
	1\\0\\ \vdots\\ 0
	\end{pmatrix}.
	\end{equation*}
	For any fixed $t\in [0,n]$, if $\alpha_{i}^{n-1+t}=\sum\limits_{j=0}^{n-1}f_{j}\alpha_{i}^j$ for all $i\in [n]$, then $f_{n-1}=\Lambda_{t}$.
\end{lemma}
In \cite{gu2023twisted}, let $\eta_{2}=\cdots=\eta_{l}=0$, then we can compute the parity check matrix of the $TRS_{k}(\mathcal{A},l,l\eta)$.
\begin{theorem}\label{thm:3}
	For $n$-elements $\mathcal{A}=\left\{\alpha_{1},\cdots,\alpha_{n}\right\}\subseteq\mathbb{F}_{q}$, let  $\prod\limits_{i=1}^{n}(x-\alpha_{i})=\sum\limits_{j=0}^{n}\sigma_{j}x^{n-j}$ and $u_{i}=\prod\limits_{j=1,j\neq i}^{n}(\alpha_{i}-\alpha_{j})^{-1}$ for all $1\leq i\leq n$. For $\eta\in \mathbb{F}_{q}^{*}$ and $0\leq l\leq k-1$, the TRS code $TRS_{k}(\mathcal{A},l,\eta)$ has a parity check matrix
	$$
	H=\begin{pmatrix}
	u_{1}&\cdots&u_{n}\\
	u_{1}\alpha_{1}&\cdots&u_{n}\alpha_{n}\\
	\vdots&\vdots&\vdots\\
	u_{1}\alpha_{1}^{n-k-2}&\cdots&u_{n}\alpha_{n}^{n-k-2}\\
	u_{1}f(\alpha_{1})&\cdots&u_{n}f(\alpha_{n})
	\end{pmatrix},
	$$
	where $f(x)=x^{n-k-1}\left(1-\eta\sum\limits_{j=0}^{k-l}\sigma_{j}x^{k-l-j}\right)$.
\end{theorem}

\subsection{Character Sums}
In this subsection, we present some basic notations and results of character sum theory~\cite{lidl1997finite}. 

Suppose $\mathbb{F}_q$ is a finite filed with characteristic $p$ and of size $q=p^m$. The absolute trace function ${\rm Tr}(x):\, \mathbb{F}_q\rightarrow \mathbb{F}_p$  is defined by 
$$
{\rm Tr}(x)=x+x^p+x^{p^2}+\cdots+x^{p^{m-1}} .
$$

An additive character $\chi$ of $\mathbb{F}_q$ is a group homomorphism from $(\mathbb{F}_q, +)$ into the multiplicative group $\mathbb{S}^1=\{c\in \mathbb{C}\mid |c|=1\}$ of complex numbers of absolute value $1$. For any positive integer $n$, denote by $\zeta_n=e^{\frac{2 \pi \sqrt{-1}}{n}}$ a $n$-th root of unity. For any $a \in \mathbb{F}_q$, the function
$$
\chi_a(x)=\zeta_p^{{\rm Tr}(a x)},\,\forall x\in \mathbb{F}_q
$$
defines an additive character of $\mathbb{F}_q$. For $a=0$, $\chi_a(x)\equiv 1$ is called the trivial additive character of $\mathbb{F}_q$. For $a=1$,  $\chi_1(x)=\zeta_p^{{\rm Tr}(x)}$ is called the canonical additive character of $\mathbb{F}_q$. 

Homomorphisms from the multiplicative group $(\mathbb{F}_q^*,\times)$ to the multiplicative group $\mathbb{S}^1$ are called multiplicative characters of $\mathbb{F}_q$. Fixing a primitive element $\xi$ of $\mathbb{F}_q$, it is known that all multiplicative characters are given by
$$
\psi_i\left(\xi^j\right)=\zeta_{q-1}^{i j} \text { for } j=0,1, \ldots, q-2,
$$
where $0 \leq i \leq q-2$. It is convenient to extend the definition of $\psi_i$ by setting $\psi_i(0)=0$ for $i\neq 0$ and $\psi_0(0)=1$. The character $\psi_0$ is called the trivial multiplicative character of $\mathbb{F}_q$. The multiplicative character $\psi_{(q-1) / 2}$ is called the quadratic character of $\mathbb{F}_q$, and is denoted by $\pi$ in this paper. That is, $\pi(x)=1$ if $x \in \mathbb{F}_q^*$ is a square; $\pi(x)=-1$ if $x \in \mathbb{F}_q^*$ is not a square. 

Let $\psi$ be a multiplicative character and $\chi$ an additive character of $\mathbb{F}_q$. The $\operatorname{Gauss} \operatorname{sum} G(\psi, \chi)$ is defined by
$$
G(\psi, \chi)=\sum_{x \in \mathbb{F}_q^*} \psi(x) \chi(x).
$$

Here we list some important facts from character sum theory.
\begin{proposition}[\cite{lidl1997finite}]\label{gauss}
	\begin{description}\label{prop1}
		\item[(i)]$G(\psi, \chi_{ab})=\overline{\psi(a)}G(\psi, \chi_b)$, for any $a \in \mathbb{F}_q^*$ and $b \in \mathbb{F}_q$.
		\item[(ii)] If $\psi \neq \psi_0$ and $\chi \neq \chi_0$, then $|G(\psi, \chi)|=\sqrt{q}$.
	\end{description}
	
\end{proposition}

\begin{proposition}[\cite{lidl1997finite}, Theorems 5.32, 5.33 and 5.34]\label{prop2}
	Let $\chi\neq \chi_0$ be a nontrivial additive character of $\mathbb{F}_q$.
	\begin{description}
		\item[(i)] Suppose $n \in \mathbb{N}$ and $d=\operatorname{gcd}(n, q-1)$. Then
		$$\left|\sum_{c \in \mathbb{F}_q} \chi\left(a c^n+b\right)\right| \leqslant(d-1) q^{1 / 2}$$
		for any $a, b \in \mathbb{F}_q$ with $a \neq 0$.
		\item[(ii)] Suppose $q$ odd, let $f(x)=a_2 x^2+a_1 x+a_0 \in \mathbb{F}_q[x]$ with $a_2 \neq 0$. Then
		$$
		\sum_{c \in \mathbb{F}_q} \chi(f(c))=\chi\left(a_0-a_1^2\left(4 a_2\right)^{-1}\right) \pi\left(a_2\right) G(\pi, \chi).
		$$
		\item[(iii)] Suppose $q$ even and  $b\in \mathbb{F}_q^*$.  Let $f(x)=a_2 x^2+a_1 x+a_0 \in \mathbb{F}_q[x]$ with $a_2 \neq 0$. Then
		$$
		\sum_{c \in \mathbb{F}_q} \chi_b(f(c))=\begin{cases}\chi_b(a_0)q & {\rm if~} ba_2+b^2a_1^2=0,\\
		0 &{\rm otherwise}.
		\end{cases}
		$$
	\end{description}
	
\end{proposition}

\begin{proposition}[\cite{lidl1997finite}, Theorem 5.41]\label{prop3}
	Let $\psi$ be a multiplicative character of $\mathbb{F}_q$ of order $m>1$ and let $f \in \mathbb{F}_q[x]$ be a monic polynomial of positive degree that is not an $m$-th power of a polynomial. Let $d$ be the number of distinct roots of $f$ in its splitting field over $\mathbb{F}_q$. Then for every $a \in \mathbb{F}_q$, we have
	$$
	\left|\sum_{c \in \mathbb{F}_q} \psi(a f(c))\right| \leqslant(d-1) q^{1 / 2}.
	$$	
\end{proposition}

For a nontrivial additive character $\chi$ of $\mathbb{F}_{q}$ and $a,b\in \mathbb{F}_{q}$ the sum
\begin{equation*}
K(\chi;a,b)=\sum\limits_{c\in \mathbb{F}_{q}^{*}}\chi(ac+bc^{-1})
\end{equation*}
is called a Kloosterman sum.

\begin{proposition}[\cite{lidl1997finite}, Theorem 5.45]\label{prop5}
	Let $\chi$ be a nontrivial additive character of $\mathbb{F}_{q}$ and $a,b\in \mathbb{F}_{q}$ not both $0$. Then the Kloosterman sum $K(\chi;a,b)$ satisfies
	\begin{equation*}
	\left|K(\chi;a,b)\right|\leq 2q^{1/2}.
	\end{equation*}
\end{proposition}
The theory of character sums is widely used in counting rational points in algebraic geometry.
\begin{proposition}[\cite{lidl1997finite}, Lemma 6.24]\label{prop4}
	For odd $q$, let $b\in \mathbb{F}_q,a_1,a_2\in\mathbb{F}^*_q$, and $\pi$ be the quadratic character of $\mathbb{F}_q$. Then ${\rm N}(a_1X^2+a_2Y^2-b)=q+v(b)\pi(-a_1a_2),$
	where $v(0)=q-1$ and $v(b)=-1$ for $b\in \mathbb{F}^*_q$.
\end{proposition}

\section{The covering radius and deep holes of TRS codes $TRS_{k}(\mathcal{A},l,\eta)$}

In this section, we first determine the covering radius of $TRS_{k}(\mathcal{A},l,\eta)$, and then give some classes of deep holes.


\begin{theorem}
	The covering radius $\rho(TRS_{k}(\mathcal{A},l,\eta))$ of $TRS_{k}(\mathcal{A},l,\eta)$ is equal to $n-k$. Moreover, vectors in $RS_{k+1}(\mathcal{A})\setminus TRS_{k}(\mathcal{A},l,\eta)$ are deep holes of $TRS_{k}(\mathcal{A},l,\eta)$.
\end{theorem}
\begin{proof}
	Since $\dim_{\mathbb{F}_{q}}(TRS_k(\mathcal{A},l,\eta))= k$, by the redundancy bound \cite[Corollary 11.1.3]{huffman2010fundamentals}, we have $$\rho(TRS_k(\mathcal{A},l,\eta))\leq n-k.$$
	Note that $TRS_k(\mathcal{A},l,\eta)$ is a subcode of the Reed-Solomon code $RS_{k+1}(\mathcal{A})$, then by the supercode lemma \cite[Lemma 11.1.5]{huffman2010fundamentals}, we have 
	$$\rho(TRS_{k}(\mathcal{A},l,\eta))\geq d(RS_{k+1}(\mathcal{A}))=n-k.$$
    So $\rho(TRS_{k}(\mathcal{A},l,\eta))=n-k.$ From the argument above, it is easy to obtain that  vectors in $RS_{k+1}(\mathcal{A})\setminus TRS_{k}(\mathcal{A},l,\eta)$ have error distance $n-k$ from the code $TRS_{k}(\mathcal{A},l,\eta)$ and hence they are deep holes.
\end{proof}

Indeed, the above theorem can be generalized to any $1$-codimensional subcode of an MDS code.
\begin{theorem}
	Suppose $\mathcal{C}_{0}$ is an $[n,k+1,d]$ MDS codes over $\mathbb{F}_{q}$. For any $k$-dimensional subcode $\mathcal{C}\subseteq\mathcal{C}_{0}$, then we have
	
	$(1)$\ The covering radius $\rho(\mathcal{C})$ of linear code $\mathcal{C}$ is $n-k$.
	
	$(2)$ The words in $\mathcal{C}_{0}\backslash\mathcal{C}$ are deep holes of code $\mathcal{C}$.
\end{theorem}
\begin{proof}
	$(1)$\ On the one hand, we have $\rho(\mathcal{C})\leq n-k$ by the redundancy bound \cite[Corollary 11.1.3]{huffman2010fundamentals}. On the other hand, by the supercode lemma \cite[Lemma 11.1.5]{huffman2010fundamentals}, we have $$\rho(\mathcal{C})\geq\min\left\{wt(x):x\in\mathcal{C}_{0}\backslash\mathcal{C}\right\}\geq\min\left\{wt(x):x\in\mathcal{C}_{0}\backslash\{\boldsymbol{0}\}\right\}=n-k.$$ Thus, $\rho(\mathcal{C})=n-k$.
	
	$(2)$\ Let G is a generator matrix of $\mathcal{C}$ and $\boldsymbol{u}\in \mathbb{F}_{q}^n$. Let $G^{\prime}=\begin{pmatrix}
	G\\ \boldsymbol{u} 
	\end{pmatrix}$ and denote $\mathcal{C}_{1}$ as the code generated by $G^{\prime}$. Then $d(\mathcal{C}_{1})=\min\left\{d(\mathcal{C}),d(\boldsymbol{u},\mathcal{C})\right\}$. Because $d(\mathcal{C})\geq n-k=\rho(\mathcal{C})\geq d(\boldsymbol{u},\mathcal{C})$, we have $d(\mathcal{C}_{1})=d(\boldsymbol{u},\mathcal{C})$. Therefore, $\boldsymbol{u}$ is a deep hole of $\mathcal{C}$, if and only if $d(\mathcal{C}_{1})=d(\boldsymbol{u},\mathcal{C})=n-k$, if and only if $\mathcal{C}_{1}$ is an $[n,k+1,n-k]$-MDS code over $\mathbb{F}_{q}$. Thus, the words in $\mathcal{C}_{0}\backslash\mathcal{C}$ are deep holes of code $\mathcal{C}$.
\end{proof}
The following proposition describes the geometry of deep holes.

\begin{proposition}\label{Prop:3.2}\cite{fang2024deep}
	Let $\mathcal{C}$ be an $[n,k]$-linear code with a parity-check matrix H and covering radius $\rho(\mathcal{C})$. For $\boldsymbol{u}\in \mathbb{F}_{q}^n$, $\boldsymbol{u}$ is a deep hole of $\mathcal{C}$ if and only if $H\cdot \boldsymbol{u}^T$ can not be expressed as a linear combination of any $\rho(\mathcal{C})-1$ columns of H over $\mathbb{F}_{q}$.
\end{proposition}


Next, we determine the deep holes of $TRS_{k}(\mathcal{A},l,\eta)$.
The following lemma plays an important  role in determining deep holes of the TRS code $TRS_{k}(\mathcal{A},l,\eta)$.

\begin{lemma}\cite[Lemma 2.3]{yan2024mutually}\label{Lem:3.5}
	Let $m$ be a fixed positive integer and $$I_{m}=\left\{0,1,\cdots,m-1\right\}=\left\{t_{1},\cdots,t_{s}\right\}\bigcup\left\{r_{1},r_{2},\cdots,r_{s^{'}}\right\}$$ be any partition of $I_{m}$ with $m=s+s^{\prime},0=t_{1}<t_2<\cdots<t_{s}=m-1$ and $r_{1}<r_{2}<\cdots<r_{s^{'}}$.
	For any $\mathcal{S}=\left\{a_{1},a_{2},\cdots,a_{s}\right\}\subseteq \mathbb{F}_{q}$, denote by $S_{i}(\mathcal{S})=\sum\limits_{1\leq j_{1}<\cdots<j_{i}\leq s}\prod\limits_{t=1}^{i}a_{j_{t}}$. Then we have the following determinant formula
	\begin{equation*}
	\det \begin{pmatrix}
	a_{1}^{t_{1}}&a_{2}^{t_{1}}&\cdots&a_{s}^{t_{1}}\\
	a_{1}^{t_{2}}&a_{2}^{t_{2}}&\cdots&a_{s}^{t_{2}}\\
	\vdots&\vdots&\vdots&\vdots\\
	a_{1}^{t_{s}}&a_{2}^{t_{s}}&\cdots&a_{s}^{t_{s}}\\
	\end{pmatrix}=\prod\limits_{1\leq i<j\leq s}(a_{j}-a_{i})\cdot\vartriangle
	\end{equation*}
	where $\vartriangle$ denotes the determinant of the following matrix
	$$\begin{pmatrix}
	S_{s-r_{1}}(\mathcal{S})&S_{s-r_{2}}(\mathcal{S})&\cdots&S_{s-r_{s^{\prime}}}(\mathcal{S})\\
	S_{s-r_{1}+1}(\mathcal{S})&S_{s-r_{2}+1}(\mathcal{S})&\cdots&S_{s-r_{s^{\prime}}+1}(\mathcal{S})\\
	\vdots&\vdots&\vdots&\vdots\\
	S_{s-r_{1}+s^{\prime}-1}(\mathcal{S})&S_{s-r_{2}+s^{\prime}-1}(\mathcal{S})&\cdots&S_{s-r_{s^{\prime}}+s^{\prime}-1}(\mathcal{S})
	\end{pmatrix}.$$

\end{lemma}

\begin{theorem}\label{The:8}
	Notations as in Theorem~\ref{thm:3} and Lemma~\ref{Lem:3.5}. For any $\boldsymbol{a}=(a_{0},a_{1},\cdots,a_{n-k-1})\in \mathbb{F}_{q}^{n-k}$,  $\boldsymbol{a}^T$ can not be expressed as a $\mathbb{F}_{q}$-linear combination of any $n-k-1$ columns of $H$, if and only if
	for each $1\leq i_{1}<i_{2}<\cdots<i_{n-k-1}\leq n$, let
	$\sum\limits_{j=0}^{n-k-1}c_{j}x^{n-k-1-j}=\prod\limits_{j=1}^{n-k-1}(x-\alpha_{i_{j}})$ and $\Lambda_{0}^{'}=1,\Lambda_{i}^{'}=-\sum\limits_{j=1}^{i}c_{j}\Lambda_{i-j}^{'},i=1,2,\cdots,n-k-1$, 
	it all holds
	$$\sum\limits_{r=0}^{n-k-2}a_{r}c_{n-k-1-r}-\eta\sum\limits_{r=0}^{n-k-2}\sum\limits_{t=0}^{k-l}a_{r}\sigma_{k-l-t}\sum\limits_{\max\{0,r-t\}\leq w\leq r}c_{n-k-1-w}\Lambda_{t+w-r}^{\prime}\neq -a_{n-k-1}.$$
    In particular, for any $a\in \mathbb{F}_{q}^{*}$, the vector $(0,\cdots,0,a)^{T}$ can not be expressed as a $\mathbb{F}_{q}$-linear combination of any $n-k-1$ columns of $H$.
\end{theorem}

\begin{proof}
    Denote by $\boldsymbol{h}_{i}$ the $i$-th column of $H$. The column vector $(a_{0},a_{1},\cdots,a_{n-k-1})^{T}$ can be linearly expressed by $\boldsymbol{h}_{i_{1}},\boldsymbol{h}_{i_{2}},\cdots,\boldsymbol{h}_{i_{n-k-1}}$ for some $1\leq i_{1}<i_{2}<\cdots<i_{n-k-1}\leq n$, if and only if the following system of equations has solutions:
	\begin{equation*}
	\begin{pmatrix}
	u_{i_{1}}&u_{i_{2}}&\cdots&u_{i_{n-k-1}}\\
	u_{i_{1}}\alpha_{i_{1}}&u_{i_{2}}\alpha_{i_{2}}&\cdots&u_{i_{n-k-1}}\alpha_{i_{n-k-1}}\\
	\vdots&\vdots&\vdots&\vdots\\
	u_{i_{1}}\alpha_{i_{1}}^{n-k-2}&u_{i_{2}}\alpha_{i_{2}}^{n-k-2}&\cdots&u_{i_{n-k-1}}\alpha_{i_{n-k-1}}^{n-k-2}\\
	u_{i_{1}}f(\alpha_{i_{1}})&u_{i_{2}}f(\alpha_{i_{2}})&\cdots&u_{i_{n-k-1}}f(\alpha_{i_{n-k-1}})
	\end{pmatrix}\cdot\begin{pmatrix}
	x_{1}\\ x_{2}\\ \vdots\\x_{n-k-1}
	\end{pmatrix}=\begin{pmatrix}
	a_0\\a_1\\ \vdots\\ a_{n-k-1}
	\end{pmatrix},
	\end{equation*}
    where $f(x)=x^{n-k-1}\left(1-\eta\sum\limits_{j=0}^{k-l}\sigma_{j}x^{k-l-j}\right)$. 
    Let $\mathcal{S}=\{\alpha_{i_{1}},\alpha_{i_{2}},\cdots,\alpha_{i_{n-k-1}}\}$, $\mathcal{S}_t=\mathcal{S}\setminus\{\alpha_{i_{t}}\}\, (1\leq t\leq n-k-1)$ and 
    \[
    V(\mathcal{S})=\prod\limits_{1\leq j<l\leq n-k-1}(\alpha_{i_{l}}-\alpha_{i_{j}}),\,\,  V(\mathcal{S}_t)=\prod\limits_{1\leq j<l\leq n-k-1\atop j,l\neq t}(\alpha_{i_{l}}-\alpha_{i_{j}})\,\mbox{ and }\, u'_{t}=(-1)^{n-k-1-t}\frac{V(\mathcal{S}_t)}{V(\mathcal{S})}.
    \]
     Denote by
    $$A=\begin{pmatrix}
	1&1&\cdots&1\\
	\alpha_{i_{1}}&\alpha_{i_{2}}&\cdots&\alpha_{i_{n-k-1}}\\
	\vdots&\vdots&\vdots&\vdots\\
	\alpha_{i_{1}}^{n-k-2}&\alpha_{i_{2}}^{n-k-2}&\cdots&\alpha_{i_{n-k-1}}^{n-k-2}
	\end{pmatrix}$$ and $A_{s,t}$ the algebraic complement obtained by removing the $s$-th row and $t$-th column of matrix $A$, respectively. Then by Lemma~\ref{Lem:3.5}, we have 
    \[
      A_{s,t}=(-1)^{s+t}{V(\mathcal{S}_t)}\cdot S_{n-k-1-s}(\mathcal{S}_t).
    \]

	\begin{equation*}
	\begin{aligned}
	&\begin{pmatrix}
	u_{i_{1}}x_{1}\\u_{i_{2}}x_{2}\\ \vdots\\ u_{i_{n-k-1}}x_{n-k-1}
	\end{pmatrix}=A^{-1}\begin{pmatrix}
	a_{0}\\a_{1}\\ \vdots\\a_{n-k-2}
	\end{pmatrix}=\left|A\right|^{-1}\begin{pmatrix}
	A_{1,1}&A_{2,1}&\cdots&A_{n-k-1,1}\\
	A_{1,2}&A_{2,2}&\cdots&A_{n-k-1,2}\\
	\vdots&\vdots&\vdots&\vdots\\
	A_{1,n-k-1}&A_{2,n-k-1}&\cdots&A_{n-k-1,n-k-1}
	\end{pmatrix}\cdot \begin{pmatrix}
	a_{0}\\a_{1}\\ \vdots\\ a_{n-k-2}
	\end{pmatrix}\\
	&=\small{ \begin{pmatrix}
		\cdots&
		(-1)^{j+1}\frac{V(\mathcal{S}_1)}{V(\mathcal{S})}\cdot S_{n-k-1-j}(\mathcal{S}_1)&\cdots\\
		\cdots&
		(-1)^{j+2}\frac{V(\mathcal{S}_2)}{V(\mathcal{S})}\cdot S_{n-k-1-j}(\mathcal{S}_2)&\cdots\\
		\vdots&\vdots&\vdots\\
		\cdots&
		(-1)^{j+n-k-1}\frac{V(\mathcal{S}_{n-k-1})}{V(\mathcal{S})}\cdot S_{n-k-1-j}(\mathcal{S}_{n-k-1})&\cdots\\	
		\end{pmatrix}_{1\leq j\leq n-k-1}}\cdot \begin{pmatrix}
	a_{0}\\a_{1}\\\vdots\\a_{n-k-2}
	\end{pmatrix}\\
	&=\begin{pmatrix}
	\cdots&(-1)^{n-k-1+j}\cdot u_{{1}}^{\prime}\cdot S_{n-k-1-j}(\mathcal{S}_1)&\cdots\\
	\cdots&(-1)^{n-k-1+j}\cdot u_{{2}}^{\prime}\cdot S_{n-k-1-j}(\mathcal{S}_2)&\cdots\\
	\vdots&\vdots&\vdots\\
	\cdots&(-1)^{n-k-1+j}\cdot u_{{n-k-1}}^{\prime}\cdot S_{n-k-1-j}(\mathcal{S}_{n-k-1})&\cdots\\
	\end{pmatrix}_{1\leq j\leq n-k-1}\cdot \begin{pmatrix}
	a_{0}\\a_{1}\\ \vdots\\a_{n-k-2}
	\end{pmatrix}\\
	&=\begin{pmatrix}
	u_{{1}}^{\prime}\sum\limits_{j=0}^{n-k-2}(-1)^{n-k+j}a_{j}S_{n-k-2-j}(\mathcal{S}_1)\\
	u_{{2}}^{\prime}\sum\limits_{j=0}^{n-k-2}(-1)^{n-k+j}a_{j}S_{n-k-2-j}(\mathcal{S}_2)\\
	\vdots\\
	u_{{n-k-1}}^{\prime}\sum\limits_{j=0}^{n-k-2}(-1)^{n-k+j}a_{j}S_{n-k-2-j}(\mathcal{S}_{n-k-1})
	\end{pmatrix}.
	\end{aligned}
	\end{equation*}
   Next, we will plug the above solution into the last parity-check condition associated to the polynomial $f$.
	
	For any positive integer $t$ let $w_{t}^{\prime}=\sum\limits_{j=1}^{n-k-1}u_{{j}}^{\prime}\alpha_{i_{j}}^t$, then 
	\begin{equation}\label{equation:3.1}
	    w_{t}^{\prime}=
    \begin{cases}
       0&\mbox{if}\ 1\leq t\leq n-k-3\\
	1&\mbox{if}\ t=n-k-2.
    \end{cases}    
	\end{equation}
    Let 
    \begin{equation*}
        \Lambda_{0}^{'}=1\ \mbox{and}\ \Lambda_{i}^{'}=-\sum\limits_{j=1}^{i}c_{j}\Lambda_{i-j}^{'},\ \forall 1\leq i\leq n-k-1.
    \end{equation*}
     For any fixed $0\leq t\leq n-k-2$, by the Lagrange interpolation there exist $f_{t,0},f_{t,1},\cdots, f_{t,n-k-2}$ such that
    \[
    \alpha_{i_{j}}^{n-k-2+t}=\sum\limits_{r=0}^{n-k-2}f_{t,r}\alpha_{i_{j}}^r,\,\forall 1\leq j\leq n-k-1.
    \]
     By Lemma~\ref{lemcont:4.1}, we have $f_{t,n-k-2}=\Lambda_{t}^{\prime}$. So \begin{equation}\label{equation:3.2}
	w_{n-k-2+t}^{\prime}=\sum\limits_{j=1}^{n-k-1}u_{{j}}^{\prime}\alpha_{i_{j}}^{n-k-2+t}=\sum\limits_{r=0}^{n-k-2}f_{t,r}\sum\limits_{j=1}^{n-k-1}u_{{j}}^{\prime}\alpha_{i_{j}}^r=f_{t,n-k-2}w_{n-k-2}^{\prime}=f_{t,n-k-2}=\Lambda_{t}^{\prime}.
	\end{equation}
	In addition, we have 
    \begin{equation}\label{equation:3.3}
        \sum\limits_{j=0}^{r}c_{r-j}\Lambda_{j}^{\prime}=\left\{
	\begin{array}{ll}
	1,&if\ r=0\\
	0,&if\ 1\leq r\leq n-k-1
	\end{array}\right..
    \end{equation}
    We divide our discussion into two cases:
\begin{flushleft}
    \textbf{Case 1:}  $\alpha_{i_{j}}\neq 0$ for all $1\leq j\leq n-k-1$. 
\end{flushleft}

	On the one hand, for all $0\leq i\leq n-k-2$ and $1\leq j\leq n-k-1$, we have 
	\begin{equation*}
	\begin{aligned}
	S_{n-k-1-i}(\mathcal{S})&=S_{n-k-2-i}(\mathcal{S}_{j})\alpha_{i_{j}}+S_{n-k-1-i}(\mathcal{S}_{j}).
	\end{aligned}
	\end{equation*}
Thus
	\begin{equation}\label{equation:3.4}
	S_{n-k-2-i}(\mathcal{S}_{j})=\sum\limits_{t=1}^{i+1}(-1)^{t-1}\frac{S_{n-k-2-i+t}(\mathcal{S})}{\alpha_{i_{j}}^t}.
	\end{equation}
    
    On the other hand, for all $0\leq j\leq n-k-1$, we have 
    \begin{equation}\label{equation:3.5}
    c_{j}=(-1)^j\sum\limits_{I\subseteq\{1,\cdots,n-k-1\}\atop \left|I\right|=j}\prod\limits_{s\in I}\alpha_{i_{s}}=(-1)^jS_{j}(\mathcal{S}).
    \end{equation}
    
	 Therefore,
	\begin{equation*}
	\begin{aligned}
	&\sum\limits_{j=1}^{n-k-1}u_{i_{j}}f(\alpha_{i_{j}})x_{j}=\sum\limits_{j=1}^{n-k-1}f(\alpha_{i_{j}})u_{{j}}^{\prime}\sum\limits_{r=0}^{n-k-2}(-1)^{n-k+r}a_{r}S_{n-k-2-r}(\mathcal{S}_{j})\\
	&=\sum\limits_{j=1}^{n-k-1}u_{{j}}^{\prime}\alpha_{i_{j}}^{n-k-1}\left(1-\eta\sum\limits_{t=0}^{k-l}\sigma_{k-l-t}\alpha_{i_{j}}^t\right)\sum\limits_{r=0}^{n-k-2}(-1)^{n-k+r}a_{r}S_{n-k-2-r}(\mathcal{S}_{j})\\
	&\stackrel{(i)}{=}\sum\limits_{j=1}^{n-k-1}u_{{j}}^{\prime}\alpha_{i_{j}}^{n-k-1}\left(1-\eta\sum\limits_{t=0}^{k-l}\sigma_{k-l-t}\alpha_{i_{j}}^t\right)\sum\limits_{r=0}^{n-k-2}\sum\limits_{w=1}^{r+1}(-1)^{n-k+r+w+1}a_{r}\frac{S_{n-k-2-r+w}(\mathcal{S})}{\alpha_{i_{j}}^w}\\
	&\stackrel{(ii)}{=}-\sum\limits_{r=0}^{n-k-2}\sum\limits_{w=1}^{r+1}a_{r}c_{n-k-2-r+w}\sum\limits_{j=1}^{n-k-1}u_{{j}}^{\prime}\alpha_{i_{j}}^{n-k-1-w}\\
    &~~~~~~~~~+\eta\sum\limits_{r=0}^{n-k-2}\sum\limits_{t=0}^{k-l}a_{r}\sigma_{k-l-t}\sum\limits_{w=1}^{r+1}c_{n-k-2-r+w}\sum\limits_{j=1}^{n-k-1}u_{{j}}^{\prime}\alpha_{i_{j}}^{n-k-1+t-w}\\
	&\stackrel{(iii)}{=}-\sum\limits_{r=0}^{n-k-2}a_{r}c_{n-k-1-r}+\eta\sum\limits_{r=0}^{n-k-2}\sum\limits_{t=0}^{k-l}a_{r}\sigma_{k-l-t}\sum\limits_{1\leq w\leq r+1\atop w\leq t+1}c_{n-k-2-r+w}\Lambda_{1+t-w}^{\prime}\\
	&=-\sum\limits_{r=0}^{n-k-2}a_{r}c_{n-k-1-r}+\eta\sum\limits_{r=0}^{n-k-2}\sum\limits_{t=0}^{k-l}a_{r}\sigma_{k-l-t}\sum\limits_{\max\{0,r-t\}\leq w\leq r}c_{n-k-1-w}\Lambda^{\prime}_{t-r+w},
	\end{aligned}
	\end{equation*}
where $(i)$ follows from Equation~(\ref{equation:3.4}), $(ii)$ follows from Equation~(\ref{equation:3.5}) and 
 $(iii)$ follows from Equations~(\ref{equation:3.1}),~(\ref{equation:3.2}) and~(\ref{equation:3.3}).
\begin{flushleft}
    \textbf{Case 2}: There exists $j\in\left\{1,\cdots,n-k-1\right\}$ such that $\alpha_{i_{j}}=0$.
\end{flushleft}

	For simplicity, let $\alpha_{i_{1}}=0$. Then we have $f(\alpha_{i_{1}})=0$ and
    \begin{equation*}
	S_{n-k-2-i}(\mathcal{S}_{j})=\sum\limits_{t=1}^{i+1}(-1)^{t-1}\frac{S_{n-k-2-i+t}(\mathcal{S})}{\alpha_{i_{j}}^t}
	\end{equation*}
	for all $0\leq i\leq n-k-2$ and $2\leq j\leq n-k-1$. Therefore,
	
	\begin{equation*}
	\begin{aligned}
	&\sum\limits_{j=1}^{n-k-1}u_{i_{j}}f(\alpha_{i_{j}})x_{j}=\sum\limits_{j=1}^{n-k-1}f(\alpha_{i_{j}})u_{{j}}^{\prime}\sum\limits_{r=0}^{n-k-2}(-1)^{n-k+r}a_{r}S_{n-k-2-r}(\mathcal{S}_{j})\\
	&=\sum\limits_{j=2}^{n-k-1}u_{{j}}^{\prime}\alpha_{i_{j}}^{n-k-1}\left(1-\eta\sum\limits_{t=0}^{k-l}\sigma_{k-l-t}\alpha_{i_{j}}^t\right)\sum\limits_{r=0}^{n-k-2}\sum\limits_{w=1}^{r+1}(-1)^{n-k+r+w+1}a_{r}\frac{S_{n-k-2-r+w}(\mathcal{S})}{\alpha_{i_{j}}^w}\\
	&=-\sum\limits_{r=0}^{n-k-2}\sum\limits_{w=1}^{r+1}a_{r}c_{n-k-2-r+w}\sum\limits_{j=2}^{n-k-1}u_{{j}}^{\prime}\alpha_{i_{j}}^{n-k-1-w}\\
    &~~~~~~~~+\eta\sum\limits_{r=0}^{n-k-2}\sum\limits_{t=0}^{k-l}a_{r}\sigma_{k-l-t}\sum\limits_{w=1}^{r+1}c_{n-k-2-r+w}\sum\limits_{j=2}^{n-k-1}u_{{j}}^{\prime}\alpha_{i_{j}}^{n-k-1+t-w}\\
	&=-\sum\limits_{r=0}^{n-k-2}a_{r}c_{n-k-1-r}+\eta\sum\limits_{r=0}^{n-k-2}\sum\limits_{t=0}^{k-l}a_{r}\sigma_{k-l-t}\sum\limits_{\max\{0,r-t\}\leq w\leq r}c_{n-k-1-w}\Lambda^{\prime}_{t-r+w}.
	\end{aligned}
	\end{equation*}

	In summary, $(a_{0},a_{1},\cdots,a_{n-k-1})^{T}$ can not be expressed as a linear combination of any $n-k-1$ columns of $H$ over $\mathbb{F}_{q}$, if and only if
	for each $1\leq i_{1}<i_{2}<\cdots<i_{n-k-1}\leq n$, let
	$\sum\limits_{j=0}^{n-k-1}c_{j}x^{n-k-1-j}=\prod\limits_{j=1}^{n-k-1}(x-\alpha_{i_{j}})$ and $\Lambda_{0}^{'}=1,\Lambda_{i}^{'}=-\sum\limits_{j=1}^{i}c_{j}\Lambda_{i-j}^{'},i=1,2,\cdots,n-k-1$, 
	it all holds
	$$\sum\limits_{r=0}^{n-k-2}a_{r}c_{n-k-1-r}-\eta\sum\limits_{r=0}^{n-k-2}\sum\limits_{t=0}^{k-l}a_{r}\sigma_{k-l-t}\sum\limits_{\max\{0,r-t\}\leq w\leq r}c_{n-k-1-w}\Lambda_{t+w-r}^{\prime}\neq -a_{n-k-1}.$$

    In particular, if $\boldsymbol{a}=(0,\cdots,0,a)\in\mathbb{F}_{q}^{n-k},a\in\mathbb{F}_{q}^{*}$, then for each $1\leq i_{1}<i_{2}<\cdots<i_{n-k-1}\leq n$, let
	$\sum\limits_{j=0}^{n-k-1}c_{j}x^{n-k-1-j}=\prod\limits_{j=1}^{n-k-1}(x-\alpha_{i_{j}})$ and $\Lambda_{0}^{'}=1,\Lambda_{i}^{'}=-\sum\limits_{j=1}^{i}c_{j}\Lambda_{i-j}^{'},i=1,2,\cdots,n-k-1$, 
	it all holds
	$$\sum\limits_{r=0}^{n-k-2}a_{r}c_{n-k-1-r}-\eta\sum\limits_{r=0}^{n-k-2}\sum\limits_{t=0}^{k-l}a_{r}\sigma_{k-l-t}\sum\limits_{\max\{0,r-t\}\leq w\leq r}c_{n-k-1-w}\Lambda_{t+w-r}^{\prime}=0\neq -a.$$
    Thus, $\boldsymbol{a}^T$ can not be expressed as a $\mathbb{F}_{q}$-linear combination of any $n-k-1$ columns of $H$.
\end{proof}

We need the following lemma to recover the vectors from their syndromes.
\begin{lemma}\label{Lemma:3.6}
	Notations as in Theorem~\ref{thm:3} and Lemma~\ref{lemcont:4.1}.  For $a_{0},\cdots,a_{n-k-1}\in\mathbb{F}_{q}$,
    the solutions of the following system of linear equations
	\begin{equation}\label{equ:3}
	\begin{pmatrix}
	u_{1}&\cdots&u_{n}\\
	\vdots&\vdots&\vdots\\
	u_{1}\alpha_{1}^{n-k-2}&\cdots&u_{n}\alpha_{n}^{n-k-2}\\
	u_{1}f(\alpha_{1})&\cdots&u_{n}f(\alpha_{n})
	\end{pmatrix}\cdot \begin{pmatrix}
	x_{0}\\x_{1}\\ \vdots\\ x_{n-1}
	\end{pmatrix}=\begin{pmatrix}
	a_{0}\\ a_{1}\\ \vdots\\a_{n-k-1}
	\end{pmatrix}
	\end{equation}
	are $(h(\alpha_{1}),h(\alpha_{2}),\cdots,h(\alpha_{n}))^T+TRS_{k}(\mathcal{A},l,\eta)$, where $$h(x)=\sum\limits_{i=0}^{n-k-1}\sum\limits_{j=0}^{i}\sigma_{i-j}a_{j}x^{n-1-i}+\eta\sum\limits_{i=0}^{n-k-2}\sum\limits_{j=0}^i\sigma_{i-j}a_{j}\sum\limits_{w=0}^{k-l}\sigma_{k-l-w}\Lambda_{n-k-1-i+w}x^k.$$
 \end{lemma}   
    

\begin{proof}
	It is sufficient to prove  
	\begin{equation*}
	\begin{pmatrix}
	u_{1}&\cdots&u_{n}\\
	\vdots&\vdots&\vdots\\
	u_{1}\alpha_{1}^{n-k-2}&\cdots&u_{n}\alpha_{n}^{n-k-2}\\
	u_{1}f(\alpha_{1})&\cdots&u_{n}f(\alpha_{n})
	\end{pmatrix}\cdot \begin{pmatrix}
	h(\alpha_{1})\\
	\vdots\\
        h(\alpha_{n-1})\\
        h(\alpha_{n})
	\end{pmatrix}=\begin{pmatrix}
	a_{0}\\ a_{1}\\ \vdots\\a_{n-k-1}
	\end{pmatrix}.
	\end{equation*}
	For $0\leq t\leq n-k-2$, we have 
    \begin{equation*}
        \begin{aligned}
            &\sum\limits_{r=1}^{n}u_{r}\alpha_{r}^th(\alpha_{r})\\
           =&\sum\limits_{r=1}^{n}u_{r}\alpha_{r}^t
           \left(\sum\limits_{i=0}^{n-k-1}\sum\limits_{j=0}^{i}\sigma_{i-j}a_{j}\alpha_{r}^{n-1-i}+\eta\sum\limits_{i=0}^{n-k-2}\sum\limits_{j=0}^i\sigma_{i-j}a_{j}\sum\limits_{w=0}^{k-l}\sigma_{k-l-w}\Lambda_{n-k-1-i+w}\alpha_{r}^k\right)\\
           =&\sum\limits_{j=0}^{n-k-1}a_{j}\sum\limits_{i=j}^{n-k-1}\sigma_{i-j}\sum\limits_{r=1}^nu_{r}\alpha_{r}^{n-1+t-i}+\eta\sum\limits_{j=0}^{n-k-2}a_{j}\sum\limits_{i=j}^{n-k-2}\sigma_{i-j}\sum\limits_{w=0}^{k-l}\sigma_{k-l-w}\Lambda_{n-k-1-i+w}\sum\limits_{r=1}^{n}u_{r}\alpha_{r}^{t+k}\\
           =&\sum\limits_{j=0}^{t}a_{j}\sum\limits_{i=j}^{t}\sigma_{i-j}\Lambda_{t-i}=a_{t}
        \end{aligned}
    \end{equation*}
	and
	\begin{small}
	\begin{equation*}
	\begin{aligned}
	&\sum\limits_{r=1}^{n}u_{r}f(\alpha_{r})\cdot h(\alpha_{r})\\
    =&\sum\limits_{r=1}^nu_{r}\alpha_{r}^{n-k-1}\left(1-\eta\sum\limits_{w=0}^{k-l}\sigma_{k-l-w}\alpha_{r}^w\right)\cdot \left(\sum\limits_{i=0}^{n-k-1}\sum\limits_{j=0}^{i}\sigma_{i-j}a_{j}\alpha_{r}^{n-1-i}\right.\\
    &~~~~~~~~~~~~~~~~\left.+\eta\sum\limits_{i=0}^{n-k-2}\sum\limits_{j=0}^i\sigma_{i-j}a_{j}\sum\limits_{w=0}^{k-l}\sigma_{k-l-w}\Lambda_{n-k-1-i+w}\alpha_{r}^k\right)\\
	=&\sum\limits_{j=0}^{n-k-1}a_{j}\sum\limits_{i=j}^{n-k-1}\sigma_{i-j}\sum\limits_{r=1}^nu_{r}\alpha_{r}^{n-1+n-k-1-i}-\eta\sum\limits_{j=0}^{n-k-1}a_{j}\sum\limits_{i=j}^{n-k-1}\sigma_{i-j}\sum\limits_{w=0}^{k-l}\sigma_{k-l-w}\sum\limits_{r=1}^{n}u_{r}\alpha_{r}^{n-1+n-k-1-i+w}\\
    &~~~~~~~~~~~~~~~~+\eta\sum\limits_{j=0}^{n-k-2}\sum\limits_{i=j}^{n-k-2}\sigma_{i-j}a_{j}\sum\limits_{w=0}^{k-l}\sigma_{k-l-w}\Lambda_{n-k-1-i+w}\cdot\left(\sum\limits_{r=1}^{n}u_{r}\alpha_{r}^{n-1}-\eta\sum\limits_{w=0}^{k-l}\sigma_{k-l-w}\sum\limits_{r=1}^{n}u_{r}\alpha_{r}^{n-1+w}\right)\\
	=&\sum\limits_{j=0}^{n-k-1}a_{j}\sum\limits_{i=j}^{n-k-1}\sigma_{i-j}\Lambda_{n-k-1-i}-\eta\sum\limits_{j=0}^{n-k-1}a_{j}\sum\limits_{i=j}^{n-k-1}\sigma_{i-j}\sum\limits_{w=0}^{k-l}\sigma_{k-l-w}\Lambda_{n-k-1-i+w}\\
    &~~~~~~~~~~~~~~~~+\eta\sum\limits_{j=0}^{n-k-2}\sum\limits_{i=j}^{n-k-2}\sigma_{i-j}a_{j}\sum\limits_{w=0}^{k-l}\sigma_{k-l-w}\Lambda_{n-k-1-i+w}\cdot\left(1-\eta\sum\limits_{w=0}^{k-l}\sigma_{k-l-w}\Lambda_{w}\right)\\
	=&a_{n-k-1}-\eta\sum\limits_{j=0}^{n-k-2}a_{j}\sum\limits_{i=j}^{n-k-2}\sigma_{i-j}\sum\limits_{w=0}^{k-l}\sigma_{k-l-w}\Lambda_{n-k-1-i+w}+\eta\sum\limits_{j=0}^{n-k-2}a_{j}\sum\limits_{i=j}^{n-k-2}\sigma_{i-j}\sum\limits_{w=0}^{k-l}\sigma_{k-l-w}\Lambda_{n-k-1-i+w}\\
    =&a_{n-k-1}.\\
	\end{aligned}
	\end{equation*}
	\end{small}
	Thus, the solutions of the linear equation system~(\ref{equ:3}) are $(h(\alpha_{1}),h(\alpha_{2}),\cdots,h(\alpha_{n}))^T+TRS_{k}(\mathcal{A},l,\eta)$.
	
\end{proof}

\begin{theorem}\label{The:3.6}
	Notations as in Lemma~\ref{lemcont:4.1}, Theorems~\ref{thm:3} and~\ref{The:8}. For
	$a_{0},a_{1},\cdots,a_{n-k-1}\in\mathbb{F}_{q}$, if for each $1\leq i_{1}<i_{2}<\cdots<i_{n-k-1}\leq n$,
    it all holds
	\begin{equation}\label{Equ:3.1}
	\begin{aligned}
	&\sum\limits_{r=0}^{n-k-2}a_{r}c_{n-k-1-r}-\eta\sum\limits_{r=0}^{n-k-2}\sum\limits_{t=0}^{k-l}a_{r}\sigma_{k-l-t}\sum\limits_{\max\{0,r-t\}\leq w\leq r}c_{n-k-1-w}\Lambda_{t+w-r}^{\prime}\neq-a_{n-k-1}.
	\end{aligned}
	\end{equation}
	Then the vector $\boldsymbol{u}_{f}$ with generating polynomial $$f(x)=\sum\limits_{i=0}^{n-k-1}\sum\limits_{j=0}^{i}\sigma_{i-j}a_{j}x^{n-1-i}+\eta\sum\limits_{i=0}^{n-k-2}\sum\limits_{j=0}^i\sigma_{i-j}a_{j}\sum\limits_{w=0}^{k-l}\sigma_{k-l-w}\Lambda_{n-k-1-i+w}x^k+f_{k,l,\eta}(x)$$ is a deep hole of $TRS_{k}(\mathcal{A},l,\eta)$, where $f_{k,l,\eta}(x)\in\mathcal{S}_{k,l,\eta}$.     
\end{theorem}

\begin{proof}
	On the one hand, from Lemma~\ref{Lemma:3.6}, we have $H\cdot \boldsymbol{u}_{f}^{T}=(a_{0},\cdots,a_{n-k-1})^T$.
	On the other hand, for each $1\leq i_{1}<i_{2}<\cdots<i_{n-k-1}\leq n$, let
	$\sum\limits_{j=0}^{n-k-1}c_{j}x^{n-k-1-j}=\prod\limits_{j=1}^{n-k-1}(x-\alpha_{i_{j}})$ and $$\Lambda_{0}^{'}=1,\Lambda_{i}^{'}=-\sum\limits_{j=1}^{i}\Lambda_{i-j}^{'}c_{j},\ \forall 1\leq i\leq n-k-1,$$
	it all holds
	\begin{equation*}
	\begin{aligned}
	&\sum\limits_{r=0}^{n-k-2}a_{r}c_{n-k-1-r}-\eta\sum\limits_{r=0}^{n-k-2}\sum\limits_{t=0}^{k-l}a_{r}\sigma_{k-l-t}\sum\limits_{\max\{0,r-t\}\leq w\leq r}c_{n-k-1-w}\Lambda_{t+w-r}^{\prime}\neq -a_{n-k-1}.
	\end{aligned}
	\end{equation*}
	Thus, from Theorem~\ref{The:8} and Proposition~\ref{Prop:3.2}, the vector $\boldsymbol{u}_{f}$ with generating polynomial $$f(x)=\sum\limits_{i=0}^{n-k-1}\sum\limits_{j=0}^{i}\sigma_{i-j}a_{j}x^{n-1-i}+\eta\sum\limits_{i=0}^{n-k-2}\sum\limits_{j=0}^i\sigma_{i-j}a_{j}\sum\limits_{w=0}^{k-l}\sigma_{k-l-w}\Lambda_{n-k-1-i+w}x^k+f_{k,l,\eta}(x)$$ is a deep hole of $TRS_{k}(\mathcal{A},l,\eta)$, where $f_{k,l,\eta}(x)\in\mathcal{S}_{k,l,\eta}$.     
\end{proof}

Finally, we can obtain some classes of deep holes of $TRS_{k}(\mathcal{A},l,\eta)$ from Theorem~\ref{The:3.6}.

\begin{corollary}\label{Lem:3.7}
	Suppose $a\in \mathbb{F}_{q}^{*}$ and  $g(x)=ax^k+f_{k,l,\eta}(x)$, where $f_{k,l,\eta}\in\mathcal{S}_{k,l,\eta}$. Then $\boldsymbol{u}=(g(\alpha_{1}),\cdots,g(\alpha_{n}))$  is a deep hole of $TRS_{k}(\mathcal{A},l,\eta)$.
\end{corollary}

\begin{proof}
    From Proposition~\ref{Prop:3.2} and Theorem~\ref{The:8}, $\boldsymbol{a}=(0,\cdots,0,a)\in\mathbb{F}_{q}^{n-k}$ is the syndrome of some deep hole of $TRS_{k}(\mathcal{A},l,\eta)$. Thus, by Theorem~\ref{The:3.6}, the vector $\boldsymbol{u}_{f}$ with generating polynomial $f(x)=ax^k+f_{k,l,\eta}$ is a deep hole of $TRS_{k}(\mathcal{A},l,\eta)$, where $f_{k,l,\eta}(x)\in\mathcal{S}_{k,l,\eta}$.
\end{proof}

\begin{corollary}
		Notations as in Lemma~\ref{lemcont:4.1}, Theorems~\ref{thm:3} and~\ref{The:8}.
        \begin{enumerate}
            \item For $a\in\mathbb{F}_{q}$,
		if for each $1\leq i_{1}<i_{2}<\cdots<i_{n-k-1}\leq n$, it all holds
		\begin{equation*}
		\begin{aligned}
		c_{1}-\eta\sum\limits_{t=0}^{k-l}\sigma_{k-l-t}\sum\limits_{\max\{0,n-k-2-t\}\leq w\leq n-k-2}c_{n-k-1-w}\Lambda_{t+w-(n-k-2)}^{'}\neq -a.
		\end{aligned}
		\end{equation*}
		then the vector $\boldsymbol{u}_{f}$ with generating polynomial $f(x)=x^{k+1}+\left(\sigma_{1}+a-\eta\sigma_{1+k-l}\right)x^k+f_{k,l,\eta}(x)$ is a deep hole of $TRS_{k}(\mathcal{A},l,\eta)$, where $f_{k,l,\eta}(x)\in\mathcal{S}_{k,l,\eta}$. 
        \item If $q>\binom{n}{k+1}$, then there exists $a\in\mathbb{F}_{q}$ such that the vector $\boldsymbol{u}_{f}$ with generating polynomial $f(x)=x^{k+1}+\left(\sigma_{1}+a-\eta\sigma_{1+k-l}\right)x^k+f_{k,l,\eta}(x)$ is a deep hole of $TRS_{k}(\mathcal{A},l,\eta)$, where $f_{k,l,\eta}(x)\in\mathcal{S}_{k,l,\eta}$. 
        \end{enumerate}
           
\end{corollary}
\begin{proof}
\begin{enumerate}
    \item Take $\boldsymbol{a}=(0,\cdots,0,1,a)$ in Theorem~\ref{The:3.6}.
    \item Due to
    \begin{small}
        \begin{equation*}
        \#\left\{c_{1}-\eta\sum\limits_{t=0}^{k-l}\sigma_{k-l-t}\sum\limits_{0\leq w\leq n-k-2\atop n-k-2-t\leq w
        }\Lambda_{t+w-(n-k-2)}^{'}:1\leq i_{1}<\cdots<i_{n-k-1}\leq n\right\}\leq \binom{n}{k+1}<q,
    \end{equation*}
    \end{small}
    there exists $a\in\mathbb{F}_{q}$ such that for all $1\leq i_{1}<i_{2}<\cdots<i_{n-k-1}\leq n$
    \begin{equation*}
		\begin{aligned}
		c_{1}-\eta\sum\limits_{t=0}^{k-l}\sigma_{k-l-t}\sum\limits_{\max\{0,n-k-2-t\}\leq w\leq n-k-2}c_{n-k-1-w}\Lambda_{t+w-(n-k-2)}^{'}\neq -a.
		\end{aligned}
		\end{equation*}
         Thus, by Theorem~\ref{The:3.6}, the vector $\boldsymbol{u}_{f}$ with generating polynomial $$f(x)=x^{k+1}+\left(\sigma_{1}+a-\eta\sigma_{1+k-l}\right)x^k+f_{k,l,\eta}(x)$$ is a deep hole of $TRS_{k}(\mathcal{A},l,\eta)$.
\end{enumerate}
    
\end{proof}


\section{On the completeness of Deep Holes of $TRS_{k}(\mathbb{F}_{q}^{*},k-1,\eta)$}

In this section, we devote to presenting our main theorems on the completeness of deep holes of $TRS_k(\mathbb{F}_q^{*},k-1, \eta)$. 
The results will be divided into two cases: when $q$ is even and when $q$ is odd.

Inspired by the polynomial methods proposed in \cite{fang2024deep,hiren2019mds}, 
we use different methods to solve the deep hole problem for the cases where $q$ is even or odd. Since the evaluation set is not the whole finite field, compared to~\cite{fang2024deep}, our situation is more complicated. 

We first recall some notations which could be simplified in the particular setting $\mathcal{A}=\mathbb{F}_{q}^{*}$ and $l=k-1$. 

Let $n=q-1,r=q-k-2,\left\{\alpha_{1},\cdots,\alpha_{q-1}\right\}=\mathcal{A}=\mathbb{F}_{q}^{*}$ and $G(x)=\prod\limits_{\alpha\in \mathbb{F}_{q}^{*}}(x-\alpha)=\sum\limits_{j=0}^{q-1}\sigma_{q-1-j}x^j$, where $$\sigma_{0}=1\ \mbox{and}\ \sigma_{i}=(-1)^{i}\sum\limits_{1\leq j_{1}<\cdots<j_{i}\leq q-1}\prod\limits_{s=1}^i\alpha_{j_{s}},\ \forall 1\leq i\leq q-1.$$ 
Since $G(x)=x^{q-1}-1$, thus 
$$\sigma_{q-1}=-1,\sigma_{q-2}=\cdots=\sigma_{1}=0.$$ 
In addition, we have $$\Lambda_{s}=-\sum\limits_{i=1}^s\sigma_{i}\Lambda_{s-i}=0,\forall 1\leq s\leq q-2.$$ Therefore, for all $a_{0},\cdots,a_{r-1}\in \mathbb{F}_{q}$, we have
\begin{equation}\label{equation:4.1}
\sum\limits_{j=0}^{r-1}a_{j}\sum\limits_{i=j}^{r-1}\sigma_{i-j}\sum\limits_{w=0}^{k-l}\sigma_{k-l-w}\Lambda_{r-i+w}=\sum\limits_{j=0}^{r-1}a_{j}\Lambda_{q-2-j-l}=0.
\end{equation}

  Let \begin{equation*}
H=\small{\begin{pmatrix}
	u_{1}&\cdots&u_{q-1}\\
	u_{1}\alpha_{1}&\cdots&u_{q-1}\alpha_{q-1}\\
	\vdots&\vdots&\vdots\\
	u_{1}\alpha_{1}^{r-1}&\cdots&u_{q-1}\alpha_{q-1}^{r-1}\\
	u_{1}\alpha_{1}^{r}(1-\eta\alpha_{1})&\cdots&u_{q-1}\alpha_{q-1}^{r}(1-\eta\alpha_{q-1`})
	\end{pmatrix}}
\end{equation*}
be a parity check matrix of $TRS_{k}(\mathbb{F}_{q}^{*},k-1,\eta)$, where $u_{i}=\prod\limits_{1\leq j\leq q-1,j\neq i}(\alpha_{i}-\alpha_{j})^{-1}$ for all $1\leq i\leq q-1$. By Theorem~\ref{The:8} and Proposition~\ref{Prop:3.2}, for any $a_{0},a_{1},\cdots,a_{r}\in \mathbb{F}_{q}$,
$(a_{0},a_{1},\cdots,a_{r})\in\mathbb{F}_{q}^{r+1}$ is a  deep hole syndrome of $TRS_{k}(\mathbb{F}_{q}^{*},k-1,\eta)$, if and only if
for each $1\leq i_{1}<i_{2}<\cdots<i_{r}\leq q-1$,
it all holds $$\sum\limits_{j=0}^{r-1}a_{j}c_{r-j}-\eta\sum\limits_{j=0}^{r-1}a_{j}\sum\limits_{\max\{0,j-1\}\leq w\leq j}c_{r-w}\Lambda_{1+w-j}^{\prime}+a_{r}\neq 0.$$
By Theorem~\ref{The:3.6} and Equation~(\ref{equation:4.1}), if $(a_{0},a_{1},\cdots,a_{r})\in\mathbb{F}_{q}^{r+1}$ is a deep hole syndrome of $TRS_{k}(\mathbb{F}_{q}^{*},k-1,\eta)$ , then the vector $\boldsymbol{u}_{f}$ with generating polynomial $f(x)=\sum\limits_{i=0}^{r}a_{i}x^{r-i}+f_{k,k-1,\eta}(x)$ is a deep hole of $TRS_{k}(\mathbb{F}_{q}^{*},k-1,\eta)$, where $f_{k,k-1,\eta}(x)\in\mathcal{S}_{k,k-1,\eta}$.

Let $(x_{1},\cdots,x_{r})\in (\mathbb{F}_{q}^{*})^r$ and $\boldsymbol{a}=(a_{0},\cdots,a_{r})\in \mathbb{F}_{q}^{r+1}$. For $0\leq i\leq j\leq r$, denote by
$$S_{i,j}=S_{i}(\{x_{1},\cdots,x_{j}\})=\sum\limits_{1\leq t_{1}<\cdots<t_{i}\leq j}\prod\limits_{w=1}^{i}x_{t_{w}}$$
and $S_{i,j}=0$ if $i>j$ or $i<0$.

Since $c_{j}=(-1)^jS_{j,r}=(-1)^jS_{j,r-1}+(-1)^jS_{j-1,r-1}x_{r}$ for $0\leq j\leq r$, we have
\small{\begin{equation}\label{Equ:3.6}
	\begin{aligned}
	&\sum\limits_{j=0}^{r-1}a_{j}c_{r-j}-\eta\sum\limits_{j=0}^{r-1}a_{j}\sum\limits_{\max\{0,j-1\}\leq w\leq j}c_{r-w}\Lambda_{1+w-j}^{\prime}+a_{r}\\
	=&\sum\limits_{j=0}^{r-1}(-1)^{r-j}a_{j}\left(S_{r-j,r-1}+S_{r-1-j,r-1}x_{r}\right)-\eta
	\sum\limits_{j=0}^{r-1}(-1)^{r-j+1}a_{j}\left(S_{r-j+1,r-1}+S_{r-j,r-1}x_{r}\right)\\
	&~~~~~~~~~~~~~~~~+\eta
	\sum\limits_{j=0}^{r-1}(-1)^{r-j+1}a_{j}\left(S_{r-j,r-1}+S_{r-j-1,r-1}x_{r}\right)\left(S_{1,r-1}+x_{r}\right)+a_{r}\\
	=&\eta^{-1}f_{2}-\eta^{-1}f_{3}x_{r}-f_{1}+f_{2}x_{r}+(-f_{2}+f_{3}x_{r})(S_{1,r-1}+x_{r})+a_{r}\\
    =&f_{3}x_{r}^2+f_{3}\left(S_{1,r-1}-\eta^{-1}\right)x_{r}+g,
	\end{aligned}
	\end{equation}}
where $f_{t}=\eta\sum\limits_{j=0}^{r-1}(-1)^{r-j+2-t}a_{j}S_{r-j+2-t,r-1}$ for $t=1,2,3$ and $g=\eta^{-1}f_{2}-f_{1}-S_{1,r-1}f_{2}+a_{r}$.


Since $c_{j}=(-1)^jS_{j,r}=(-1)^jS_{j,r-2}+(-1)^jS_{j-1,r-2}(x_{r-1}+x_{r})+(-1)^jS_{j-2,r-2}x_{r-1}x_{r}$ for $0\leq j\leq r$, we have
\small{\begin{equation}\label{Equ:4.6}
	\begin{aligned}
	&\sum\limits_{j=0}^{r-1}a_{j}c_{r-j}-\eta\sum\limits_{j=0}^{r-1}a_{j}\sum\limits_{\max\{0,j-1\}\leq w\leq j}c_{r-w}\Lambda_{1+w-j}^{\prime}+a_{r}\\
	=&\sum\limits_{j=0}^{r-1}(-1)^{r-j}a_{j}S_{r-j,r}-\eta\sum\limits_{j=0}^{r-1}(-1)^{r-j+1}a_{j}S_{r-j+1,r}+\eta\sum\limits_{j=0}^{r-1}(-1)^{r-j+1}a_{j}S_{r-j,r}S_{1,r}+a_{r}\\
	=&\sum\limits_{j=0}^{r-1}(-1)^{r-j}a_{j}\left(S_{r-j,r-2}+S_{r-j-1,r-2}(x_{r-1}+x_{r})+S_{r-j-2,r-2}x_{r-1}x_{r}\right)\\
	&~~~~~~-\eta\sum\limits_{j=0}^{r-1}(-1)^{r-j+1}a_{j}\left(S_{r-j+1,r-2}+S_{r-j,r-2}(x_{r-1}+x_{r})+S_{r-j-1,r-2}x_{r-1}x_{r}\right)+\eta\sum\limits_{j=0}^{r-1}(-1)^{r-j+1}a_{j}\\
	&~~~~~~\cdot\left(S_{r-j,r-2}+S_{r-j-1,r-2}(x_{r-1}+x_{r})+S_{r-j-2,r-2}x_{r-1}x_{r}\right)\cdot\left(S_{1,r-2}+x_{r-1}+x_{r}\right)+a_{r}\\
	=&\eta^{-1}g_{2}-\eta^{-1}g_{3}(x_{r-1}+x_{r})+\eta^{-1}g_{4}x_{r-1}x_{r}-g_{1}+g_{2}(x_{r-1}+x_{r})-g_{3}x_{r-1}x_{r}\\
	&~~~~~~+(-g_{2}+g_{3}(x_{r-1}+x_{r})-g_{4}x_{r-1}x_{r})(S_{1,r-2}+x_{r-1}+x_{r})+a_{r}\\
	=&(g_{3}-g_{4}x_{r})(x_{r-1}+x_{r})^2+(S_{1,r-2}-\eta^{-1}-x_{r})(g_{3}-g_{4}x_{r})(x_{r-1}+x_{r})\\
	&~~~~~~+(g_{4}(S_{1,r-2}-\eta^{-1})+g_{3})x_{r}^2-g_{2}(S_{1,r-2}-\eta^{-1})-g_{1}+a_{r},
	\end{aligned}
	\end{equation}}
where $g_{t}=\eta\sum\limits_{j=0}^{r-1}(-1)^{r-j+2-t}a_{j}S_{r-j+2-t,r-2}$ for $t=1,2,3,4$.

\subsection{Deep Holes of $TRS_{k}(\mathbb{F}_{q}^{*},k-1,\eta)$
	for Even $q$}

In this subsection, we consider the even $q$ case. Firstly, we provide the following necessary condition for a deep hole syndrome of $TRS_{k}(\mathbb{F}_{q}^{*}, k-1,\eta)$.



\begin{lemma}\label{Lem:3.1.1}
	Notations as above. Suppose $r-2\geq 1$, i.e. $k\leq q-5$ and $q=2^m$. Denote by $\tilde{f_{i}}= f_{i}(x_{1},\cdots,x_{r-2},\eta^{-1}+S_{1,r-2})$ for $i=1,2,3$, $\tilde{g}=\tilde{f}_{1}+a_{r}$ and $V(x_{1},\cdots,x_{r-2})=\prod\limits_{1\leq i<j\leq r-2}(x_{j}-x_{i})$. If $\boldsymbol{a}=(a_{0},a_{1},\cdots,a_{r})$ is a deep hole syndrome of $TRS_{k}(\mathbb{F}_{q}^{*},k-1,\eta)$, then \begin{equation*}
	\begin{aligned}
	P(x_{1},\cdots,x_{r-2})&=V(x_{1},\cdots,x_{r-2})\cdot
 \prod\limits_{t=1}^{r-2}\left(\eta^{-1}+S_{1,r-2}+x_{t}\right)\cdot \tilde{f}_{3}\cdot\tilde{g}\cdot\left(\tilde{f}_{3}(\eta^{-1}+S_{1,r-2})^2+\tilde{g}\right)\prod\limits_{i=1}^{r-2}(\tilde{f}_{3}x_{i}^2+\tilde{g})
	\end{aligned}
	\end{equation*}
	vanishes on $\underbrace{\mathbb{F}_{q}^{*}\times \mathbb{F}_{q}^{*}\times\cdots\times \mathbb{F}_{q}^{*}}_{r-2}$. 
\end{lemma}

\begin{proof}
	For any $x_{1},\cdots,x_{r-2}\in \mathbb{F}_{q}^{*}$, if $x_{i}=x_{j}$ for some $1\leq i\neq j\leq r-2$, then $V(x_{1},\cdots,x_{r-2})=0$. 
	If $\eta^{-1}+S_{1,r-2}+x_{t}=0$ for some $1\leq t\leq r-2$, then $\prod\limits_{t=1}^{r-2}\left(\eta^{-1}+S_{1,r-2}+x_{t}\right)=0$ and we have done. Thus, let $x_{r-1}=\eta^{-1}+S_{1,r-2}$
	, then $x_{1},\cdots,x_{r-1}$ are pairwise distinct. If $\tilde{f}_{3}=0$, we are done. So we assume that $\tilde{f}_{3}\neq 0$.
	
	Since $\boldsymbol{a}=(a_{0},a_{1},\cdots,a_{r})$ is a deep hole syndrome of $TRS_{k}(\mathbb{F}_{q}^{*},k-1,\eta)$, in other words, $\boldsymbol{a}^T=(a_{0},a_{1},\cdots,a_{r})^T$ can not be expressed as a linear combination of any $r$ columns of $H$ over $\mathbb{F}_{q}$. From Theorem~\ref{The:8}, we have 
    $$\sum\limits_{j=0}^{r-1}a_{j}c_{r-j}-\eta\sum\limits_{j=0}^{r-1}a_{j}\sum\limits_{\max\{0,j-1\}\leq w\leq j}c_{r-w}\Lambda_{1+w-j}^{\prime}\neq -a_{r}$$ 
    for any distinct elements $x_{1},x_{2},\cdots,x_{r}\in \mathbb{F}_{q}^{*}$. From Equation~(\ref{Equ:3.6}) and $\eta^{-1}+ S_{1,r-1}=0$, we obtain that 
	\begin{equation*}
	\sum\limits_{j=0}^{r-1}a_{j}c_{r-j}-\eta\sum\limits_{j=0}^{r-1}a_{j}\sum\limits_{\max\{0,j-1\}\leq w\leq j}c_{r-w}\Lambda_{1+w-j}^{\prime}+ a_{r}=\tilde{f}_{3}x_{r}^2+\tilde{g}.
	\end{equation*}
Thus, $\tilde{f}_{3}x_{r}^2+\tilde{g}\neq 0$ for any $x_{r}\in \mathbb{F}_{q}^{*}\backslash \left\{x_{1},\cdots,x_{r-1}\right\}$. Since $\mathbb{F}_{q}$ has characteristic $2$ and $(2,q-1)=1$, the equation $\tilde{f}_{3}X^2+\tilde{g}=0$ has a unique solution $X\in \mathbb{F}_{q}$. Thus, we can deduce that the solution can only be one of $x_{1},\cdots,x_{r-1},0$, that is \begin{equation*}
	\tilde{g}\left(\tilde{f}_{3}(\eta^{-1}+S_{1,r-2})^2+\tilde{g}\right)\prod\limits_{i=1}^{r-2}(\tilde{f}_{3}x_{i}^2+\tilde{g})=0.
	\end{equation*}
\end{proof}

The following two lemmas characterize certain types of syndromes.

\begin{lemma}\label{Lem:3.1.2}
	Suppose $q=2^m\geq 8$ and $\frac{3q+2\sqrt{q}-10}{4}<k\leq q-5$. Let $a_{r-2}+\eta a_{r-1}=0,a_{r-2}\neq 0$ and $a_{r}\in\mathbb{F}_{q}$, then $\boldsymbol{a}=(0,\cdots,0,a_{r-2},a_{r-1},a_{r})\in\mathbb{F}_{q}^{r+1}$ is not a deep hole syndrome of $TRS_{k}(\mathbb{F}_{q}^{*},k-1,\eta)$.
\end{lemma}
\begin{proof}
With loss of generality, we suppose $a_{r-2}=1$, then $a_{r-1}=\eta^{-1}$.  From Theorem~\ref{The:8}, $\boldsymbol{a}=(0,\cdots,0,1,\eta^{-1},a_{r})^T$ can not be expressed as a linear combination of any $r$ columns of H over $\mathbb{F}_{q}$, if and only if
	for each $r$-subset $\{x_{1},\cdots,x_{r}\}\in \mathbb{F}_{q}^{*}$,
    it all holds $$\sum\limits_{j=0}^{r-1}a_{j}c_{r-j}+\eta\sum\limits_{j=0}^{r-1}a_{j}\sum\limits_{\max\{0,j-1\}\leq w\leq j}c_{r-w}\Lambda_{1+w-j}^{\prime}\neq a_{r}.$$

	Let $\beta_{0}=S_{1,r-2}+\eta^{-1},X=x_{r-1}+x_{r}$ and $Y=\beta_{0}+x_{r}$. From Equation~(\ref{Equ:4.6}), we have
	\begin{equation*}
	\begin{aligned}	&\sum\limits_{j=0}^{r-1}a_{j}c_{r-j}+\eta\sum\limits_{j=0}^{r-1}a_{j}\sum\limits_{\max\{0,j-1\}\leq w\leq j}c_{r-w}\Lambda_{1+w-j}^{\prime}+a_{r}\\
	=&(g_{3}+g_{4}x_{r})(x_{r-1}+x_{r})^2+(S_{1,r-2}+\eta^{-1}+x_{r})(g_{3}+g_{4}x_{r})(x_{r-1}+x_{r})\\
	+&(g_{4}(S_{1,r-2}+\eta^{-1})+g_{3})x_{r}^2+g_{2}(S_{1,r-2}+\eta^{-1})+g_{1}+a_{r}\\
	=&(g_{3}+g_{4}(Y+\beta_{0}))X^2+(g_{3}+g_{4}(Y+\beta_{0}))XY+(g_{4}\beta_{0}+g_{3})(Y+\beta_{0})^2+g_{2}\beta_{0}+g_{1}+a_{r}\\
	\stackrel{(1)}{=}&\eta XY(X+Y)+g_{2}\beta_{0}+g_{1}+a_{r},
	\end{aligned}
	\end{equation*}
	where $(1)$ follows from 
    $$\left\{
    \begin{array}{l}
       g_{4}=\eta\sum\limits_{j=0}^{r-1}a_{j}S_{r-j-2,r-2}=\eta \\ 
       g_{3}=\eta\sum\limits_{j=0}^{r-1}a_{j}S_{r-j-1,r-2}=\eta S_{1,r-2}+1=\eta\beta_{0}\\
       g_{2}=\eta\sum\limits_{j=0}^{r-1}a_{j}S_{r-j,r-2}=\eta S_{2,r-2}+S_{1,r-2}\\
       g_{1}=\eta\sum\limits_{j=0}^{r-1}a_{j}S_{r-j+1,r-2}=\eta S_{3,r-2}+S_{2,r-2}\\
       g_{4}\beta_{0}+g_{3}=0
    \end{array}\right..$$
	
	If $g_{2}\beta_{0}+g_{1}+a_{r}=0$, then let \begin{equation*}
	\begin{aligned}
	\tilde{g}(x)&=(\eta S_{2,r-2}x^2+S_{1,r-2}x)(S_{1,r-2}x+\eta^{-1})+(\eta S_{3,r-2}x^3+S_{2,r-2}x^2)+a_{r}\\
	&=\eta(S_{3,r-2}+S_{1,r-2}S_{2,r-2})x^3+S^2_{1,r-2}x^2+\eta^{-1}S_{1,r-2}x+a_{r}
	\end{aligned}
	\end{equation*}
	Since $S_{1,r-2}=S_{1,r-3}+x_{r-2}$ and $q-1>r-2$, there exist $x_{r-2}\in \mathbb{F}_{q}^{*}\backslash\{x_{1},\cdots,x_{r-3},S_{1,r-3}\}$ such that $x_{1},\cdots,x_{r-2}\in\mathbb{F}_{q}^{*}$ are pairwise distinct and  $S^2_{1,r-2}\neq 0$. Thus, $\tilde{g}(x)\neq 0$. Since $\deg(\tilde{g}(x))\leq 3$ and $q-1\geq 4$, there exist $\gamma\in \mathbb{F}_{q}^{*}$ such that $\tilde{g}(\gamma)\neq 0$. Let $\tilde{x}_{i}=\gamma\cdot x_{i},i=1,\cdots,r-2$, we have 
	\begin{equation*}
	\begin{aligned}
	\tilde{g}_{2}\tilde{\beta}_{0}+\tilde{g}_{1}+a_{r}&=(\eta S_{2}(\tilde{x}_{1},\cdots,\tilde{x}_{r-2})+S_{1}(\tilde{x}_{1},\cdots,\tilde{x}_{r-2}))\cdot(S_{1}(\tilde{x}_{1},\cdots,\tilde{x}_{r-2})+\eta^{-1})\\
	&+\eta S_{3}(\tilde{x}_{1},\cdots,\tilde{x}_{r-2})+S_{2}(\tilde{x}_{1},\cdots,\tilde{x}_{r-2})+a_{r}=\tilde{g}(\gamma)\neq 0.
	\end{aligned}
	\end{equation*}
   Thus, let us assume that  $g_{2}\beta_{0}+g_{1}+a_{r}\neq 0$. Let $h=g_{2}\beta_{0}+g_{1}+a_{r}$ and $F(X,Y)=\eta XY(X+Y)+h$, then we only need to show that the equation
	$F(X,Y)$
	has a solution $(X, Y) \in \mathbb{F}^2_q$ with $X+Y+\beta_{0}\neq Y+\beta_{0}\in \mathbb{F}_q \backslash \mathcal{S}$, where $\mathcal{S}=\{x_{1},\cdots,x_{r-2},0\}$. For each $\beta\in\mathcal{S}$, we have $N(F(X,X+\beta_{0}+\beta))\leq 2, N(F(X,\beta_{0}+\beta))\leq 2$ and $N(F(0,Y))=0$. Therefore, we just prove that $N(F(X,Y))>4\left|\mathcal{S}\right|=4r-4$. 
	
	Let $\chi(x)$ be the canonical additive character of $\mathbb{F}_{q}$. Form Proposition~(\ref{prop2}) $(iii)$, we have 
	\begin{equation*}
	\begin{aligned}
	&N(F(X,Y))=\frac{1}{q}\sum\limits_{X,Y,z\in \mathbb{F}_{q}}\chi(zF(X,Y))\\
    &=\frac{1}{q}\sum\limits_{X,Y\in \mathbb{F}_{q}}\chi(0\cdot F(X,Y))+\frac{1}{q}\sum\limits_{X\in \mathbb{F}_{q},z\in \mathbb{F}_{q}^{*}}\chi(zh)+\frac{1}{q}\sum\limits_{Y,z\in \mathbb{F}_{q}^{*}}\sum\limits_{X\in \mathbb{F}_{q}}\chi(z\cdot F(X,Y))\\
	&=q-1+\frac{1}{q}\sum\limits_{Y,z\in \mathbb{F}_{q}^{*}}\sum\limits_{X\in \mathbb{F}_{q}}\chi(\eta zYX^2+\eta zY^2X+zh)=q-1+\sum\limits_{Y,z\in \mathbb{F}_{q}^{*}\atop z=\eta^{-1}Y^{-3}}\chi(zh)\\
	&=q-1+\sum\limits_{Y\in \mathbb{F}_{q}^{*}}\chi(\frac{h}{\eta Y^3})=q-1+\sum\limits_{Y\in \mathbb{F}_{q}^{*}}\chi(\eta^{-1}hY^3)=q-2+\sum\limits_{Y\in \mathbb{F}_{q}}\chi(\eta^{-1}hY^3)
	\end{aligned}
	\end{equation*}
		From\cite[Corollary 5.31, Theorem 5.32]{lidl1997finite}, we have \begin{equation*}
	\left|\sum\limits_{Y\in \mathbb{F}_{q}}\chi(\eta^{-1}hY^3)\right|\leq 2\sqrt{q}.
	\end{equation*}
 Since $r<\frac{q+2-2\sqrt{q}}{4}$, we have $N(F(X,Y))\geq q-2-2\sqrt{q}>4r-4$. Thus, there exist $x_{r-1}\neq x_{r}\in \mathbb{F}_{q}\backslash \mathcal{S}$ such that \begin{equation*}
 \sum\limits_{j=0}^{r-1}a_{j}c_{r-j}+\eta\sum\limits_{j=0}^{r-1}a_{j}\sum\limits_{\max\{0,j-1\}\leq w\leq j}c_{r-w}\Lambda_{1+w-j}^{\prime}=a_{r}.
 \end{equation*}
 Therefore, $\boldsymbol{a}=(0,\cdots,0,1,\eta^{-1},a_{r})\in\mathbb{F}_{q}^{r+1}$ is not a deep hole syndrome of $TRS_{k}(\mathbb{F}_{q}^{*},k-1,\eta)$
 \end{proof}

\begin{lemma}\label{Lem:3.1.3}
	Suppose $q=2^m\geq 8$ and $\frac{3q+2\sqrt{q}-8}{4}<k\leq q-5$. Let $a_{0},a_{1}\in \mathbb{F}_{q}$ are not all zero elements, then $\boldsymbol{a}=(a_{0},a_{1},0,\cdots,0)\in\mathbb{F}_{q}^{r+1}$ is not a deep hole syndrome of $TRS_{k}(\mathbb{F}_{q}^{*},k-1,\eta)$.
\end{lemma}
\begin{proof}
	From Theorem~\ref{The:8} and Proposition~\ref{Prop:3.2}, $\boldsymbol{a}=(a_{0},a_{1},0,\cdots,0)$ is not a deep hole syndrome of $TRS_{k}(\mathbb{F}_{q}^{*},k-1,\eta)$, if and only if
	there exists $r$-subset $\left\{x_{1},\cdots,x_{r}\right\}\subseteq \mathbb{F}_{q}^{*}$ such that
    $$\sum\limits_{j=0}^{r-1}a_{j}c_{r-j}+\eta\sum\limits_{j=0}^{r-1}a_{j}\sum\limits_{\max\{0,j-1\}\leq w\leq j}c_{r-w}\Lambda_{1+w-j}^{\prime}= a_{r}.$$
    where $a_{0},a_{1}\in \mathbb{F}_{q}$ are not all zero elements.

    If $a_{1}=0$, then $\sum\limits_{j=0}^{r-1}a_{j}c_{r-j}+\eta\sum\limits_{j=0}^{r-1}a_{j}\sum\limits_{\max\{0,j-1\}\leq w\leq j}c_{r-w}\Lambda_{1+w-j}^{\prime}+a_{r}=a_{0}c_{r}(1+\eta c_{1})$. Thus, we can choose $r$-subset $\left\{x_{1},\cdots,x_{r}\right\}\subseteq \mathbb{F}_{q}^{*}$ such that $1+\eta c_{1}=0$. In other words, $(a_{0},0,0,\cdots,0)\in\mathbb{F}_{q}^{r+1}$ is not a deep hole syndrome of $TRS_{k}(\mathbb{F}_{q}^{*},k-1,\eta)$.
	
   If $a_{1}\neq 0$, let\begin{equation*}
	    \left\{
        \begin{array}{l}
             g_{1}=\eta\sum\limits_{j=0}^{r-1}a_{j}S_{r-j+1,r-2}=0\\
             g_{2}=\eta\sum\limits_{j=0}^{r-1}a_{j}S_{r-j,r-2}=0\\
             g_{3}=\eta\sum\limits_{j=0}^{r-1}a_{j}S_{r-j-1,r-2}=a_{1}\eta S_{r-2,r-2}\\
             g_{4}=\eta\sum\limits_{j=0}^{r-1}a_{j}S_{r-j-2,r-2}=a_{0}\eta S_{r-2,r-2}+a_{1}\eta S_{r-3,r-2}
        \end{array}\right.
	\end{equation*}
    By the similar proof of Lemma~\ref{Lem:3.1.2}, we can choose $x_{1},\cdots,x_{r-2}\in \mathbb{F}_{q}^{*}$ such that $g_{4}\neq 0,S_{1,r-2}+\eta^{-1}\neq 0$ and $g_{3}+g_{4}(S_{1,r-2}+\eta^{-1})\neq 0$. Let $\beta_{0}=S_{1,r-2}+\eta^{-1},b=g_{4}^{-1}g_{3},X=x_{r-1}+x_{r}+\beta_{0}+b$ and $Y=x_{r}+b$. From Equation~(\ref{Equ:4.6}), we have
	\begin{equation*}
	\begin{aligned}	&\sum\limits_{j=0}^{r-1}a_{j}c_{r-j}+\eta\sum\limits_{j=0}^{r-1}a_{j}\sum\limits_{\max\{0,j-1\}\leq w\leq j}c_{r-w}\Lambda_{1+w-j}^{\prime}+a_{r}\\
	&=(g_{3}+g_{4}x_{r})(x_{r-1}+x_{r})^2+(S_{1,r-2}+\eta^{-1}+x_{r})(g_{3}+g_{4}x_{r})(x_{r-1}+x_{r})\\
	&+(g_{4}(S_{1,r-2}+\eta^{-1})+g_{3})x_{r}^2+g_{2}(S_{1,r-2}+\eta^{-1})+g_{1}+a_{r}\\
	&=g_{4}XY(X+Y)+(g_{4}\beta_{0}+g_{3})XY+(g_{4}\beta_{0}+g_{3})b^2.
	\end{aligned}
	\end{equation*}

Let $F(X,Y)=g_{4}XY(X+Y)+(g_{4}\beta_{0}+g_{3})XY+(g_{4}\beta_{0}+g_{3})b^2$, then we only need to show that the equation
$F(X,Y)$
has a solution $(X, Y) \in \mathbb{F}^2_q$ with $X+Y+\beta_{0}\neq Y+b\in \mathbb{F}_q \backslash \mathcal{S}$, where $\mathcal{S}=\{x_{1},\cdots,x_{r-2},0\}$. For each $\beta\in\mathcal{S}$, we have $N(F(X,X+\beta_{0}+\beta))\leq 2, N(F(X,b+\beta))\leq 2$ and $N(F(\beta_{0}+b,Y))\leq 2$. Therefore, we just prove that $N(F(X,Y))>4\left|\mathcal{S}\right|+2=4(r-1)+2=4r-2$.

Let $\chi(x)$ be the canonical additive character of $\mathbb{F}_{q}$. Form Proposition~(\ref{prop2}) $(iii)$, we have 
\begin{equation*}
\begin{aligned}
&N(F(X,Y))=\frac{1}{q}\sum\limits_{X,Y,z\in \mathbb{F}_{q}}\chi(zF(X,Y))\\
&=\frac{1}{q}\sum\limits_{X,Y\in \mathbb{F}_{q}}\chi(0\cdot F(X,Y))+\frac{1}{q}\sum\limits_{Y\in \mathbb{F}_{q}\atop z\in \mathbb{F}_{q}^{*}}\sum\limits_{X\in \mathbb{F}_{q}}\chi(g_{4}YzX^2+g_{4}Yz(Y+\beta_{0}+b)X+z(g_{4}\beta_{0}+g_{3})b^2)\\
&=q+\sum\limits_{Y\in \mathbb{F}_{q},z\in \mathbb{F}_{q}^{*}\atop Y=g_{4}z(Y+\beta_{0}+b)^2Y^2}\chi(z(g_{4}\beta_{0}+g_{3})b^2)=q-1+\sum\limits_{Y\in \mathbb{F}_{q}^{*}\atop Y\neq \beta_{0}+b}\chi(\frac{(\beta_{0}+b)b^2}{(Y+\beta_{0}+b)^2Y})\\
&=q-1+\sum\limits_{Y\in \mathbb{F}_{q}^{*}\atop Y\neq \beta_{0}+b}\chi\left(\frac{(\beta_{0}+b)b^2}{Y^2(Y+\beta_{0}+b)}\right)=q-1+\sum\limits_{Y\in \mathbb{F}_{q}^{*}\atop Y\neq (\beta_{0}+b)^{-1}}\chi\left(\frac{(\beta_{0}+b)b^2Y^3}{(\beta_{0}+b)Y+1}\right)\\
&=q-1+\sum\limits_{Y\in \mathbb{F}_{q}^{*}\atop Y\neq 1}\chi\left(\frac{b^2}{(\beta_{0}+b)^2}(Y^2+Y+1+Y^{-1})\right)=q-2+\sum\limits_{Y\in \mathbb{F}_{q}^{*}}\chi\left(\frac{b^2}{(\beta_{0}+b)^2}(Y^2+Y+Y^{-1}+1)\right)\\
&=q-2+\chi\left(\frac{b^2}{(\beta_{0}+b)^2}\right)\sum\limits_{Y\in \mathbb{F}_{q}^{*}}\chi\left(\frac{b^2}{(\beta_{0}+b)^2}Y^2\right)\chi\left(\frac{b^2}{(\beta_{0}+b)^2}(Y+Y^{-1})\right)\\
&=q-2+\chi\left(\frac{b}{\beta_{0}+b}\right)\sum\limits_{Y\in \mathbb{F}_{q}^{*}}\chi\left(\frac{b}{\beta_{0}+b}Y\right)\chi\left(\frac{b^2}{(\beta_{0}+b)^2}(Y+Y^{-1})\right)\\
&=q-2+\chi\left(\frac{b}{\beta_{0}+b}\right)\sum\limits_{Y\in \mathbb{F}_{q}^{*}}\chi\left(\frac{\beta_{0}b}{(\beta_{0}+b)^2}Y+\frac{b^2}{(\beta_{0}+b)^2}Y^{-1}\right)
\end{aligned}
\end{equation*}

From Proposition~\ref{prop5}, we have \begin{equation*}
\left|N(F(X,Y))-(q-2)\right|= \left|\chi\left(\frac{b}{\beta_{0}+b}\right)\sum\limits_{Y\in \mathbb{F}_{q}^{*}}\chi\left(
\frac{\beta_{0}b}{(\beta_{0}+b)^2}Y+\frac{b^2}{(\beta_{0}+b)^2}Y^{-1}
\right)\right|\leq 2\sqrt{q}.
\end{equation*}
Thus, $N(F(X,Y))\geq q-2-2\sqrt{q}$.  Since $r<\frac{q-2\sqrt{q}}{4}$, we have $N(F(X,Y))\geq    
  q-2-2\sqrt{q}>4r-2$. Thus, there exist $x_{r-1}\neq x_{r}\in \mathbb{F}_{q}\backslash \mathcal{S}$ such that \begin{equation*}
\sum\limits_{j=0}^{r-1}a_{j}c_{r-j}+\eta\sum\limits_{j=0}^{r-1}a_{j}\sum\limits_{\max\{0,j-1\}\leq w\leq j}c_{r-w}\Lambda_{1+w-j}^{\prime}=0.
\end{equation*}
Thus,
$\boldsymbol{a}=(a_{0},a_{1},0,\cdots,0)$ is not a deep hole syndrome of $TRS_{k}(\mathbb{F}_{q}^{*},k-1,\eta)$
	
\end{proof}

We will need the following form of combinatorial nullstellensatz. 
\begin{lemma}\cite{hiren2019mds}\label{Lem:3.1.4}
	Let $\mathbb{F}_{q}[X_{1},\cdots,X_{n}]$ be a polynomial ring in $n$ variables over the finite field  $\mathbb{F}_{q}$. Let $S\subseteq \mathbb{F}_{q}$ is a finite set. Suppose $P(X_{1},\cdots,X_{n})\in \mathbb{F}_{q}[X_{1},\cdots,X_{n}]$ vanishes on $S\times\cdots\times S$. If $\deg_{X_{i}}(P)<\left|S\right|$ for each $1\leq i\leq n$, then $P\equiv 0$ in $\mathbb{F}_{q}[X_{1},\cdots,X_{n}]$.
\end{lemma}

Now we present the main result for deep holes of $TRS_{k}(\mathbb{F}_{q}^{*},k-1,\eta)$ in the even $q$ case.
\begin{theorem}
		Suppose $q=2^m\geq 8$ and $\frac{3q+2\sqrt{q}-8}{4}<k\leq q-5$. If $\boldsymbol{a}=(a_{0},\cdots,a_{r})\in\mathbb{F}_{q}^{r+1}$ is  a deep hole syndrome of $TRS_{k}(\mathbb{F}_{q}^{*},k-1,\eta)$, then $\boldsymbol{a}=(0,\cdots,0,a_{r})$, where $a_{r}\in \mathbb{F}_{q}^{*}$. Thus, Corollary~\ref{Lem:3.7} provides all deep holes of $TRS_{k}(\mathbb{F}_{q}^{*},k-1,\eta)$.
\end{theorem}

\begin{proof}
	
	Because $\boldsymbol{a}=(a_{0},\cdots,a_{r})\in\mathbb{F}_{q}^{r+1}$ is a deep hole syndrome of $TRS_{k}(\mathbb{F}_{q}^{*},k-1,\eta)$, by Lemma~\ref{Lem:3.1.1} the polynomial 
    \begin{equation*}
	\begin{aligned}
	P(x_{1},\cdots,x_{r-2})&=V(x_{1},\cdots,x_{r-2})\cdot
 \prod\limits_{t=1}^{r-2}\left(\eta^{-1}+S_{1,r-2}+x_{t}\right)\cdot \tilde{f}_{3}\cdot\tilde{g}\cdot\left(\tilde{f}_{3}(\eta^{-1}+S_{1,r-2})^2+\tilde{g}\right)\prod\limits_{i=1}^{r-2}(\tilde{f}_{3}x_{i}^2+\tilde{g})
	\end{aligned}
	\end{equation*}
    vanishes on $\underbrace{\mathbb{F}_{q}^{*}\times\cdots\times \mathbb{F}_{q}^{*}}_{r-2}$. Note that $$\deg_{x_{i}}(P)=r-3+r-3+2+2+4+4+2(r-3)=4(q-k-2)<q-1.$$ By Lemma~\ref{Lem:3.1.4} , we have $P(x_{1},\cdots,x_{r-2})\equiv 0$. We divide our discussion into four cases:
	
	\textbf{Case 1}: If $\tilde{f}_{3}\equiv 0$, then
	\begin{equation*}
	\begin{aligned}
	0\equiv\eta^{-1}\tilde{f}_{3} &=\sum\limits_{j=0}^{r-1}a_{j}S_{r-j-1}(x_{1},\cdots,x_{r-2},\eta^{-1}+S_{1,r-2})\\
	&=\sum\limits_{j=0}^{r-1}a_{j}S_{r-j-1}(x_{1},\cdots,x_{r-2})+(\eta^{-1}+S_{1,r-2})\sum\limits_{j=0}^{r-2}a_{j}S_{r-j-2}(x_{1},\cdots,x_{r-2})	
	\end{aligned}
	\end{equation*}
	
	For $1\leq j\leq r-2$, the coefficient of the term $x_{1}\prod\limits_{i=1}^{j}x_{i}$  is equal to $a_{r-j-2}$, which implies that $a_{0}=a_{1}=\cdots=a_{r-3}=0$. Thus, $$0=a_{r-2}S_{1,r-2}+a_{r-1}+a_{r-2}(\eta^{-1}+S_{1,r-2})=a_{r-1}+\eta^{-1}a_{r-2},$$
    which implies that $a_{r-1}+\eta^{-1}a_{r-2}=0$. 
	If $a_{r-2}=0$, then $a_{r-1}=0$. We claim that $a_{r}\neq 0$, otherwise, we have $\boldsymbol{a}=\boldsymbol{0}$, which implies that $\boldsymbol{a}\in TRS_{k}(\mathbb{F}_{q}^{*},k-1,\eta)$ is not a deep hole syndrome of $TRS_{k}(\mathbb{F}_{q}^{*},k-1,\eta)$.

	If $a_{r-2}\neq 0$, by Lemma~\ref{Lem:3.1.2}, 
	$\boldsymbol{a}=(0,\cdots,0,a_{r-2},\eta^{-1}a_{r-2},a_{r})$ is not a deep hole syndrome of $TRS_{k}(\mathbb{F}_{q}^{*},k-1,\eta)$ , which is a contradiction.
	
	\textbf{Case 2}: If $\tilde{g}\equiv 0$, i.e. 
	\begin{equation*}
	\begin{aligned}
	0\equiv \eta^{-1}\left(\tilde{f}_{1}+a_{r}\right)
	&=\sum\limits_{j=0}^{r-1}a_{j}S_{r-j+1}(x_{1},\cdots,x_{r-2},\eta^{-1}+S_{1,r-2})+\eta^{-1}a_{r}\\
	&=\sum\limits_{j=3}^{r-1}a_{j}S_{r-j+1}(x_{1},\cdots,x_{r-2})+(\eta^{-1}+S_{1,r-2})\sum\limits_{j=2}^{r-1}a_{j}S_{r-j}(x_{1},\cdots,x_{r-2})+\eta^{-1}a_{r}	
	\end{aligned}
	\end{equation*}
	
	For $1\leq j\leq r-2$, the coefficient of the term $x_{1}\prod\limits_{i=1}^{j}x_{i}$  is equal to $a_{r-j}$, which implies that $a_{2}=a_{3}=\cdots=a_{r-1}=0$. In addition, the constant term is equal to $ \eta^{-1}a_{r}$, which implies that $a_{r}=0$. Thus, $\boldsymbol{a}=(a_{0},a_{1},0,\cdots,0)$, where $a_{0},a_{1}\in \mathbb{F}_{q}$.
    
    If $a_{0}=a_{1}=0$, then $\boldsymbol{a}=\boldsymbol{0}$ is not a deep hole syndrome of $TRS_{k}(\mathbb{F}_{q}^{*},k-1,\eta)$, which is a contradiction.  
	
	If $a_{0},a_{1}\in \mathbb{F}_{q}$ are not both zero, by Lemma~\ref{Lem:3.1.3}, 
	$\boldsymbol{a}=(a_{0},a_{1},0,\cdots,0)$ is not a deep hole syndrome of $TRS_{k}(\mathbb{F}_{q}^{*},k-1,\eta)$ , which is a contradiction.
	
	\textbf{Case 3}: If $\tilde{f}_{3}\cdot x_{i}^2+\tilde{g}\equiv 0$ for some $1\leq i\leq r-2$. Note that $\tilde{f}_{3}$ and $\tilde{g}$ are symmetric polynomials and $r\geq 3$, we know that $\tilde{f}_{3}=\tilde{g}\equiv 0$ or $r-2=1$. If $\tilde{f}_{3}=\tilde{g}\equiv 0$, then from Case 1 and Case 2, we have $a_{0}=a_{1}=\cdots=a_{r-1}=a_{r}=0$. Thus, $\boldsymbol{a}=\boldsymbol{0}$. If $r-2=1$. Then 
	\begin{equation*}
	\begin{aligned}
	\tilde{f}_{3}=\eta\sum\limits_{j=0}^{2}a_{j}S_{2-j}(x_{1},\eta^{-1}+x_{1})=\eta a_{0}x_{1}(\eta^{-1}+x_{1})+\eta a_{1}
\eta^{-1}+\eta a_{2}=a_{1}+\eta a_{2}+a_{0}x_{1}+\eta a_{0}x_{1}^2	\end{aligned}
	\end{equation*}
	and
	\begin{equation*}
	\tilde{g}=\tilde{f}_{1}+a_{3}=\eta\sum\limits_{j=0}^{2}a_{j}S_{4-j}(x_{1},\eta^{-1}+x_{1})+a_{3}=a_{3}+a_{2}x_{1}+\eta a_{2}x_{1}^2
	\end{equation*}
	
	It deduces that $a_{0}=a_{1}=a_{2}=a_{3}=0$. Therefore, in this case, we have $\boldsymbol{a}=\boldsymbol{0}$.

	\textbf{Case 4}: If $\tilde{f}_{3}(\eta^{-1}+S_{1,r-2})^2+\tilde{g}\equiv 0$, i.e. 
	\begin{equation*}
	\begin{aligned}
	0&\equiv\eta^{-1}\tilde{f}_{3}(\eta^{-1}+S_{1,r-2})^2+\eta^{-1}\tilde{g}\\
	&=(\eta^{-1}+S_{1,r-2})^2\sum\limits_{j=0}^{r-1}a_{j}S_{r-j-1}(x_{1},\cdots,x_{r-2})+(\eta^{-1}+S_{1,r-2})^3\sum\limits_{j=0}^{r-2}a_{j}S_{r-j-2}(x_{1},\cdots,x_{r-2})\\
	&+\sum\limits_{j=3}^{r-1}a_{j}S_{r-j+1}(x_{1},\cdots,x_{r-2})+(\eta^{-1}+S_{1,r-2})\sum\limits_{j=2}^{r-1}a_{j}S_{r-j}(x_{1},\cdots,x_{r-2})+\eta^{-1}a_{r}
	\end{aligned}
	\end{equation*}

	For $1\leq j\leq r-2$, the coefficient of the term $x_{1}^3\prod\limits_{i=1}^{j}x_{i}$  is equal to $a_{r-j-2}$, which implies that $a_{0}=a_{1}=\cdots=a_{r-3}=0$. In addition, the coefficient of the term $x_{1}^2x_{2},$ is equal to $a_{r-2}$, which implies that $a_{r-2}=0$. Thus, 
	\begin{equation*}
	\begin{aligned}
	0&\equiv\eta^{-1}\tilde{f}_{3}(\eta^{-1}+S_{1,r-2})^2+\eta^{-1}\tilde{g}\\&=a_{r-1}(\eta^{-1}+S_{1,r-2})^2+a_{r-1}S_{2,r-2}+a_{r-1}(\eta^{-1}+S_{1,r-2})S_{1,r-2}+\eta^{-1}a_{r}\\
	&=a_{r-1}S_{2,r-2}+\eta^{-1}a_{r-1}S_{1,r-2}+\eta^{-2}a_{r-1}+\eta^{-1}a_{r},
	\end{aligned}
	\end{equation*} 
	which implies that $a_{r-1}=a_{r}=0$. Therefore, in this case, we have $\boldsymbol{a}=\boldsymbol{0}$.
\end{proof}

Finally, we determine all deep holes of $TRS_{k}(\mathbb{F}_{q}^{*},k-1,\eta)$ for even $q\geq 16$ and $q-4\leq k\leq q-2$.
\begin{theorem}
	Let $q=2^m\geq 16$ and $H\cdot \boldsymbol{u}^{T}=\boldsymbol{a}^T=(a_{0},\cdots,a_{q-k-2})^T\in \mathbb{F}_{q}^{q-k-1}$, then
	
	$(i)$\ For $k=q-2$, $\boldsymbol{u}$ is a deep hole of $TRS_{k}(\mathbb{F}_{q}^{*},k-1,\eta)$ if and only if $\boldsymbol{u}$ is generated by Corollary~\ref{Lem:3.7};
	
	$(ii)$\ For $k=q-3$, $\boldsymbol{u}$ is a deep hole of $TRS_{k}(\mathbb{F}_{q}^{*},k-1,\eta)$ if and only if $\boldsymbol{u}$ is generated by $a_{0}x^{q-2}+a_{1}x^{q-3}+f_{q-3,k-1,\eta}(x)$ with
	$a_{0}=0,a_{1}\neq 0$ or $a_{0}\neq 0,Tr(\frac{a_{1}\eta}{a_{0}})=1$.
	
	$(iii)$\ For $k=q-4$, $\boldsymbol{u}$ is a deep hole of $TRS_{k}(\mathbb{F}_{q}^{*},k-1,\eta)$ if and only if $\boldsymbol{u}$ is given by Corollary~\ref{Lem:3.7} or generated by $a_{1}(x^{q-3}+\eta^{-1}x^{q-4})+f_{q-4,k-1,\eta}(x)$ with $a_{1}\neq 0$ and $2\nmid m$. 
\end{theorem}

\begin{proof}
	For $k=q-2$, then $\rho(TRS_{k}(\mathbb{F}_{q}^{*},k-1,\eta))=q-1-k=1$. Thus, every non-codeword is a deep hole. The conclusion can be easily verified.
	
	For $k=q-3$, if $a_{0}=0$, then by Corollary~\ref{Lem:3.7}, $\boldsymbol{u}$ is a deep hole if and only if $a_{1}\neq 0$. If $a_{0}\neq 0$, then $\boldsymbol{u}$ is a deep hole, if and only  $H\cdot \boldsymbol{u}^T=\boldsymbol{a}^T=(a_{0},a_{1})^{T}$ can not be expressed as a linear combination of any one  column of $H$ over $\mathbb{F}_{q}$ if and only if,
	\begin{equation}\label{Equ:3.7}
	\eta\alpha^2+\alpha+\frac{a_{1}}{a_{0}}\neq 0\ \mbox{for any}\ 
	\alpha\in \mathbb{F}_{q}^{*}
	\end{equation}
	Hence, Equation~(\ref{Equ:3.7}) holds, if and only if $\eta\alpha^2+\alpha+\frac{a_{1}}{a_{0}}=0$ has no roots in $\mathbb{F}_{q}^{*}$. By \cite[Corollary 3.79]{lidl1997finite}, the latter is equivalent to $Tr\left(\frac{a_{1}\eta}{a_{0}}\right)=1$.
	
	For $k=q-4$, if $a_{0}=a_{1}=0$, then by Corollary~\ref{Lem:3.7}, $\boldsymbol{u}$ is a deep hole if and only if $a_{2}\neq 0$. If $a_{0},a_{1}$ are not all zero elements, then $\boldsymbol{u}$ is a deep hole, if and only if $H\cdot \boldsymbol{u}^T=\boldsymbol{a}^T=(a_{0},a_{1},a_{2})^{T}$ can not be expressed as a linear combination of any two columns of $H$ over $\mathbb{F}_{q}$, if and only if for each $x_{1}\neq x_{2}\in \mathbb{F}_{q}^{*}$, it all holds 
	$$\sum\limits_{j=0}^{1}a_{j}c_{2-j}+\eta\sum\limits_{j=0}^{1}a_{j}\sum\limits_{\max\{0,j-1\}\leq w\leq j}c_{2-w}\Lambda_{1+w-j}^{'}\neq a_{2},$$
	where $\Lambda_{0}^{\prime}=1,\Lambda_{1}^{'}=c_{1}=x_{1}+x_{2},\Lambda_{2}=c_{1}\Lambda_{1}^{'}+c_{2}=c_{2}+c_{1}^2=x_{1}^2+x_{2}^2+x_{1}x_{2}$. W.l.o.g., we suppose $x_{1}\in \mathbb{F}_{q}^{*},x_{2}=\lambda x_{1}$, then for any $x_{1}\in \mathbb{F}_{q}^{*}$ and $\lambda\in \mathbb{F}_{q}^{*}\backslash \{1\}$, it all holds
	\begin{equation}\label{Equ:9}
	F(x_{1},\lambda)\stackrel{\vartriangle}{=}\eta a_{0}(\lambda+\lambda^2)x_{1}^3+\left(a_{0}\lambda+\eta a_{1}(1+\lambda+\lambda^2)\right)x_{1}^2+a_{1}(1+\lambda)x_{1}+a_{2}\neq 0.
	\end{equation}
     Next, we determine all vectors $\boldsymbol{a}^T\in \mathbb{F}_{q}^3\backslash\{0\}$ that satisfy Equation~(\ref{Equ:9}).

	\textbf{Case 1}: $a_{1}\neq 0,a_{0}=0$, then Equation~(\ref{Equ:9}) holds if and only if $a_{2}=\eta^{-1}a_{1}$ and $2\nmid m$.
	
	Firstly, we have $$F(x_{1},\lambda)=\eta a_{1}(1+\lambda+\lambda^2)x_{1}^2+a_{1}(1+\lambda)x_{1}+a_{2}=
	a_{1}x_{1}\left(\eta(1+\lambda+\lambda^2)x_{1}+(1+\lambda)\right)+a_{2}.$$ 
    
    If $a_{2}=0$, let $\lambda\in \left\{\alpha\in \mathbb{F}_{q}^{*}:1+\alpha+\alpha^2\neq 0\right\}\backslash\{1\}$ and $x_{1}=\frac{1+\lambda}{\eta(1+\lambda+\lambda^2)}$, then we have $F(x_{1},\lambda)=0$, i.e. the Equation~(\ref{Equ:9}) does not hold. 
    
    If $a_{2}\neq 0$ and $2|m$, let $\lambda=w,x_{1}=\frac{a_{2}}{a_{1}}w$, where $w$ is the primitive cubic root of unity. Then we have
	\begin{equation*}
	F(\frac{a_{2}}{a_{1}}w,w)=\eta a_{1}(1+w+w^2)\frac{a_{2}^2}{a_{1}^2}w^2+a_{1}(1+w)\frac{a_{2}}{a_{1}}w+a_{2}=\left(\eta\frac{a_{2}^2}{a_{1}}+a_{2}\right)(1+w+w^2)=0.
	\end{equation*}

    Next, we show that if $a_{2}\neq 0$ and $2\nmid m$, then the Equation~(\ref{Equ:9}) holds, if and only if $a_{2}=\eta^{-1}a_{1}$. On the one hand, if the Equation (\ref{Equ:9}) holds, then $\frac{a_{2}}{a_{1}}\neq \eta (1+\lambda+\lambda^2)x_{1}^2+(1+\lambda)x_{1}$ for any $x_{1}\in \mathbb{F}_{q}^{*}$ and $\lambda\in \mathbb{F}_{q}^{*}\backslash \{1\}$. Let $x_{1}=\eta^{-1}$, then $\frac{a_{2}}{a_{1}}\neq \eta^{-1}\lambda^{2}$ for all $\lambda\in \mathbb{F}_q^{*}\backslash \{1\}$. Thus, $\frac{a_{2}}{a_{1}}=\eta^{-1}$.
    
    On the another hand, if $a_{2}=\eta^{-1}a_{1}$, then 
    $$F(x_{1},\lambda)=\eta a_{1}(1+\lambda+\lambda^2)x_{1}^2+a_{1}(1+\lambda)x_{1}+a_{1}\eta^{-1}.$$
Because of \begin{equation*}
\frac{\eta(1+\lambda+\lambda^2)}{a_{1}(1+\lambda^2)}F(x_{1},\lambda)=\left(\frac{\eta(1+\lambda+\lambda^2)x_{1}}{1+\lambda}\right)^2+\frac{\eta(1+\lambda+\lambda^2)x_{1}}{1+\lambda}+\frac{1+\lambda+\lambda^2}{1+\lambda^2}
	\end{equation*}
	and $$Tr\left(\frac{1+\lambda+\lambda^2}{1+\lambda^2}\right)=Tr(1)+Tr\left(\frac{1}{1+\lambda}\right)+Tr\left(\frac{1}{1+\lambda^2}\right)=Tr(1)=1,$$
	we have $F(x_{1},\lambda)\neq 0$ for any $x_{1}\in \mathbb{F}_{q}^{*}$ and $\lambda\in \mathbb{F}_{q}^{*}\backslash \{1\}$. Thus, the Equation
~(\ref{Equ:9}) holds.

	\textbf{Case 2}: $a_{0}\neq 0$, then (\ref{Equ:9}) does not hold. In other words, there exists $x_{1}\in\mathbb{F}_{q}^{*},\lambda\in\mathbb{F}_{q}^{*}\backslash\{1\}$ such that $F(x_{1},\lambda)=0$.
	
	W.l.o.g., we suppose $a_{0}=1$, thus
	\begin{equation*}
	F(x_{1},\lambda)=(\eta x_{1}+\eta a_{1})x_{1}^2\lambda^2+(\eta x_{1}^2+x_{1}+\eta a_{1}x_{1}+a_{1})x_{1}\lambda+\eta a_{1}x_{1}^2+a_{1}x_{1}+a_{2}.
	\end{equation*}
	
	If $a_{2}=0$ and $a_{1}=\eta^{-1}$, then $F(\eta^{-1},\lambda)=0$.
	
	If $a_{2}=0$ and $a_{1}\neq \eta^{-1}$, we want to find $x_{1}\in\mathbb{F}_{q}^{*}\backslash\{a_{1},\eta^{-1}\}$ and $\lambda\in\mathbb{F}_{q}^{*}\backslash\{1\}$ such that $F(x_{1},\lambda)=0$.
    For any $x_{1}\in \mathbb{F}_{q}^{*}\backslash
    \{a_{1},\eta^{-1}\}$, we know $F(x_{1},\lambda)=0$ has a solution $\lambda\in \mathbb{F}_{q}^{*}\backslash\{1\}$ if and only if 
	\begin{equation*}
	0=Tr(\frac{\eta a_{1}x_{1}}{(x_{1}+a_{1})(\eta x_{1}+1)})=Tr(\frac{\eta a_{1}^2}{(\eta a_{1}+1)(x_{1}+a_{1})})+Tr(\frac{\eta a_{1}}{(\eta a_{1}+1)(\eta x_{1}+1)}).
	\end{equation*}
	If $Tr(\frac{\eta a_{1}x_{1}}{(x_{1}+a_{1})(\eta x_{1}+1)})=1$ for all $x_{1}\in \mathbb{F}_{q}^{*}\setminus\{a_{1},\eta^{-1}\}$.
	Then 
	\begin{equation*}
	\begin{aligned}
	-(q-3)&=\sum\limits_{x_{1}\in \mathbb{F}_{q}^{*}\setminus\{a_{1},\eta^{-1}\}}(-1)^{Tr(\frac{\eta a_{1}x_{1}}{(\eta x_{1}+1)(x_{1}+a_{1})})}=\sum\limits_{x_{1}\in \mathbb{F}_{q}^{*}\setminus\{a_{1},\eta^{-1}+a_{1}\}}(-1)^{Tr(\frac{\eta a_{1}x_{1}+\eta a_{1}^2}{(x_{1}(\eta x_{1}+\eta a_{1}+1)})}\\
	&=\sum\limits_{x_{1}\in \mathbb{F}_{q}^{*}\setminus\{a_{1}^{-1},\eta(a_{1}\eta+1)^{-1}\}}(-1)^{Tr(\frac{ \eta a_{1}x_{1}+\eta a_{1}^2x_{1}^2}{ (\eta a_{1}+1)x_{1}+\eta})}\\
	&=\sum\limits_{x_{1}\in \mathbb{F}_{q}^{*}\setminus\{a_{1}^{-1},\eta\}}(-1)^{Tr(
	\frac{\eta a_{1}^2}{(\eta a_{1}+1)^2}x_{1}+	\frac{\eta^2 a_{1}}{(\eta a_{1}+1)^2}x_{1}^{-1}+\frac{\eta a_{1}}{\eta a_{1}+1})}\\
    &=(-1)^{Tr(\frac{\eta a_{1}}{\eta a_{1}+1})}\cdot \left(
    \sum\limits_{x_{1}\in \mathbb{F}_{q}^{*}}(-1)^{Tr(\frac{\eta a_{1}^2}{(\eta a_{1}+1)^2}x_{1}+	\frac{\eta^2 a_{1}}{(\eta a_{1}+1)^2}x_{1}^{-1})}-2(-1)^{Tr(\frac{\eta a_{1}}{\eta a_{1}+1})}
    \right)\\
    &=(-1)^{Tr(\frac{\eta a_{1}}{\eta a_{1}+1})}\cdot K(\chi;\frac{\eta a_{1}^2}{(\eta a_{1}+1)^2},\frac{\eta^2 a_{1}}{(\eta a_{1}+1)^2})-2.
	\end{aligned}
	\end{equation*}
where $\chi$ be the canonial additive character of $\mathbb{F}_{q}$. By Proposition~\ref{prop5}, we have $$q-5=\left|K(\chi;\frac{\eta a_{1}^2}{(\eta a_{1}+1)^2},\frac{\eta^2 a_{1}}{(\eta a_{1}+1)^2})\right|\leq 2q^{1/2},$$ 
    which leads to $q-5\leq 2\sqrt{q}\Rightarrow q<16$, contradiction. Thus, there exists $x_{1}\in\mathbb{F}_{q}^{*}\backslash\left\{a_{1},\eta^{-1}\right\}$ and $\lambda\in\mathbb{F}_{q}^{*}\backslash\{1\}$ such that $F(x_{1},\lambda)=0$.

	If $a_{1}\neq \eta^{-1}$ and $a_{2}\neq 0,\eta^{-2}+a_{1}\eta^{-1}$. Let $x_{1}=\eta^{-1}$, then $F(\eta^{-1},\lambda)=(\eta^{-2}+\eta^{-1}a_{1})\lambda^2+a_{2}$. Because of $a_{2}\neq 0,\eta^{-2}+a_{1}\eta^{-1}$, we can choose $\lambda\in \mathbb{F}_{q}^{*}\backslash\{1\}$ such that $F(\eta^{-1},\lambda)=0$.
	
	If $a_{1}\neq \eta^{-1}$ and $a_{2}=\eta^{-2}+a_{1}\eta^{-1}$, then $F(x_{1},\lambda)=\eta(x_{1}+a_{1})x_{1}^2\lambda^2+(a_{1}+x_{1})(\eta x_{1}+1)x_{1}\lambda+\eta a_{1}x_{1}^2+a_{1}x_{1}+\eta^{-2}+a_{1}\eta^{-1}$. Because of \begin{equation*}
	\frac{\eta}{(x_{1}+a_{1})(\eta x_{1}+1)^2}F(x_{1},\lambda)=\frac{\eta^2x_{1}^2\lambda^2}{(\eta x_{1}+1)^2}+\frac{\eta x_{1}\lambda}{\eta x_{1}+1}+\frac{\eta^2 a_{1}x_{1}^2+\eta a_{1}x_{1}+\eta^{-1}+a_{1}}{(x_{1}+a_{1})(\eta x_{1}+1)^2}
	\end{equation*}
	and 
	\begin{equation*}
	Tr\left(\frac{\eta^2 a_{1}x_{1}^2+\eta a_{1}x_{1}+\eta^{-1}+a_{1}}{(x_{1}+a_{1})(\eta x_{1}+1)^2}\right)=Tr(\frac{1}{1+\eta x_{1}})+Tr(\frac{1}{(1+\eta x_{1})^2})+Tr(\frac{\eta^{-1}+a_{1}}{x_{1}+a_{1}})=Tr(\frac{\eta^{-1}+a_{1}}{x_{1}+a_{1}}),
	\end{equation*}
	we can choose $x_{1}\in \mathbb{F}_{q}^{*}\backslash\{\eta^{-1},a_{1}\}$ such that $Tr(\frac{\eta^{-1}+a_{1}}{x_{1}+a_{1}})=0$. In other words,there exists $x_{1}\in \mathbb{F}_{q}^{*}\backslash\{\eta^{-1},a_{1}\}$ and $\lambda\in \mathbb{F}_{q}^{*}\backslash\{1\}$ such that $F(x_{1},\lambda)=0$.
	
	If $a_{1}=\eta^{-1}$ and $a_{2}\neq 0$, then  
$
	F(x_{1},\lambda)=(1+\eta x_{1})x_{1}^2\lambda^2+\eta^{-1}(1+\eta x_{1})^2x_1\lambda+x_{1}^2+\eta^{-1}x_{1}+a_{2},
	$. When $Tr(a_{2}\eta^2)=0$, there exists $x_{1}\in \mathbb{F}_{q}^{*}\backslash\{\eta^{-1}\}$ and $\lambda=1+\eta^{-1}x_{1}^{-1}$ such that $x_{1}^2+\eta^{-1}x_{1}+a_{2}=0$ and $F(x_{1},\lambda)=0$. When $Tr(a_{2}\eta^2)=1$,for $x_{1}\in \mathbb{F}_{q}^{*}\backslash\{\eta^{-1}\}$, we know $F(x_{1},\lambda)=0$ has a solution $\lambda\neq 0,1\in \mathbb{F}_{q}$ if and only if 
	\begin{equation*}
	0=Tr\left(\frac{\eta^2x_{1}^2+\eta x_{1}+\eta^2a_{2}}{(1+\eta x_{1})^3}\right)=Tr(\frac{1}{1+\eta x_{1}})+Tr(\frac{1}{(1+\eta x_{1})^2})+Tr(\frac{a_{2}\eta^2}{(1+\eta x_{1})^3})=Tr(\frac{a_{2}\eta^2}{(1+\eta x_{1})^3}).
	\end{equation*}
	 If $3\nmid q-1$, then we can choose $x_{1}\in \mathbb{F}_{q}^{*}$ such that $Tr(\frac{a_{2}\eta^2}{(1+\eta x_{1})^3})=0$. If $3|q-1$, we suppose that $Tr(\frac{a_{2}\eta^2}{(1+\eta x_{1})^3})=1$ for any $x_{1}\in \mathbb{F}_{q}^{*}\setminus\{\eta^{-1}\}$. Then 
	 \begin{equation*}
	 \begin{aligned}
	 -(q-2)&=\sum\limits_{x_{1}\in \mathbb{F}_{q}^{*}\setminus\{\eta^{-1}\}}(-1)^{Tr(\frac{a_{2}\eta^2}{(1+\eta x_{1})^3})}=\sum\limits_{\gamma\in \mathbb{F}_{q}^{*}\setminus\{1\}}(-1)^{Tr(a_{2}\eta^2\gamma^3)}\\
	 &=\sum\limits_{\gamma\in \mathbb{F}_{q}}(-1)^{Tr(a_{2}\eta^2\gamma^3)}-1-(-1)^{Tr(a_{2}\eta^2)}=\sum\limits_{\gamma\in \mathbb{F}_{q}}(-1)^{Tr(a_{2}\eta^2\gamma^3)}
	 \end{aligned}
	 \end{equation*} 
	 By Proposition~\ref{prop2}, we have $$q-2=\left|\sum\limits_{\gamma\in \mathbb{F}_{q}}(-1)^{Tr(a_{2}\eta^2\gamma^3)}\right|\leq 2\sqrt{q},$$
         which leads to $q-2\leq 2\sqrt{q}\Rightarrow q<8$, contradiction.
	 
	 In short, $\boldsymbol{a}$ is a deep hole syndrome of $TRS_{k}(\mathbb{F}_{q}^{*},k-1,\eta)$, if and only if $\boldsymbol{a}=(0,0,\gamma)$ or $\boldsymbol{a}=\gamma(0,1,\eta^{-1})$ when $2\nmid m$, where $\gamma\in \mathbb{F}_{q}^{*}$.
	
	

\end{proof}

\subsection{Deep Holes of $TRS_{k}(\mathbb{F}_{q}^{*},k-1,\eta)$ for odd $q$}

In this subsection, for $q$ is odd, let $x_{1},\cdots,x_{r}\in \mathbb{F}_{q}^{*}$ be pairwise distinct, $\eta\in \mathbb{F}_{q}^{*}$ and $\boldsymbol{a}=(a_{0},\cdots,a_{r})\in \mathbb{F}_{q}^{r+1}$. Denote $X=x_{r-1}+x_{r},Y=x_{r},\beta_{0}=S_{1,r-2}-\eta^{-1}$, from Equation~(\ref{Equ:4.6}), we have 
   \small{\begin{equation}\label{Equ:10}
   	\begin{aligned}
   	&\sum\limits_{j=0}^{r-1}a_{j}c_{r-j}-\eta\sum\limits_{j=0}^{r-1}a_{j}\sum\limits_{\max\{0,j-1\}\leq w\leq j}c_{r-w}\Lambda_{1+w-j}^{\prime}+a_{r}\\
   	&=(g_{3}-g_{4}Y)X^2+(\beta_{0}-Y)(g_{3}-g_{4}Y)X+(g_{4}\beta_{0}+g_{3})Y^2-g_{2}\beta_{0}-g_{1}+a_{r}.
   	\end{aligned}
   	\end{equation}} 

\begin{lemma}\label{Lem:4.7}
	Suppose $q$ is odd and $3\leq r\leq q-2$. Then there exists $r$-subset $\left\{x_{1},\cdots,x_{r}\right\}\subseteq \mathbb{F}_{q}^{*}$ such that $\sum\limits_{i=1}^{r}x_{i}=\eta^{-1}$.
\end{lemma}

\begin{proof}
	 For any $\left\{x_{1},\cdots,x_{r}\right\}\subseteq \mathbb{F}_{q}^{*}$, if $\sum\limits_{i=1}^{r}x_{i}\neq 0$, let $\tilde{x}_{i}=\frac{x_{i}}{\eta \sum\limits_{j=1}^rx_{j}}$ for all $1\leq i\leq r$, then $\tilde{x}_{1},\cdots,\tilde{x}_{r}\in\mathbb{F}_{q}^{*}$ are pairwise distinct and $\sum\limits_{i=1}^{r}\tilde{x}_{i}=\eta^{-1}$. If 
   $\sum\limits_{i=1}^{r}x_{i}=0$, then there exists 
    $b\in \mathbb{F}_{q}^{*}\backslash\{-x_{1},x_{2}-x_{1},\cdots,x_{r}-x_{1}\}$
    such that $x_{1}+b,x_{2},\cdots,x_{r}$ are pairwise distinct and $x_{1}+b+\sum\limits_{i=2}^{r}{x}_{i}\neq 0$. Finally, repeat the above process for $r$-subset $\left\{x_{1}+b,x_{2},\cdots,x_{r}\right\}$.
    
\end{proof}

\begin{lemma}\label{lem:4.7}
	Suppose $q$ is odd and $3\leq r\leq q-2$. Let $a_{0},a_{1}\in\mathbb{F}_{q}^{*},a_{r}=a_{1}^{r}a_{0}^{-(r-1)}-\eta a_{1}^{r+1}a_{0}^{-r}$ and 
    $a_{j}=a_{1}^ja_{0}^{-(j-1)}$ for all $2\leq j\leq r-1.$
    Then there exists $r$-subset $\left\{x_{1},\cdots,x_{r}\right\}\subseteq \mathbb{F}_{q}^{*}$ such that
	$$\sum\limits_{j=0}^{r-1}a_{j}c_{r-j}-\eta\sum\limits_{j=0}^{r-1}a_{j}\sum\limits_{\max\{0,j-1\}\leq w\leq j}c_{r-w}\Lambda_{1+w-j}^{\prime}+a_{r}=0.$$
\end{lemma}

\begin{proof}
	Let $M=a_{1}a_{0}^{-1}$, for any $\left\{x_{1},\cdots,x_{r}\right\}\subseteq \mathbb{F}_{q}^{*}$, we have 
	 \begin{equation*}
	\begin{aligned}
	&\sum\limits_{j=0}^{r-1}a_{j}c_{r-j}-\eta\sum\limits_{j=0}^{r-1}a_{j}\sum\limits_{\max\{0,j-1\}\leq w\leq j}c_{r-w}\Lambda_{1+w-j}^{\prime}+a_{r}\\
	=&\sum\limits_{j=0}^{r-1}a_{j}c_{r-j}-\eta\Lambda_{1}^{\prime}\sum\limits_{j=0}^{r-1}a_{j}c_{r-j}-\eta\sum\limits_{j=1}^{r-1}a_{j}c_{r-j+1}+a_{r}\\
	=&a_{0}\sum\limits_{j=0}^{r-1}c_{r-j}M^j-\eta a_{0}\Lambda_{1}^{\prime}\sum\limits_{j=0}^{r-1}c_{r-j}M^j-\eta\sum\limits_{j=0}^{r-2}a_{j+1}c_{r-j}+a_{r}\\
	=&a_{0}\sum\limits_{j=0}^{r}c_{r-j}M^j-a_{0}M^r-\eta a_{0}\Lambda_{1}^{\prime}\left(\sum\limits_{j=0}^{r}c_{r-j}M^j-M^r\right)-\eta a_{0}M\left(\sum\limits_{j=0}^rc_{r-j}M^j-M^r-c_{1}M^{r-1}\right)+a_{r}\\
    =&a_{0}\sum\limits_{j=0}^rc_{r-j}M^j\left(1-\eta\Lambda_{1}^{\prime}-\eta M\right)-a_{0}M^r+\eta a_{0}\Lambda_{1}^{\prime}M^r+\eta a_{0}M^{r+1}+\eta a_{0}c_{1}M^r+a_{0}M^r-\eta a_{0}M^{r+1}\\
	=&a_{0}\sum\limits_{j=0}^rc_{r-j}M^j\left(1-\eta\Lambda_{1}^{\prime}-\eta M\right)=a_{0}(1-\eta\Lambda_{1}^{\prime}-\eta M)\prod\limits_{j=1}^{r}(M-x_{j}).
	\end{aligned}
	\end{equation*}
	Thus, we can choose $x_{1}=M, \left\{x_{2},\cdots,x_{r}\right\}\subseteq \mathbb{F}_{q}^{*}\backslash\{M\}$, then we have 
	\begin{equation*}
	\sum\limits_{j=0}^{r-1}a_{j}c_{r-j}-\eta\sum\limits_{j=0}^{r-1}a_{j}\sum\limits_{\max\{0,j-1\}\leq w\leq j}c_{r-w}\Lambda_{1+w-j}^{\prime}+a_{r}=a_{0}\left(1-\eta\Lambda_{1}^{\prime}-\eta M\right)\prod\limits_{j=1}^{r}(M-x_{j})=0.
	\end{equation*}
\end{proof}

\begin{lemma}\label{Lem:4.8}
	Suppose $q$ is odd and $3\leq r\leq q-4$. Let
	\begin{equation*}
    \begin{aligned}
        &T_{1}=\left\{(u_{0},u_{1},\cdots,u_{r})\in \mathbb{F}_{q}^{r+1}:u_{0},u_{1}\neq 0,u_{r}=u_{1}^ru_{0}^{-(r-1)}-\eta u_{1}^{r+1}u_{0}^{-r}\ \mbox{and}\  u_{i}=u_{1}^{i}u_{0}^{-(i-1)}\ \mbox{for all}\ 2\leq i\leq r-1
	\right\},\\
        &~~~~~~~~~~~~~~~~ T_{2}=\left\{(u_{0},\cdots,u_{r})\in \mathbb{F}_{q}^{r+1}:u_{0}=u_{1}=\cdots=u_{r-2}=0\ or\ u_{0}\neq 0,u_{1}=\cdots=u_{r}=0\right\} 
    \end{aligned}
	\end{equation*}
    and $T_{3}=\left\{(0,2b\eta,b,\frac{b}{4\eta})\in\mathbb{F}_{q}^4:b\in\mathbb{F}_{q}^{*}\right\}$.  For any $\boldsymbol{a}=(a_{0},a_{1},\cdots,a_{r})\in \mathbb{F}_{q}^{r+1}\backslash (\bigcup\limits_{i=1}^{3}T_{i})$,
	there exists $\{x_{1},\cdots,x_{r-2}\}\subseteq \mathbb{F}_{q}^{*}$ such that $$\tilde{g}(x)\neq 0 \quad\mbox{and}\quad\tilde{g}_{4}(x)\neq 0$$
	where 
	$\widetilde{g_{t}}(x)=\eta\sum\limits_{j=0}^{r-1}(-1)^{r-j+2-t}a_{j}S_{r-j+2-t,r-2}x^{r-j+2-t},t=1,2,3,4,$
   $ \tilde{\beta}_{0}(x)=S_{1,r-2}x-\eta^{-1}$ and 
	$\widetilde{g}(x)=\tilde{\beta}_{0}(x)\widetilde{g_{4}}(x)\widetilde{g_{3}}(x)^2+\widetilde{g_{3}}(x)^3-\tilde{\beta}_{0}(x)\widetilde{g_{2}}(x)\widetilde{g_{4}}(x)^2-\widetilde{g_{1}}(x)\widetilde{g_{4}}(x)^2+a_{r}\widetilde{g_{4}}(x)^2.$
\end{lemma} 
\begin{proof}
	Please refer to Appendix A for the proof.
\end{proof}
   
   \begin{lemma}\label{Lem:4.9}
Notations as in Lemma~\ref{Lem:4.8}. Suppose $q$ is odd and $3\leq r\leq \frac{q-3\sqrt{q}-3}{4}$. Then for $\boldsymbol{a}=(a_{0},\cdots,a_{r})\in\mathbb{F}_{q}^{r+1}\backslash (\bigcup\limits_{i=1}^{3}T_{i})$,
   	there exists $\left\{x_{1},\cdots,x_{r}\right\}\subseteq \mathbb{F}_{q}^{*}$  such that 
   	\begin{equation*}
   	\sum\limits_{j=0}^{r-1}a_{j}c_{r-j}-\eta\sum\limits_{j=0}^{r-1}a_{j}\sum\limits_{\max\{0,j-1\}\leq w\leq j}c_{r-w}\Lambda_{1+w-j}^{\prime}+a_{r}=0.
   	\end{equation*}
   \end{lemma}

   \begin{proof}
      For any ($r-2$)-subset $\left\{x_{1},\cdots,x_{r-2}\right\}\subseteq\mathbb{F}_{q}^{*}$, let
     \begin{equation*}
     \widetilde{g_{t}}(x)=\eta\sum\limits_{j=0}^{r-1}(-1)^{r-j+2-t}a_{j}S_{r-j+2-t,r-2}x^{r-j+2-t},t=1,2,3,4,
     \end{equation*}
     $ \tilde{\beta}_{0}(x)=S_{1,r-2}x-\eta^{-1}$ and 
	$\widetilde{g}(x)=\tilde{\beta}_{0}(x)\widetilde{g_{4}}(x)\widetilde{g_{3}}(x)^2+\widetilde{g_{3}}(x)^3-\tilde{\beta}_{0}(x)\widetilde{g_{2}}(x)\widetilde{g_{4}}(x)^2-\widetilde{g_{1}}(x)\widetilde{g_{4}}(x)^2+a_{r}\widetilde{g_{4}}(x)^2.$ From Lemma~\ref{Lem:4.8}, since $4\leq r\leq\frac{q-3\sqrt{q}-3}{4}<q-4$, we have $\tilde{g}(x),\tilde{g}_{4}(x)\neq 0$.
      Because of $\deg(\widetilde{g_{4}}(x))\leq r-2,\deg(\widetilde{g}(x))\leq 3r-5$ and $q-1>4r-7$, we can choose $\gamma\in \mathbb{F}_{q}^{*}$ such that  $\tilde{g}_{4}(\gamma)\neq 0,\tilde{g}(\gamma)\neq 0$. Therefore, it can be assumed that there exists $r-2$-subset $\left\{x_{1},\cdots,x_{r-2}\right\}\subseteq\mathbb{F}_{q}^{*}$ such that $g_{4}\neq 0$ and $(\beta_{0}g_{4}+g_{3})g_{3}^2-\beta_{0}g_{2}g_{4}^2-g_{1}g_{4}^2+a_{r}g_{4}^2\neq 0$, i.e. $a_{r}\neq -\frac{(\beta_{0}g_{4}+g_{3})g_{3}^2}{g^2_{4}}+\beta_{0}g_{2}+g_{1}$.

      From Equation~(\ref{Equ:10}),  we have 
   \small{\begin{equation}
   	\begin{aligned}
   	&\sum\limits_{j=0}^{r-1}a_{j}c_{r-j}-\eta\sum\limits_{j=0}^{r-1}a_{j}\sum\limits_{\max\{0,j-1\}\leq w\leq j}c_{r-w}\Lambda_{1+w-j}^{\prime}+a_{r}\\
    &=(g_{3}-g_{4}Y)X^2+(\beta_{0}-Y)(g_{3}-g_{4}Y)X+(g_{4}\beta_{0}+g_{3})Y^2-g_{2}\beta_{0}-g_{1}+a_{r}.
   	\end{aligned}
   	\end{equation}} 
Let $$F(X,Y)\stackrel{\vartriangle}{=}(g_{3}-g_{4}Y)X^2+(\beta_{0}-Y)(g_{3}-g_{4}Y)X+(g_{4}\beta_{0}+g_{3})Y^2-g_{2}\beta_{0}-g_{1}+a_{r},$$ 
then we only need to show that the equation $F(X,Y)$ has a solution $(X,Y)\in \mathbb{F}_{q}^2$ with 
$$X-Y\neq Y\in \mathbb{F}_{q}^{*}\backslash\mathcal{S},$$ where $\mathcal{S}=\left\{x_{1},\cdots,x_{r-2},0\right\}$. For each $\beta\in\mathcal{S},N(F(X,\beta))\leq 2$ and $N(F(X,X-\beta))\leq 2$. In addition, $N(F(2X,X))\leq 3$. Thus, we only need to show that $N(F(X,Y))>(2+2)\left|\mathcal{S}\right|+3=4r-1$, i.e., $N(F(X,Y))\geq 4r$. Next we use character sums to estimate the value of $N(F(X,Y))$.

Let $\chi(x)$ be the canonical additive character of $\mathbb{F}_q$ and  $\pi$ be  the quadratic character  of $\mathbb{F}_q$. Then 
\begin{equation*}
\begin{aligned}
N(F(X,Y))&= \frac{1}{q}\sum\limits_{X,Y,z\in \mathbb{F}_{q}}\chi(zF(X,Y))=q+\frac{1}{q}\sum\limits_{X\in \mathbb{F}_{q}}\sum\limits_{z\in \mathbb{F}_{q}^{*}}\chi(zF(X,0))+\frac{1}{q}\sum\limits_{X\in \mathbb{F}_{q}}\sum\limits_{Y,z\in \mathbb{F}_{q}^{*}}\chi(zF(X,Y))\\
\end{aligned}
\end{equation*}

Let $h_{1}(X)\stackrel{\vartriangle}{=}F(X,0)=g_{3}X^2+\beta_{0}g_{3}X-g_{2}\beta_{0}-g_{1}+a_{r}=g_{3}(X+\frac{\beta_{0}}{2})^2-\frac{g_{3}\beta_{0}^2}{4}-g_{2}\beta_{0}-g_{1}+a_{r}$. Thus, it is easy to verify that 
\begin{equation*}
\frac{1}{q}\sum\limits_{X\in \mathbb{F}_{q}}\sum\limits_{z\in \mathbb{F}_{q}^{*}}\chi(zF(X,0))=
\begin{cases}
   q-1,&\mbox{if}\ g_{3}=0\ \mbox{and}\  \frac{g_{3}\beta_{0}^2}{4}+g_{2}\beta_{0}+g_{1}-a_{r}=0;\\
-1,&\mbox{if}\ g_{3}=0\ \mbox{and}\  \frac{g_{3}\beta_{0}^2}{4}+g_{2}\beta_{0}+g_{1}-a_{r}\neq 0;\\
0,&\mbox{if}\ g_{3}\neq 0\ \mbox{and}\  \frac{g_{3}\beta_{0}^2}{4}+g_{2}\beta_{0}+g_{1}-a_{r}=0;\\
1,&\mbox{if}\ g_{3}\neq 0\ \mbox{and}\  g_{3}^{-1}\left(\frac{g_{3}\beta_{0}^2}{4}+g_{2}\beta_{0}+g_{1}-a_{r}\right)\  \mbox{is a square in}\  \mathbb{F}_{q}^{*};\\
-1,&\mbox{if}\ g_{3}\neq 0\ \mbox{and}\  g_{3}^{-1}\left(\frac{g_{3}\beta_{0}^2}{4}+g_{2}\beta_{0}+g_{1}-a_{r}\right)\  \mbox{is not a square.} 
\end{cases}
\end{equation*}

Next, we estimate the value of $\frac{1}{q}\sum\limits_{X\in \mathbb{F}_{q}}\sum\limits_{Y,z\in \mathbb{F}_{q}^{*}}\chi(zF(X,Y))$. 

\begin{equation*}
\begin{aligned}
&\frac{1}{q}\sum\limits_{X\in \mathbb{F}_{q}}\sum\limits_{Y,z\in \mathbb{F}_{q}^{*}}\chi(zF(X,Y))\\
&=\frac{1}{q}\sum\limits_{Y\in \mathbb{F}_{q}^{*}\backslash\{g_{3}g_{4}^{-1}\}}\sum\limits_{z\in \mathbb{F}_{q}^{*}}\sum\limits_{X\in \mathbb{F}_{q}}\chi(z(g_{3}-g_{4}Y)X^2+z(\beta_{0}-Y)(g_{3}-g_{4}Y)X+z(g_{4}\beta_{0}+g_{3})Y^{2}-zg_{2}\beta_{0}-zg_{1}+za_{r})\\
&+\frac{1}{q}\cdot\sum\limits_{z\in \mathbb{F}_{q}^{*}}\sum\limits_{X\in \mathbb{F}_{q}}\chi\left(z( \frac{(g_{4}\beta_{0}+g_{3})g_{3}^2}{g_{4}^{2}}-g_{2}\beta_{0}-g_{1}+a_{r}        )\right)\cdot \boldsymbol{1}[g_{3}\neq 0]\\
&\stackrel{(1)}{=}\frac{1}{q}\sum\limits_{Y\in \mathbb{F}_{q}^{*}\backslash \{g_{3}g_{4}^{-1}\}}\sum\limits_{z\in \mathbb{F}_{q}^{*}}
\chi(z(g_{4}\beta_{0}+g_{3})Y^2-zg_{2}\beta_{0}-zg_{1}+za_{r}-z\frac{(\beta_{0}-Y)^2(g_{3}-g_{4}Y)}{4})\cdot\\
&\pi(z(g_{3}-g_{4}Y))G(\pi,\chi)-\boldsymbol{1}[g_{3}\neq 0]\\
&\stackrel{(2)}{=}\frac{G(\pi,\chi)}{q}\sum\limits_{W\in \mathbb{F}_{q}^{*}\backslash \{g_{3}\}}\pi(W)\sum\limits_{z\in \mathbb{F}_{q}^{*}}\chi(zh_{2}(W))\pi(z)-\boldsymbol{1}[g_{3}\neq 0]
\end{aligned}
\end{equation*}
where $\boldsymbol{1}[g_{3}\neq 0]=\left\{
\begin{array}{cc}
1,&if\ g_{3}\neq 0\\
0,&if\ g_{3}=0
\end{array}\right.$,  $(1)$ follows from Proposition~\ref{prop2} $(ii)$ and $a_{r}\neq -\frac{(\beta_{0}g_{4}+g_{3})g_{3}^2}{g^2_{4}}+\beta_{0}g_{2}+g_{1}$, $(2)$ follows from that $W=g_{3}-g_{4}Y$ and $h_{2}(W)=\frac{g_{4}\beta_{0}
+g_{3}}{g_{4}^2}(g_{3}-W)^2-g_{2}\beta_{0}-g_{1}+a_{r}-\frac{W(\beta_{0}g_{4}-g_{3}+W)^2}{4g_{4}^2}$. 
  
Denote $\mathcal{Z}_0=\{W\in \mathbb{F}_q^{*}\backslash\{g_{3}\}:h_{2}(W)=0\}$, then by Proposition \ref{prop1} and $\sum\limits_{z\in\mathbb{F}^*_q}\pi(z)=0$, we have 

\begin{equation*}
\begin{aligned}
&\frac{G(\pi,\chi)}{q}\sum\limits_{W\in \mathbb{F}_{q}^{*}\backslash \{g_{3}\}}\pi(W)\sum\limits_{z\in \mathbb{F}_{q}^{*}}\chi(zh_{2}(W))\pi(z)-\boldsymbol{1}[g_{3}\neq 0]\\
&=\frac{G(\pi,\chi)}{q}\sum\limits_{W\in \mathcal{Z}_{0}}\pi(W)\sum\limits_{z\in \mathbb{F}_{q}^{*}}\chi(zh_{2}(W))\pi(z)+\frac{G(\pi,\chi)}{q}\sum\limits_{W\in \mathbb{F}_{q}^{*}\backslash (\{g_{3}\}\cup\mathcal{Z}_{0})}\pi(W)\sum\limits_{z\in \mathbb{F}_{q}^{*}}\chi(zh_{2}(W))\pi(z)-\boldsymbol{1}[g_{3}\neq 0]\\
&=\frac{G(\pi,\chi)}{q}\sum\limits_{W\in \mathbb{F}_{q}^{*}\backslash (\{g_{3}\}\cup\mathcal{Z}_{0})}\pi(Wh_{2}(W))\sum\limits_{z\in \mathbb{F}_{q}^{*}}\chi(zh_{2}(W))\pi(zh_{2}(W))-\boldsymbol{1}[g_{3}\neq 0]\\
&=\frac{G(\pi,\chi)^2}{q}\sum\limits_{W\in \mathbb{F}_{q}^{*}\backslash (\{g_{3}\}\cup\mathcal{Z}_{0})}\pi(Wh_{2}(W))-\boldsymbol{1}[g_{3}\neq 0]\\
&=\frac{G(\pi,\chi)^2}{q}\left(\sum\limits_{W\in \mathbb{F}_{q}^{*}}\pi(Wh_{2}(W))
- \pi(g_{3}h_{2}(g_{3}))\right)-\boldsymbol{1}[g_{3}\neq 0]
\end{aligned}
\end{equation*}

If $Wh_{2}(W)=g(W)^2$ for some nonzero polynomial $g(W) \in \mathbb{F}_q[W]$, then $g(0)=0$, write $g(W)=W\cdot h_{3}(W)$ for some nonzero polynomial $h_{3}(W) \in \mathbb{F}_q[W]$. Then $h_{2}(W)=Wh_{3}(W)^2$, which implies that  $ h_{2}(0)=\frac{(\beta_{0}g_{4}+g_{3})g_{3}^2}{g^2_{4}}-\beta_{0}g_{2}-g_{1}+a_{r}=0$, a contradiction.  Thus we have that $Wh_{2}(W)$ can not be a square of some polynomial in $\mathbb{F}_q[W]$. By Proposition \ref{prop3},	 it has 
\[\left|\sum_{W \in \mathbb{F}^*_q}\pi(Wh_{2}(W))\right|\leq 3\sqrt{q}.\]
Therefore, 
\begin{equation*}
\begin{aligned}
N(F(X,Y))&=q+\frac{1}{q}\sum\limits_{X\in \mathbb{F}_{q}}\sum\limits_{z\in \mathbb{F}_{q}^{*}}\chi(zF(X,0))+\frac{G(\pi,\chi)^2}{q}\left(\sum\limits_{W\in \mathbb{F}_{q}^{*}}\pi(Wh_{2}(W))
- \pi(g_{3}h_{2}(g_{3}))\right)-\boldsymbol{1}[g_{3}\neq 0]\\
&\geq q-3-3\sqrt{q}.
\end{aligned}
\end{equation*} 
Since $r \leq \frac{q-3\sqrt{q}-3}{4}$, we can deduce that ${\rm N}(F(X,Y)) \geq q-3-3\sqrt{q} \geq 4r$. The proof is completed.

\end{proof}

\begin{lemma}\label{Lem:4.10}
Suppose $q$ is odd and $3\leq r\leq\frac{q}{4}$. For any $b\in \mathbb{F}_{q}$, there exist pairwise distinct $x_{1},\cdots,x_{r}\in \mathbb{F}_{q}^{*}$ such that \begin{equation*}
	c_{1}-\eta c_{2}+\eta c_{1}^2+b=0.
	\end{equation*}
\end{lemma}

\begin{proof}
	Let $f=-S_{1,r-2}-\eta S_{2,r-2}+\eta S_{1,r-2}^2+b,X=\frac{x_{r-1}+x_{r}}{2},Y=\frac{x_{r-1}-x_{r}}{2}$, we have\begin{equation*}
	c_{1}-\eta c_{2}+\eta c_{1}^2+b=3\eta X^2+\eta Y^2+2(\eta S_{1,r-2}-1)X+f.
	\end{equation*}
	Thus, it is sufficient to show that there exists $x_{1},\cdots,x_{r-2}\in \mathbb{F}_{q}^{*}$ such that the equation
	\begin{equation*}
	F(X,Y)=3\eta X^2+\eta Y^2+2(\eta S_{1,r-2}-1)X+f=0
	\end{equation*}
	has a solution $(X,Y)\in \mathbb{F}_{q}^2$, where $x_{1},\cdots,x_{r-2},X+Y,X-Y\in\mathbb{F}_{q}^{*}$ are all distinct.
	
	Note that $N(F(X,X-x_{i}))\leq 2,N(F(X,x_{i}-X))\leq 2$ for all $1\leq i\leq r-2$ and $N(F(X,0))\leq 2,N(F(X,X))\leq 2,N(F(X,-X))\leq 2$. So we only need to show that there exists pairwise distinct $x_{1},\cdots,x_{r-2}\in \mathbb{F}_{q}^{*}$ such that $N(F(X,Y))>4r-2$.
	
	The case $Char(\mathbb{F}_{q})=3$. Since $q>r-1$, we choose $x_{r-2}\in \mathbb{F}_{q}^{*}\setminus\left\{x_{1},\cdots,x_{r-3},\eta^{-1}-S_{1,r-3}\right\}$ such that $x_{1},x_{2},\cdots,x_{r-2}\in\mathbb{F}_{q}^{*}$ are pairwise distinct and $\eta S_{1,r-2}-1\neq 0$. Thus, we have $$N(F(X,Y))=q\geq 4r>4r-2.$$ 
    
    The case $Char(\mathbb{F}_{q})>3$. Since $q>r-1$, we choose $x_{r-2}\in \mathbb{F}_{q}^{*}\setminus\left\{x_{1},\cdots,x_{r-3},\eta^{-1}-S_{1,r-3}\right\}$ such that $x_{1},\cdots,x_{r-2}\in \mathbb{F}_{q}^{*}$ are pairwise distinct and $f-\frac{(\eta S_{1,r-2}-1)^2}{3\eta}\neq 0$. Thus, by Proposition~\ref{prop4}, the equation
	\begin{equation*}
	F(X,Y)=3\eta\left(X+\frac{\eta S_{1,r-2}-1}{3\eta}\right)^2+\eta Y^2+f-\frac{(\eta S_{1,r-2}-1)^2}{3\eta}
	\end{equation*}
	satisfies $N(F(X,Y))\geq q-1\geq 4r-1>4r-2$.
\end{proof}

\begin{lemma}\label{Lem:4.11}
    Suppose $q\geq 7$ is odd, $r=3$ and $\boldsymbol{a}=(a_{0},a_{1},a_{2},a_{3})\in\mathbb{F}_{q}^4$, where $a_{0}=0,a_{1}=2b\eta,a_{2}=b,a_{3}=\frac{b}{4\eta}$ and $b\neq 0$. Then there exists pairwise distinct $x_{1},x_{2},x_{3}\in\mathbb{F}_{q}^{*}$ such that 
    \begin{equation*}
   	\sum\limits_{j=0}^{2}a_{j}c_{3-j}-\eta\sum\limits_{j=0}^{2}a_{j}\sum\limits_{\max\{0,j-1\}\leq w\leq j}c_{3-w}\Lambda_{1+w-j}^{\prime}+a_{3}=0.
   	\end{equation*}
\end{lemma}
\begin{proof}
    For any $2$-subset $\left\{x_{1},x_{2}\right\}\subseteq \mathbb{F}_{q}^{*}$, let $x_{3}=-x_{1}-x_{2},X=x_{1}+x_{2},Y=x_{1}-x_{2}$. Thus,
     We know that \begin{equation*}\begin{aligned}         &\sum\limits_{j=0}^{2}a_{j}c_{3-j}-\eta\sum\limits_{j=0}^{2}a_{j}\sum\limits_{\max\{0,j-1\}\leq w\leq j}c_{3-w}\Lambda_{1+w-j}^{\prime}+a_{3}\\         =&2b\eta c_{2}+bc_{1}-2b\eta^2(c_{3}-c_{1}c_{2})-b\eta(c_{2}-c_{1}^2)+\frac{b}{4\eta}\\
    =&\frac{b\eta}{4}\left(4c_{2}+4\eta^{-1}c_{1}-8\eta(c_{3}-c_{1}c_{2})+4c_{1}^2+\frac{1}{\eta^2}\right)\\
    =&\frac{b\eta}{4}\left(4x_{1}x_{2}-4(x_{1}+x_{2})^2-8\eta x_{1}x_{2}(x_{1}+x_{2})+\frac{1}{\eta^2}\right)\\
    =&\frac{b\eta}{4}(X^2-Y^2-4X^2-2\eta X(X^2-Y^2)+\frac{1}{\eta^2})\\
    =&\frac{b\eta(2\eta X-1)}{4}\left(Y^2-(X+\eta^{-1})\right)^2.
    \end{aligned}\end{equation*}

    Thus, we can choose $X=\frac{1}{2\eta}, Y=\frac{1}{2\eta}+2a$ such that $x_{1}=\frac{1}{2\eta}+a,x_{2}=-a,x_{3}=-\frac{1}{2\eta}\in\mathbb{F}_{q}^{*}$ are pairwise distinct and $\sum\limits_{j=0}^{2}a_{j}c_{3-j}-\eta\sum\limits_{j=0}^{2}a_{j}\sum\limits_{\max\{0,j-1\}\leq w\leq j}c_{3-w}\Lambda_{1+w-j}^{\prime}+a_{3}=0$, where $a\in\mathbb{F}_{q}^{*}\backslash \{-\frac{1}{2\eta},-\frac{1}{4\eta},-\eta^{-1},\frac{1}{2\eta}\}$.
    
\end{proof}

Now we present the main result for deep holes of $TRS_{k}(\mathbb{F}_{q}^{*},k-1,\eta)$ in the odd $q$ case.
\begin{theorem}
	Suppose $q\geq 7$ is odd and $\frac{3q-5+3\sqrt{q}}{4}\leq k\leq q-5$. Then Corollary~\ref{Lem:3.7} provides all deep holes of $TRS_{k}(\mathbb{F}_{q}^{*},k-1,\eta)$.
\end{theorem}
	   \begin{proof}
	   	Denote $r=q-k-2$, then $3\leq r\leq\frac{q-3\sqrt{q}-3}{4}$. Suppose	$H\cdot \boldsymbol{u}^{T}=\boldsymbol{a}=(a_{0},\cdots,a_{r})\in \mathbb{F}_{q}^{r+1}$. Let 
	   	\begin{equation*}
	   	\begin{array}{l}
	   	T_{1}=\left\{(0,2b\eta,b,\frac{b}{4\eta})\in\mathbb{F}_{q}^{4}:b\neq 0\right\},\\
	   	T_{2}=\left\{(0,\cdots,0,a_{r-1},a_{r})\in \mathbb{F}_{q}^{r+1}:a_{r-1},a_{r}\in\mathbb{F}_{q}\right\},\\
	   	T_{3}=\left\{(a_{0},0,\cdots,0)\in \mathbb{F}_{q}^{r+1}:a_{0}\neq 0
	   	\right\},
	   	\end{array}
	   	\end{equation*}
        and
        \begin{equation*}
            T_{4}=\left\{(a_{0},a_{1},\frac{a_{1}^{2}}{a_{0}},\cdots,\frac{a_{1}^{r-1}}{a_{0}^{r-2}},
            \frac{a_{1}^{r}}{a_{0}^{r-1}}-\eta\frac{a_{1}^{r+1}}{a_{0}^r})\in\mathbb{F}_{q}^{*}:a_{0},a_{1}\neq 0\right\}.
        \end{equation*}
	   	
        If $\boldsymbol{a}=(0,\cdots,0,a_{r})$, where $a_{r}\in \mathbb{F}_{q}^{*}$, then by Corollary~\ref{Lem:3.7}, $\boldsymbol{u}$ is a deep hole of $TRS_{k}(\mathbb{F}_{q}^{*},k-1,\eta)$. 
        
        If $\boldsymbol{a}\in T_{1}$, then by Lemma~\ref{Lem:4.11}, Proposition~\ref{Prop:3.2} and Theorem~\ref{The:8}, $\boldsymbol{u}$ is not a deep hole of $TRS_{k}(\mathbb{F}_{q}^{*},k-1,\eta)$. 
        
        If $\boldsymbol{a}\in T_{2}\backslash\left\{(0,\cdots,0,b)\in\mathbb{F}_{q}^{r+1}:b\neq 0\right\}$, then $\boldsymbol{a}=\boldsymbol{0}$ or $\boldsymbol{a}=(0,\cdots,0,a_{r-1},a_{r})$, where $a_{r-1}\neq 0$. For the former, $\boldsymbol{a}$ is not a  deep hole syndrome of $TRS_{k}(\mathbb{F}_{q}^{*},k-1,\eta)$. For the latter, from Lemma~\ref{Lem:4.10}, Since $3\leq r\leq \frac{q-3\sqrt{q}-3}{4}<\frac{q}{4}$, there exists $x_{1},\cdots,x_{r}\in\mathbb{F}_{q}^{*}$ such that $c_{1}-\eta c_{2}+\eta c_{1}^2+a_{r}a_{r-1}^{-1}=0$. Thus,
        \begin{equation*}
            \begin{aligned}
              \sum\limits_{j=0}^{r-1}a_{j}c_{r-j}-\eta\sum\limits_{j=0}^{r-1}a_{j}\sum\limits_{\max\{0,j-1\}\leq w\leq j}c_{r-w}\Lambda_{1+w-j}^{\prime}+a_{r}=a_{r-1}\left(c_{1}-\eta c_{2}+\eta c_{1}^2+a_{r}a_{r-1}^{-1}\right)=0.  
            \end{aligned}
        \end{equation*}
        Therefore, from Proposition~\ref{Prop:3.2} and Theorem~\ref{The:8}, $\boldsymbol{u}$ is not a deep hole of $TRS_{k}(\mathbb{F}_{q}^{*},k-1,\eta)$.

        If $\boldsymbol{a}\in T_{3}$, then $\sum\limits_{j=0}^{r-1}a_{j}c_{r-j}-\eta\sum\limits_{j=0}^{r-1}a_{j}\sum\limits_{\max\{0,j-1\}\leq w\leq j}c_{r-w}\Lambda_{1+w-j}^{\prime}+a_{r}=a_{0}c_{r}(1+\eta c_{1}).$ From Lemma~\ref{Lem:4.7}, there exists $r$-subset $\left\{x_{1},\cdots,x_{r}\right\}\subseteq \mathbb{F}_{q}^{*}$ such that $\sum\limits_{i=1}^{r}x_{i}=\eta^{-1}$. Thus, $$\sum\limits_{j=0}^{r-1}a_{j}c_{r-j}-\eta\sum\limits_{j=0}^{r-1}a_{j}\sum\limits_{\max\{0,j-1\}\leq w\leq j}c_{r-w}\Lambda_{1+w-j}^{\prime}+a_{r}=0.$$
        Therefore, from Proposition~\ref{Prop:3.2} and Theorem~\ref{The:8}, $\boldsymbol{u}$ is not a deep hole of $TRS_{k}(\mathbb{F}_{q}^{*},k-1,\eta)$.

        If $\boldsymbol{a}\in T_{4}$, from Lemma~\ref{lem:4.7}, there exists $r$-subset $\left\{x_{1},\cdots,x_{r}\right\}\subseteq \mathbb{F}_{q}^{*}$ such that $$\sum\limits_{j=0}^{r-1}a_{j}c_{r-j}-\eta\sum\limits_{j=0}^{r-1}a_{j}\sum\limits_{\max\{0,j-1\}\leq w\leq j}c_{r-w}\Lambda_{1+w-j}^{\prime}+a_{r}=0.$$
        Therefore, from Proposition~\ref{Prop:3.2} and Theorem~\ref{The:8}, $\boldsymbol{u}$ is not a deep hole of $TRS_{k}(\mathbb{F}_{q}^{*},k-1,\eta)$.

        If $\boldsymbol{a}\in\mathbb{F}_{q}^{*}\backslash (\bigcup\limits_{i=1}^{4}T_{i})$, from Lemma~\ref{Lem:4.9}, there exists $\left\{x_{1},\cdots,x_{r}\right\}\subseteq \mathbb{F}_{q}^{*}$  such that 
   	\begin{equation*}
   	\sum\limits_{j=0}^{r-1}a_{j}c_{r-j}-\eta\sum\limits_{j=0}^{r-1}a_{j}\sum\limits_{\max\{0,j-1\}\leq w\leq j}c_{r-w}\Lambda_{1+w-j}^{\prime}+a_{r}=0.
   	\end{equation*}
    Therefore, from Proposition~\ref{Prop:3.2} and Theorem~\ref{The:8}, $\boldsymbol{u}$ is not a deep hole of $TRS_{k}(\mathbb{F}_{q}^{*},k-1,\eta)$.

    In summary, Corollary~\ref{Lem:3.7} provides all deep holes of $TRS_{k}(\mathbb{F}_{q}^{*},k-1,\eta)$.
	   \end{proof}

       \section{Conclusion}
       As a generalization of Reed-Solomon codes, TRS codes has received great attention from scholars in recent years. In this paper, we mainly study the deep hole problem of TRS codes. Firstly, for a general evaluation set $\mathcal{A}\subseteq\mathbb{F}_{q}$ and $0\leq l\leq k-1$, we give a sufficient and necessary condition that the vector $\boldsymbol{u}\in\mathbb{F}_{q}^{n-k}$ is a deep hole syndrome of the $TRS_{k}(\mathcal{A},l,\eta)$. Next, for special evaluation set $\mathcal{A}=\mathbb{F}_{q}^*$ and $l= k-1$ we prove that there are no other deep holes of $TRS_{k}(\mathbb{F}_{q}^{*},k-1,\eta)$ for $\frac{3q+2\sqrt{q}-8}{4}\leq k\leq q-5$ when q is even; and for  $\frac{3q+3\sqrt{q}-5}{4}\leq k\leq q-5$ when q is odd. Finally, we completely determine their deep holes for $q-4\leq k\leq q-2$ when $q$ is even.

       Lastly, we leave some directions for the future research:
       \begin{itemize}
           \item Find more deep holes of $TRS_k(\mathcal{A},l,\eta)$ for a general evaluation set $\mathcal{A}$ and general $l$. In this paper, we just give a   sufficient and necessary condition that the vector $\boldsymbol{u}\in\mathbb{F}_{q}^{n-k}$ to be a deep hole syndrome in  $TRS_k(\mathcal{A},l,\eta)$. However, based on this sufficient and necessary condition, it is still difficult to provide explicit expressions for deep holes in $TRS_{k}(\mathcal{A},l,\eta)$. It is not hard to see that this problem will become more difficult as the evaluation set $\mathcal{A}$ becomes smaller, and we predict that there will be more deep holes for small evaluation set $\mathcal{A}$.

           \item 
          Study the covering radius problem and deep hole problem of TRS codes with more twists. For example, the authors~\cite{gu2023twisted} considered the following class of TRS codes. Let $1\leq\ell<\min\{k,n-k\},\boldsymbol{\eta}=(\eta_{1},\cdots,\eta_{\ell})\in(\mathbb{F}_{q}^{*})^{\ell}$, pairwise distinct $\alpha_{1},\cdots,\alpha_{n}\in\mathbb{F}_{q}$, and \begin{equation*}
               \mathcal{S}=\left\{\sum\limits_{i=0}^{k-1}f_{i}x^i+\sum\limits_{i=0}^{\ell-1}\eta_{i+1}f_{k-\ell+i}x^{k+i}:for\ all\ 0\leq i\leq k-1, f_{i}\in\mathbb{F}_{q}\right\}.
           \end{equation*}
           Then $\mathcal{C}=\left\{(f(\alpha_{1}),\cdots,f(\alpha_{n}):f(x)\in\mathcal{S}\right\}$ is a class of TRS code with $\ell$ twists. It is not hard to prove that the covering radius of TRS code $\mathcal{C}$ ranges from $n-k-\ell+1$ to $n-k$.
       \end{itemize}

\appendices

\section{Proof of Lemma~\ref{Lem:4.8}}

Before proving Lemma~\ref{Lem:4.8}, we first give some lemmas. 
\begin{lemma}\label{Lem:C.1}
	Suppose $q$ is odd and $1\leq j\leq i-1\leq q-5$, then there exists $\left\{x_{1},\cdots,x_{i}\right\}\subseteq \mathbb{F}_{q}^{*}$ such that 
	\begin{equation*}
	S_{j,i}\neq 0\quad \mbox{and} \quad S_{1,i} S_{j,i}-S_{j+1,i}\neq 0.
	\end{equation*}
\end{lemma}
\begin{proof}
Let $x_{1},\cdots,x_{i-j}\in\mathbb{F}_{q}^{*}$ are pairwise distinct. For all $1\leq\ell\leq j-1$,let 
$$x_{i-j+\ell}\in\mathbb{F}_{q}^{*}\backslash\{x_{1},\cdots,x_{i-j+\ell-1},-\frac{S_{\ell,i-j+\ell-1}}{S_{\ell-1,i-j+\ell-1}}\}$$
such that $x_{1},\cdots,x_{i-j+\ell}\in\mathbb{F}_{q}^{*}$ are pairwise distinct and 
$$S_{\ell,i-j+\ell}=S_{\ell,i-j+\ell-1}+S_{\ell-1,i-j+\ell-1}\cdot x_{i-j+\ell}\neq 0.$$
Thus, there exists $i-1$-subset $\left\{x_{1},\cdots,x_{i-1}\right\}\subseteq\mathbb{F}_{q}^{*}$ such that $S_{j-1,i-1}\neq 0$. In addition, 
\begin{equation*}
\begin{aligned}
   S_{1,i}S_{j,i}-S_{j+1,i}&=(S_{1,i-1}+x_{i})\cdot (S_{j,i-1}+S_{j-1,i-1}x_{i})-(S_{j+1,i-1}+S_{j,i-1}x_{i}) \\
   &=S_{j-1,i-1}x_{i}^2+S_{1,i-1}\cdot S_{j-1,i-1}x_{i}+S_{1,i-1}\cdot S_{j,i-1}-S_{j+1,i-1}.
\end{aligned}
\end{equation*}
Let 
     \begin{equation*}
     \begin{aligned}
     A=&\left\{\alpha\in \mathbb{F}_{q}^{*}:
     S_{j-1,i-1}\alpha^2+S_{1,i-1}S_{j-1,i-1}\alpha+S_{1,i-1}S_{j,i-1}-S_{j+1,i-1}
     \right\}.
     \end{aligned}
     \end{equation*}
     Since $i\leq q-4$, 
      there exists $x_{i}\in \mathbb{F}_{q}^{*}\backslash \left(\{{x}_{1},\cdots,{x}_{i-1},-\frac{S_{j,i-1}}{S_{j-1,i-1}}\}\cup A\right)$ such that
     $S_{j,i}=S_{j,i-1}+S_{j-1,i-1}{x}_{i}\neq 0$
     and $S_{1,i}\cdot S_{j,i}-S_{j+1,i}\neq 0$.

\end{proof}

\begin{lemma}\label{Lem:C.2}
	Suppose $q$ is odd and $1\leq j\leq i-1\leq q-5$, then there exists $\left\{x_{1},\cdots,x_{i}\right\}\subseteq \mathbb{F}_{q}^{*}$ such that $S_{j-1,i}\neq 0$ and $S_{j,i}^2-S_{j-1,i}\cdot S_{j+1,i}\neq 0.$
\end{lemma}

\begin{proof}
For any $i-j$-subset $\left\{x_{1},\cdots,x_{i-j}\right\}\subseteq \mathbb{F}_{q}^{*}$, let $f_{1}(X)=X^2+S_{1,i-j}\cdot X+S_{1,i-j}^2-S_{2,i-j}$ and $B_{1}$ denote the root in $\mathbb{F}_{q}$ of $f_1(X)$. Since $i-j\leq i-1\leq  q-5$, we can choose $x_{i-j+1}\in \mathbb{F}_{q}^{*}\backslash\left(\{x_{1},\cdots,x_{i-j}\}\cup B_{1}\right)$ such that $x_{1},\cdots,x_{i-j+1}$ are pairwise distinct  and
\begin{equation*}
       \begin{aligned}
              S_{1,i-j+1}^2-S_{0,i-j+1}S_{2,i-j+1}&=(S_{1,i-j}+x_{i-j+1})^2-(S_{2,i-j}+S_{1,i-j}x_{i-j+1})\\
              &=x_{i-j+1}^2+S_{1,i-j}x_{i-j+1}+S_{1,i-j}^2-S_{2,i-j}=f_{1}(x_{i-j+1})\neq 0.
       \end{aligned}
	   \end{equation*}
	    For all $2\leq\ell\leq j$, let 
        \begin{equation*}
            \begin{aligned}
                f_{\ell}(X)&=(S_{\ell-1,i-j+\ell-1}^2-S_{\ell-2,i-j+\ell-1}\cdot S_{\ell,i-j+\ell-1})X^2+(S^2_{\ell,i-j+\ell-1}-S_{\ell-1,i-j+\ell-1}\cdot S_{\ell+1,i-j+\ell-1})\\
                &+(S_{\ell,i-j+\ell-1}\cdot S_{\ell-1,i-j+\ell-1}-S_{\ell-2,i-j+\ell-1}\cdot S_{\ell+1,i-j+\ell-1})X
            \end{aligned}
        \end{equation*}
        and $B_{\ell}$
	   denote the root in  $\mathbb{F}_{q}$ of $f_{\ell}(X)$. Since $S_{\ell-1,i-j+\ell-1}^2-S_{\ell-2,i-j+\ell-1}\cdot S_{\ell,i-j+\ell-1},B_{\ell-2,i-j+\ell-1}\neq 0$ and $i-j+\ell\leq i\leq q-4$, we can choose $x_{i-j+\ell}\in \mathbb{F}_{q}^{*}\backslash\left(
	   \left\{ x_{1},\cdots,x_{i-j+\ell-1},-\frac{S_{\ell-1,i-j+\ell-1}}{S_{\ell-2,i-j+\ell-1}}\right\}\cup B_{\ell}\right)$ such that
	   $x_{1},\cdots,x_{i-j+\ell}\in \mathbb{F}_{q}^{*}$ are pairwise distinct, 
       $S_{\ell-1,i-j+\ell}=S_{\ell-1,i-j+\ell-1}+S_{\ell-2,i-j+\ell-1}x_{i-j+\ell}\neq 0$ and
	   \begin{equation*}
       \begin{aligned}
            &S_{\ell,i-j+\ell}^2-S_{\ell-1,i-j+\ell}\cdot S_{\ell+1,i-j+\ell}\\
            =&(S_{\ell,i-j+\ell-1}+S_{\ell-1,i-j+\ell-1}x_{i-j+\ell})^2-(S_{\ell-1,i-j+\ell-1}+S_{\ell-2,i-j+\ell-1}x_{i-j+\ell})\cdot (S_{\ell+1,i-j+\ell-1}+S_{\ell,i-j+\ell-1}x_{i-j+\ell})\\
            =&(S_{\ell-1,i-j+\ell-1}^2-S_{\ell-2,i-j+\ell-1}\cdot S_{\ell,i-j+\ell-1})x_{i-j+\ell}^2+(S_{\ell,i-j+\ell-1}\cdot S_{\ell-1,i-j+\ell-1}-S_{\ell-2,i-j+\ell-1}\cdot S_{\ell+1,i-j+\ell-1})x_{i-j+\ell}\\
                +&(S_{\ell,i-j+\ell-1}^2-S_{\ell-1,i-j+\ell-1}\cdot S_{\ell+1,i-j+\ell-1})=f_{\ell}(x_{i-j+\ell})\neq 0.
       \end{aligned}.
	   \end{equation*}
	 Hence, there exists $\left\{x_{1},\cdots,x_{i}\right\}\subseteq \mathbb{F}_{q}^{*}$ such that $S_{j-1,i}\neq 0$ and $S_{j,i}^2-S_{j-1,i}\cdot S_{j+1,i}\neq 0.$

\end{proof}

\begin{proof}[Proof of Lemma~\ref{Lem:4.8}]
		We know that the coefficient of $x^{3r-5}$ in $\tilde{g}(x)$ is $$(-1)^{r-2}\eta^3(a_{0}a_{1}^2-a_{0}^2a_{2})S_{1,r-2}S_{r-2,r-2}^3$$ and the coefficient of $x^{3r-6}$ in $\tilde{g}(x)$ is 
        \begin{equation*}
        \begin{aligned}
           &(-1)^{r-2}\eta^3\left(a_{0}^2a_{3}-a_{1}^3\right)S_{1,r-2}S_{r-3,r-2}S_{r-2,r-2}^2+(-1)^{r-2}\eta^2\left(a_{0}^2a_{2}-a_{0}a_{1}^2\right)S_{r-2,r-2}^3\\
           &~~~~~~~~~~~~~~~~~~~~~~~~~~~~~~~~~~~~~~~~~~~~~~~~~~~~~~~~~~~~~+(-1)^{r-2}\eta^3\left(a_{1}^3-a_{0}^2a_{3}\right)S_{r-2,r-2}^3. 
        \end{aligned}
        \end{equation*}
	We will consider the following situations:
	
	\begin{itemize}
		\item If $a_{0}=0,a_{1}\neq 0$, then  the coefficient of $x^{r-3}$ in $\tilde{g}_{4}(x)$ is $(-1)^{r-3}\eta a_{1}S_{r-3,r-2}$ and the coefficient of
		$x^{3r-6}$ in $\tilde{g}(x)$ is $$(-1)^{r-3}a_{1}^3\eta^3\left(S_{1,r-2}S_{r-3,r-2}S_{r-2,r-2}^2-S_{r-2,r-2}^3\right)=(-1)^{r-3}a_{1}^3\eta^3S_{r-2,r-2}^2\left(S_{1,r-2}S_{r-3,r-2}-S_{r-2,r-2}\right).
		$$
		If $r\geq 4$,  from Lemma~\ref{Lem:C.1}, since $4\leq r\leq q-4<q-2$, there exists $r-2$-subset $\left\{x_{1},\cdots,x_{r-2}\right\}\subseteq \mathbb{F}_{q}^{*}$ such that $S_{r-3,r-2}\neq 0$ and $S_{1,r-2}S_{r-3,r-2}-S_{r-2,r-2}\neq 0$. Thus, $\tilde{g}_{4}(x),\tilde{g}(x)\neq 0$.

        If $r=3$, then $\tilde{g}_{4}(x)=a_{1}\eta\neq 0$ and 
       \begin{equation*} \begin{small}
            \begin{aligned}
                &\tilde{g}(X)=(S_{1,1}X-\eta^{-1})\tilde{g}_{4}(X)\tilde{g}_{3}(X)^2+\tilde{g}_{3}(X)^3-(S_{1,1}X-\eta^{-1})\tilde{g}_{2}(X)\tilde{g}_{4}(X)^2-\tilde{g}_{1}(X)\tilde{g}_{4}(X)^2+a_{r}\tilde{g}_{4}(X)^2\\
                &=a_{1}\eta(S_{1,1}X-\eta^{-1})(-a_{1}\eta S_{1,1}X+a_{2}\eta)^2+(-a_{1}\eta S_{1,1}X+a_{2}\eta)^3+a_{1}^2a_{2}\eta^3S_{1,1}X(S_{1,1}X-\eta^{-1})+a_{1}^2a_{3}\eta^2\\
                &=a_{1}^2\eta^2S_{1,1}^2(2a_{2}\eta-a_{1})X^2+a_{1}a_{2}\eta^2S_{1,1}(a_{1}-2a_{2}\eta)X-a_{1}a_{2}^2\eta^2+a_{2}^3\eta^3+a_{1}^2a_{3}\eta^2\\
                &=a_{1}^2\eta^2S_{1,1}^2(2a_{2}\eta-a_{1})X^2+a_{1}a_{2}\eta^2S_{1,1}(a_{1}-2a_{2}\eta)X-(a_{1}-2a_{2}\eta)a_{2}^2\eta^2+(a_{1}^2-4a_{2}^2\eta^2)a_{3}\eta^2+(4a_{3}\eta-a_{2})a_{2}^2\eta^3.
            \end{aligned}  \end{small}
        \end{equation*}
        Because of $\boldsymbol{a}=(a_{0},a_{1},a_{2},a_{3})\notin\left\{(0,2b\eta,b,\frac{b}{4\eta})\in\mathbb{F}_{q}^4:b\neq 0\right\}$, we have $\tilde{g}(x)\neq 0$.

	
	\item If $a_{0}=a_{1}= 0,a_{2}\neq 0$. Because of $$\boldsymbol{a}=(a_{0},\cdots,a_{r})\notin \left\{(u_{0},\cdots,u_{r})\in\mathbb{F}_{q}^{r+1}:u_{0}=\cdots=u_{r-2}=0\ or\ u_{0}\neq 0,u_{1}=\cdots=u_{r}=0\right\},$$ 
    we have $r\geq 4$. 
	Notice that the coefficient of $x^{r-4}$ in $\tilde{g}_{4}(x)$ is $(-1)^{r-4}a_{2}\eta S_{r-4,r-2}$ and the coefficient of $x^{3r-9}$ in $\tilde{g}(x)$ is $$(-1)^{r-4}a_{2}^3\eta^3\left(S_{1,r-2}S_{r-4,r-2}S_{r-3,r-2}^2-S_{r-3,r-2}^3-S_{1,r-2}S_{r-4,r-2}^2S_{r-2,r-2}
	\right).$$
    
    If $r=4$, then the constant term in $\tilde{g}_{4}(x)$ is $a_{2}\eta\neq 0$ and the coefficient of $x^{3}$ in $\tilde{g}(x)$ is $-a_{2}^3\eta^3 S_{1,2}S_{2,2}$. Thus, we just choose $x_{1}\in \mathbb{F}_{q}^{*}$ and $x_{2}\in\mathbb{F}_{q}^{*}\backslash\{-x_{1}\}$ such that $-a_{2}^3\eta^3S_{1,2}S_{2,2}\neq 0$. 
    
    If $r\geq 5$, from Lemma~\ref{Lem:C.2}, since $1\leq r-4\leq q-8<q-5$, there exists $\left\{x_{1},\cdots,x_{r-3}\right\}\subseteq \mathbb{F}_{q}^{*}$ such that $S_{r-5,r-3}\neq 0$ and $S_{r-4,r-3}^2-S_{r-5,r-3}\cdot S_{r-3,r-3}\neq 0.$
	Let
	\begin{equation*}
	\begin{aligned}
	f(X)&=(S_{1,r-3}+X)(S_{r-4,r-3}+S_{r-5,r-3}X)(S_{r-3,r-3}+S_{r-4,r-3}X)^2-(S_{r-3,r-3}+S_{r-4,r-3}X)^3\\
	&-(S_{1,r-3}+X)(S_{r-4,r-3}+S_{r-5,r-3}X)^2S_{r-3,r-3}X\\
	\end{aligned}
	\end{equation*} 
	and $C_{1}$ denote the root in the $\mathbb{F}_{q}$ in $f(X)$.
	Because the coefficient of $X^4$ in $f(X)$ is $S_{r-5,r-3}(S_{r-4,r-3}^2-S_{r-5,r-3}\cdot S_{r-3,r-3})\neq 0$, $S_{r-5,r-3}\neq 0$ and $5\leq r\leq q-4$,
	there exists $$x_{r-2}\in \mathbb{F}_{q}^{*}\backslash\left(
	\left\{x_{1},\cdots,x_{r-3},-\frac{S_{r-4,r-3}}{S_{r-5,r-3}}\right\}\cup C_{1}\right)$$ such that
	$x_{1},\cdots,x_{r-2}$ are pairwise distinct,
	$S_{r-4,r-2}=S_{r-4,r-3}+S_{r-5,r-3}x_{r-2}\neq 0$ and $$S_{1,r-2}S_{r-4,r-2}S_{r-3,r-2}^2-S_{r-3,r-2}^3-S_{1,r-2}S_{r-4,r-2}^2S_{r-2,r-2}=f(x_{r-2})\neq 0.$$

		\item If $a_{0}=a_{1}=a_{2}= 0$ and there exists $3\leq t\leq r-2$ such that $a_{t}\neq 0$.  Let $U_{1}=\min\left\{3\leq i\leq r-2:a_i\neq 0\right\}$. Then $r\geq U_{1}+2\geq 5$.
		We know that the coefficient of $x^{r-U_{1}-2}$ in $\tilde{g}_{4}(x)$ is $(-1)^{r-U_{1}-2}a_{U_{1}}\eta S_{r-U_{1}-2,r-2}$ and the coefficient of $x^{3r-3U_{1}-3}$ in $\tilde{g}(x)$ is 
        \begin{equation*}
            \begin{aligned}
                &(-1)^{r-U_{1}-2}a_{U_{1}}^3\eta^3\left(S_{1,r-2}S_{r-U_{1}-2,r-2}S_{r-U_{1}-1,r-2}^2-S_{r-U_{1}-1,r-2}^3\right.\\
                &~~~~~~~~~~~~~~~~~~~~~~~~~\left.-S_{1,r-2}S_{r-U_{1}-2,r-2}^2S_{r-U_{1},r-2}+S_{r-U_{1}-2,r-2}^2S_{r-U_{1}+1,r-2}
		\right).
            \end{aligned}
        \end{equation*}
        If $r=U_{1}+2$, then the constant term in $\tilde{g}_{4}(x)$ is $a_{U_{1}}\eta\neq 0$
         and the coefficient of $x^{3}$ in $\tilde{g}(x)$ is $-a_{U_{1}}^3\eta^3(S_{1,r-2}S_{2,r-2}-S_{3,r-2})$. Since $1\leq 2\leq r-3\leq q-5$, from Lemma~\ref{Lem:C.1}, we have $S_{1,r-2}S_{2,r-2}-S_{3,r-2}\neq 0$. Thus, $\tilde{g}_{4}(x),\tilde{g}(x)\neq 0$.

         If $r\geq U_{1}+3$, from Lemma~\ref{Lem:C.2} and $1\leq r-U_{1}-2\leq r-4\leq q-5$, there exists $\left\{x_{1},\cdots,x_{r-3}\right\}\subseteq \mathbb{F}_{q}^{*}$ such that $S_{r-U_{1}-3,r-3}\neq 0$ and $S_{r-U_{1}-2,r-3}^2-S_{r-U_{1}-3,r-3}\cdot S_{r-U_{1}-1,r-3}\neq 0.$
		Let
		\begin{equation*}
		\begin{aligned}
		&f(X)=(S_{1,r-3}+X)(S_{r-U_{1}-2,r-3}+S_{r-U_{1}-3,r-3}X)(S_{r-U_{1}-1,r-3}+S_{r-U_{1}-2,r-3}X)^2\\
		&-(S_{1,r-3}+X)(S_{r-U_{1}-2,r-3}+S_{r-U_{1}-3,r-3}X)^2(S_{r-U_{1},r-3}+S_{r-U_{1}-1,r-3}X)\\
		&-(S_{r-U_{1}-1,r-3}+S_{r-U_{1}-2,r-3}X)^3+(S_{r-U_{1}-2,r-3}+S_{r-U_{1}-3,r-3}X)^2(S_{r-U_{1}+1,r-3}+S_{r-U_{1},r-3}X),
		\end{aligned}
		\end{equation*} 
		and $C_{2}$ denote the root in the $\mathbb{F}_{q}$ in $f(X)$.
		 Because the coefficient of $x^4$ in $f(x)$ is $$S_{r-U_{1}-3,r-3}\left(S_{r-U_{1}-2,r-3}^2-S_{r-U_{1}-3,r-3}\cdot S_{r-U_{1}-1,r-3}\right)\neq 0,S_{r-U_{1}-3,r-3}\neq 0$$
		and $r\leq q-4$,  there exists $x_{r-2}\in \mathbb{F}_{q}^{*}\backslash\left(
		\left\{x_{1},\cdots,x_{r-3},-\frac{S_{r-U_{1}-2,r-3}}{S_{r-U_{1}-3,r-3}}\right\}\cup C_{2}\right)$ such that
		$x_{1},\cdots,x_{r-2}$ are pairwise distinct,
		$S_{r-U_{1}-2,r-2}=S_{r-U_{1}-2,r-3}+S_{r-U_{1}-3,r-3}x_{r-2}\neq 0$ and
        \begin{equation*}
            \begin{aligned}
              &S_{1,r-2}S_{r-U_{1}-2,r-2}S_{r-U_{1}-1,r-2}^2-S_{r-U_{1}-1,r-2}^3-S_{1,r-2}S_{r-U_{1}-2,r-2}^2S_{r-U_{1},r-2}\\
              &-S_{r-U_{1}-2,r-2}^2S_{r-U_{1}+1,r-2}=f(x_{r-2})\neq 0.  
            \end{aligned}
        \end{equation*}
        Thus, $\tilde{g}_{4}(x),\tilde{g}(x)\neq 0.$

		\item If $a_{0}\neq 0,a_{1}=\cdots=a_{r-1}=0$ and $a_{r}\neq 0$, then $\tilde{g}(x)=a_{r}\tilde{g}_{4}(x)^2=a_{r}a_{0}^2\eta^2S_{r-2,r-2}^2x^{2r-4}\neq 0,\tilde{g}_{4}(x)=\eta a_{0}S_{r-2,r-2}x^{r-2}\neq 0$. 
		\item If $a_{0}\neq 0,a_{1}=a_{2}=0$ and there exists $3\leq i\leq r-1$ such that $a_{i}\neq 0$. Let $U_{2}=\min\left\{3\leq i\leq r-1:a_{i}\neq 0\right\}$, then the coefficients of $x^{r-2}$ in $\tilde{g}_{4}(x)$ is $(-1)^{r-2}a_{0}\eta S_{r-2,r-2}\neq 0$ and the coefficient of $x^{3r-3-U_{2}}$ in $\tilde{g}(x)$ is 
		\begin{equation*}
		\begin{aligned}
		&(-1)^{r-U_{2}+1}a_{0}^2a_{U_{2}}\eta^3S_{1,r-2}S_{r-U_{2},r-2}S_{r-2,r-2}^2+(-1)^{r-U_{2}}a_{0}^2a_{U_{2}}\eta^3S_{r-U_{2}+1,r-2}S_{r-2,r-2}^2\\
		&=(-1)^{r-U_{2}+1}a_{0}^2a_{U_{2}}\eta^3\left(S_{1,r-2}S_{r-U_{2},r-2}-S_{r-U_{2}+1,r-2}\right)S_{r-2,r-2}^2.
		\end{aligned}
		\end{equation*}
		From Lemma~\ref{Lem:C.1}, since $1\leq r-U_{2}\leq r-3\leq q-5$,  there exists $\left\{x_{1},\cdots,x_{r-2}\right\}\subseteq \mathbb{F}_{q}^{*}$ such that  $S_{1,r-2}S_{r-U_{2},r-2}-S_{r-U_{2}+1,r-2}\neq 0$. Thus, $\tilde{g}_{4}(x),\tilde{g}(x)\neq 0$.
		
		\item If $a_{0},a_{2}\neq 0,a_{1}=0$, then the coefficients of $x^{r-2}$ in $\tilde{g}_{4}(x)$ is $(-1)^{r-2}a_{0}\eta S_{r-2,r-2}\neq 0$ and the coefficient of $x^{3r-5}$ in $\tilde{g}(x)$ is $(-1)^{r-1}a_{0}^2a_{2}\eta^3S_{1,r-2}S_{r-2,r-2}^3.$ For any $\left\{x_{1},\cdots,x_{r-3}\right\}\subseteq \mathbb{F}_{q}^{*}$, Since $r\leq q$, we can choose $x_{r-2}\in \mathbb{F}_{q}^{*}\backslash\{x_{1},\cdots,x_{r-3},-\sum\limits_{i=1}^{r-3}x_{i}\}$ such that $S_{1,r-2}\neq 0$. Thus, $\tilde{g}(x),\tilde{g}_{4}(x)\neq 0$.
		
		\item If $a_{0},a_{1}\neq 0$ and there exists $2\leq t\leq r-1$ such that $a_{i}=a_{1}^ia_{0}^{-(i-1)}$ for all $2\leq i\leq t-1$ and $a_{t}\neq a_{1}^ta_{0}^{-(t-1)}$. Then the coefficients of $x^{r-2}$ in $\tilde{g}_{4}(x)$ is $(-1)^{r-2}a_{0}\eta S_{r-2,r-2}\neq 0$ and the coefficient of $x^{3r-3-t}$ in $\tilde{g}(x)$ is
        \begin{equation*}
\begin{aligned}
&S_{1,r-2}\eta^3\sum\limits_{
	0\leq i\leq r-2\atop{
		1\leq j\leq r-1\atop{
			1\leq s\leq r-1\atop{
				i+j+s=t}}}}(-1)^{r-t}a_{i}a_{j}a_{s}S_{r-2-i,r-2}S_{r-j-1,r-2}S_{r-s-1,r-2}\\
-&S_{1,r-2}\eta^3\sum\limits_{
	2\leq i\leq r-1\atop{
		0\leq j\leq r-2\atop{
			0\leq s\leq r-2\atop{
				i+j+s=t}}}}(-1)^{r-t}a_{i}a_{j}a_{s}S_{r-i,r-2}S_{r-j-2,r-2}S_{r-s-2,r-2}\\
&-\eta^2\sum\limits_{
	0\leq i\leq r-2\atop{
		1\leq j\leq r-1\atop{
			1\leq s\leq r-1\atop{
				i+j+s=t-1}}}}(-1)^{r-t-1}a_{i}a_{j}a_{s}S_{r-2-i,r-2}S_{r-j-1,r-2}S_{r-s-1,r-2}\\
+&\eta^2\sum\limits_{
	2\leq i\leq r-1\atop{
		0\leq j\leq r-2\atop{
			0\leq s\leq r-2\atop{
				i+j+s=t-1}}}}(-1)^{r-t+1}a_{i}a_{j}a_{s}S_{r-i,r-2}S_{r-j-2,r-2}S_{r-s-2,r-2}\\
+&\eta^3\sum\limits_{
	1\leq i\leq r-1\atop{
		1\leq j\leq r-1\atop{
			1\leq s\leq r-1\atop{
				i+j+s=t}}}}(-1)^{r-t-1}a_{i}a_{j}a_{s}S_{r-1-i,r-2}S_{r-j-1,r-2}S_{r-s-1,r-2}\\
-&\eta^3\sum\limits_{
	3\leq i\leq r-1\atop{
		0\leq j\leq r-2\atop{
			0\leq s\leq r-2\atop{
				i+j+s=t}}}}(-1)^{r-t+1}a_{i}a_{j}a_{s}S_{r-i+1,r-2}S_{r-j-2,r-2}S_{r-s-2,r-2}.\\
\end{aligned}
\end{equation*}

 If $t=2$, then the coefficient of $x^{3r-5}$ in $\tilde{g}(x)$ is $(-1)^{r-2}a_{0}^2\eta^3\left(\frac{a_{1}^2}{a_{0}}-a_{2}\right)S_{1,r-2}S_{r-2,r-2}^3.$
 On the one hand, we can choose $r-2$-subset $\left\{x_{1},\cdots,x_{r-2}\right\}\subseteq \mathbb{F}_{q}^{*}$ such that $S_{1,r-2}\neq 0$. On the other hand, $a_{2}\neq a_{1}^2a_{0}^{-1}$. Thus, there exists $r-2$-subset $\left\{x_{1},\cdots,x_{r-2}\right\}\subseteq \mathbb{F}_{q}^{*}$ such that $\tilde{g}(x)\neq 0$.
 
 If $3\leq t\leq r-1$, then 
\begin{equation*}
\begin{aligned}
&S_{1,r-2}\eta^3\sum\limits_{
	0\leq i\leq r-2\atop{
		1\leq j\leq r-1\atop{
			1\leq s\leq r-1\atop{
				i+j+s=t}}}}(-1)^{r-t}a_{i}a_{j}a_{s}S_{r-2-i,r-2}S_{r-j-1,r-2}S_{r-s-1,r-2}\\
-&S_{1,r-2}\eta^3\sum\limits_{
	2\leq i\leq r-1\atop{
		0\leq j\leq r-2\atop{
			0\leq s\leq r-2\atop{
				i+j+s=t}}}}(-1)^{r-t}a_{i}a_{j}a_{s}S_{r-i,r-2}S_{r-j-2,r-2}S_{r-s-2,r-2}\\
=&S_{1,r-2}\eta^3\sum\limits_{
	2\leq i\leq t\atop{
		0\leq j\leq t-2\atop{
			0\leq s\leq t-2\atop{
				i+j+s=t}}}}(-1)^{r-t}\left(a_{i-2}a_{j+1}a_{s+1}-a_{i}a_{j}a_{s}\right)S_{r-i,r-2}S_{r-j-2,r-2}S_{r-s-2,r-2}\\
\stackrel{(1)}{=}&S_{1,r-2}\eta^3\sum\limits_{
	1\leq j+s\leq t-2\atop{
		0\leq j\leq t-2\atop{
			0\leq s\leq t-2
				}}}(-1)^{r-t}\left(\frac{a_{1}^t}{a_{0}^{t-3}}-\frac{a_{1}^t}{a_{0}^{t-3}}\right)S_{r-t+j+s,r-2}S_{r-j-2,r-2}S_{r-s-2,r-2}\\
				+&(-1)^{r-t}\eta^3\left(a_{t-2}a_{1}^2-a_{t}a_{0}^2\right)S_{1,r-2}S_{r-t,r-2}S_{r-2,r-2}^2\\
                =&(-1)^{r-t}\eta^3a_{0}^2\left(\frac{a_{1}^{t}}{a_{0}^{t-1}}-a_{t}\right)S_{1,r-2}S_{r-t,r-2}S_{r-2,r-2}^2,
\end{aligned}
\end{equation*}
where $(1)$ follows from $a_{i}=a_{1}^ia_{0}^{-(i-1)}$ for all $3\leq i\leq t-1$. In addition,

\begin{equation*}
\begin{aligned}
-&\eta^2\sum\limits_{
	0\leq i\leq r-2\atop{
		1\leq j\leq r-1\atop{
			1\leq s\leq r-1\atop{
				i+j+s=t-1}}}}(-1)^{r-t-1}a_{i}a_{j}a_{s}S_{r-2-i,r-2}S_{r-j-1,r-2}S_{r-s-1,r-2}\\
+&\eta^2\sum\limits_{
	2\leq i\leq r-1\atop{
		0\leq j\leq r-2\atop{
			0\leq s\leq r-2\atop{
				i+j+s=t-1}}}}(-1)^{r-t+1}a_{i}a_{j}a_{s}S_{r-i,r-2}S_{r-j-2,r-2}S_{r-s-2,r-2}\\
=&-\eta^2\sum\limits_{
	2\leq i\leq t-1\atop{
		0\leq j\leq t-3\atop{
			0\leq s\leq t-3\atop{
				i+j+s=t-1}}}}(-1)^{r-t-1}\left(a_{i-2}a_{j+1}a_{s+1}-a_{i}a_{j}a_{s}\right)S_{r-i,r-2}S_{r-j-2,r-2}S_{r-s-2,r-2}\\
=&-\eta^2\sum\limits_{
	2\leq i\leq t-1\atop{
		0\leq j\leq t-3\atop{
			0\leq s\leq t-3\atop{
				i+j+s=t-1}}}}(-1)^{r-t-1}\left(
                 \frac{a_{1}^{t-1}}{a_{0}^{t-4}}-\frac{a_{1}^{t-1}}{a_{0}^{t-4}}
\right)S_{r-i,r-2}S_{r-j-2,r-2}S_{r-s-2,r-2}=0\\
\end{aligned}
\end{equation*}
and
\begin{equation*}
\begin{aligned}
&\eta^3\sum\limits_{
	1\leq i\leq r-1\atop{
		1\leq j\leq r-1\atop{
			1\leq s\leq r-1\atop{
				i+j+s=t}}}}(-1)^{r-t-1}a_{i}a_{j}a_{s}S_{r-1-i,r-2}S_{r-j-1,r-2}S_{r-s-1,r-2}\\
-&\eta^3\sum\limits_{
	3\leq i\leq r-1\atop{
		0\leq j\leq r-2\atop{
			0\leq s\leq r-2\atop{
				i+j+s=t}}}}(-1)^{r-t+1}a_{i}a_{j}a_{s}S_{r-i+1,r-2}S_{r-j-2,r-2}S_{r-s-2,r-2}\\
=&\eta^3\sum\limits_{
	3\leq i\leq t\atop{
		0\leq j\leq t-3\atop{
			0\leq s\leq t-3\atop{
				i+j+s=t}}}}(-1)^{r-t-1}\left(a_{i-2}a_{j+1}a_{s+1}-a_{i}a_{j}a_{s}\right)S_{r-i+1,r-2}S_{r-j-2,r-2}S_{r-s-2,r-2}\\
=&\eta^3\sum\limits_{
	1\leq j+s\leq t-3\atop{
		0\leq j\leq t-3\atop{
			0\leq s\leq t-3}}}(-1)^{r-t-1}\left(
\frac{a_{1}^t}{a_{0}^{t-3}}-\frac{a_{1}^t}{a_{0}^{t-3}}
\right)S_{r-i+1,r-2}S_{r-j-2,r-2}S_{r-s-2,r-2}\\
+&(-1)^{r-t-1}\eta^3(a_{t-2}a_{1}^2-a_{t}a_{0}^2)S_{r-t+1,r-2}S_{r-2,r-2}^2\\
=&(-1)^{r-t-1}\eta^3a_{0}^2\left(\frac{a_{1}^t}{a_{0}^{t-1}}-a_{t}\right)S_{r-t+1,r-2}S_{r-2,r-2}^2.
\end{aligned}
\end{equation*}

Thus, the coefficient of $x^{3r-3-t}$ in $\tilde{g}(x)$ is
\begin{equation*}
    (-1)^{r-t}\eta^3a_{0}^2\left(\frac{a_{1}^t}{a_{0}^{t-1}}-a_{t}\right)\left(S_{1,r-2}S_{r-t,r-2}-S_{r-t+1,r-2}\right)S_{r-2,r-2}^2
\end{equation*}

		On the one hand, from Lemma~\ref{Lem:C.1}, since $1\leq r-t\leq r-3\leq q-5$, there exists $\left\{x_{1},\cdots,x_{r-2}\right\}\subseteq \mathbb{F}_{q}^{*}$ such that  $S_{1,r-2}S_{r-t,r-2}-S_{r-t+1,r-2}\neq 0$.
		 On the other hand, $a_{t}\neq \frac{a_{1}^t}{a_{0}^{t-1}}$. Thus, there exists $\left\{x_{1},\cdots,x_{r-2}\right\}\subseteq \mathbb{F}_{q}^{*}$ such that $\tilde{g}_{4}(x)\neq 0$.

         In summary, there exists $\left\{x_{1},\cdots,x_{r-2}\right\}\subseteq \mathbb{F}_{q}^{*}$ such that $\tilde{g}(x),\tilde{g}_{4}(x)\neq 0$.
		\item If $a_{0},a_{1}\neq 0,a_j=a_{1}^ja_{0}^{-(j-1)}$ for all $2\leq j\leq r-1$  and $a_{r}\neq a_{1}^ra_{0}^{-(r-1)}-\eta a_{1}^{r+1}a_{0}^{-r}$. Then the coefficients of $x^{r-2}$ in $\tilde{g}_{4}(x)$ is $(-1)^{r-2}a_{0}\eta S_{r-2,r-2}\neq 0$ and the coefficients of constant terms in $\tilde{g}(x)$ is 
		\begin{equation*}
		-\eta^2a_{r-2}a_{r-1}^2+\eta^3a_{r-1}^3+a_{r}\eta^2a_{r-2}^2=\eta^2a_{r-2}^2\left(a_{r}-a_{1}^ra_{0}^{-(r-1)}+\eta a_{1}^{r+1}a_{0}^{-r}\right)\neq 0.
		\end{equation*}
		Thus, $\tilde{g}(x),\tilde{g}_{4}(x)\neq 0$.

	\end{itemize}
	
\end{proof}

\bibliographystyle{plain}
\bibliography{TRS}

\newpage

\end{document}